\documentclass[12pt]{amsart}
\usepackage{amsfonts}
\usepackage{changepage}
\usepackage{natbib}
\usepackage{eurosym}
\usepackage{booktabs}
\usepackage{amssymb, amsmath, url,color,  amsthm, amssymb,graphicx,multirow, appendix}
\usepackage[margin=1in]{geometry}
\usepackage{bigints}
\usepackage{lscape,amssymb, amsmath,verbatim,graphicx}
\usepackage{epsfig,subfigure,float,subfigure}
\usepackage{graphicx}
\usepackage{epsfig}
\usepackage{algorithm}
\usepackage{algpseudocode}
\usepackage{threeparttable}
\usepackage{ragged2e}
\usepackage{epstopdf}
\usepackage{booktabs}
\usepackage{float}
\usepackage{multirow}
\usepackage{adjustbox}
\usepackage{lscape,lipsum}
\usepackage{url}
\usepackage{dsfont}
\usepackage{natbib}
\usepackage[T1]{fontenc}
\usepackage{a4wide, url}
\usepackage[english]{babel}
\usepackage{graphicx}
\usepackage{array}
\usepackage{multirow}
\usepackage{amsthm}
\usepackage{amsmath}
\usepackage[foot]{amsaddr}
\usepackage{amsfonts}
\usepackage{amssymb}
\usepackage{pgf}
\usepackage{paralist}
\usepackage{pdflscape}
\usepackage{lscape,amssymb, amsmath,verbatim,graphicx}
\usepackage{epsfig,subfigure,float,subfigure}
\usepackage{graphicx}
\usepackage{epsfig}
\usepackage{amssymb}
\usepackage{amsmath}
\usepackage{amsfonts}
\usepackage{geometry}
\usepackage{scalefnt}
\usepackage{graphicx}
\usepackage{caption}
\usepackage{setspace}
\usepackage{array}
\usepackage{multirow}
\usepackage{amsthm}
\usepackage{amsmath}
\usepackage{amsfonts}
\usepackage{amssymb}
\usepackage{pgf}
\usepackage{paralist}
\usepackage{pdflscape}
\usepackage{rotating}
\usepackage{scalefnt}
\usepackage{xr}

\allowdisplaybreaks
\DeclareFontFamily{OT1}{pzc}{}
\DeclareFontShape{OT1}{pzc}{m}{it}{<-> s * [1.10] pzcmi7t}{}
\DeclareMathAlphabet{\mathpzc}{OT1}{pzc}{m}{it}

\numberwithin{table}{section}
\numberwithin{figure}{section}
\newtheorem{theorem}{Theorem}
\newtheorem{lemma}{Lemma}
\newtheorem{definition}{Definition}

\newtheorem{assumption}{Assumption}

\def\beq{\begin{equation}}
\def\eeq{\end{equation}}
\def\bals{\begin{align*}}
\def\eals{\end{align*}}
\def\bal{\begin{align}}
\def\eal{\end{align}}

\allowdisplaybreaks
\input{tcilatex}

\begin{document}
\title[Fast Changepoint Detection]{Fast Online Changepoint Detection}
\author{Fabrizio Ghezzi}
\address{Department of Economics and Management, University of Pavia, 27100
Pavia, Italy, email:fabrizio.ghezzi01@universitadipavia.it}
\author{Eduardo Rossi}
\address{Department of Economics and Management, University of Pavia, 27100
Pavia, Italy, email:eduardo.rossi@unipv.it}
\author{Lorenzo Trapani}
\address{University of Leicester Business School, University Road, Leicester
LE1 7RH, UK, and Department of Economics and Management, University of
Pavia, 27100 Pavia, Italy, email: lorenzo.l.trapani@gmail.com}
\thanks{\textbf{Acknowledgement. }We are grateful to the participants to:
the VTSS workshop for junior researcher (virtual, April 20th, 2023); the
UCSD Econometrics Lunch Seminars (San Diego, May 17th, 2023), in particular
Graham Elliott and Allan Timmermann; the 10th Italian Congress of
Econometrics and Empirical Economics (Cagliari, May 26-28, 2023), in
particular Alessandro Casini and Pierluigi Vallarino. The usual disclaimer
applies.}

\begin{abstract}
We study \textit{online} changepoint detection in the context of a linear
regression model. We propose a class of heavily weighted statistics based on
the CUSUM process of the regression residuals, which are specifically
designed to ensure timely detection of breaks occurring early on during the
monitoring horizon. We subsequently propose a class of composite statistics,
constructed using different weighing schemes; the decision rule to mark a
changepoint is based on the \textit{largest} statistic across the various
weights, thus effectively working like a veto-based voting mechanism, which
ensures fast detection irrespective of the location of the changepoint. Our
theory is derived under a very general form of weak dependence, thus being
able to apply our tests to virtually all time series encountered in
economics, medicine, and other applied sciences. Monte Carlo simulations
show that our methodologies are able to control the procedure-wise Type I
Error, and have short detection delays in the presence of breaks.
\end{abstract}

\maketitle

\doublespacing

\section{Introduction\label{intro}}

In this paper, we consider the linear regression model%
\begin{equation}
y_{t}=\mathbf{x}_{t}^{\prime }\beta _{t}+\epsilon _{t},  \label{regression}
\end{equation}%
and test whether the slopes $\beta _{t}$ are constant over time - i.e. $%
\beta _{t}=\beta _{0}$ for all $t\geq 1$ - or not. In particular, we study 
\textit{sequential/online }changepoint detection: after observing (\ref%
{regression}) over a training period $1\leq t\leq m$ with constant slopes $%
\beta _{0}$, we check whether, as new data come in, $\beta _{t}$ differs
from the previous $\beta _{0}$ for some $t>m$. Tests are carried out in real
time, at each $t>m$, as soon as the new datapoint $\left( y_{t},\mathbf{x}%
_{t}^{\prime }\right) ^{\prime }$ has been recorded.

Since the seminal contribution by \citet{page1954continuous}, online
changepoint detection has proven highly relevant to all applied sciences,
and we refer to a recent review by \citet{aue2023state} for the current
state of the art. In manufacturing, timely detection of malfunctioning in a
production line is important for quality control (see e.g. %
\citealp{page1955control}, and \citealp{page1955test}); in medicine,
changepoint analysis is relevant in the context of sequential medical trials
(\citealp{armitage1960sequential}), or to detect the onset of an epileptic
seizure using EEG trajectories (see e.g. \citealp{ombao}); and, in
epidemiology, detecting the onset of a pandemic by monitoring whether the
number of (daily) cases starts exhibiting an explosive dynamics (see e.g. %
\citealp{HT2023}) is crucial to public health decisions. In economics and
finance, the instability of model parameters is an issue of pivotal
importance due to its consequences on forecasting and policy decision
making; see, \textit{inter alia}, the contributions by \cite%
{pastor2001equity}, \cite{pettenuzzo2011predictability}, and \cite%
{smith2021break}. We also refer to \citet{romano2023fast}, and the
references therein, for further examples illustrating the importance of
quick online changepoint detection. Finding instability in $\beta _{t}$ has
obvious implications on the reliability of a model like (\ref{regression}),
and on decisions based on it: as \citet{hansen2001new} puts it, in the
presence of a changepoint \textquotedblleft inferences about economic
relationships can go astray, forecasts can be inaccurate, and policy
recommendations can be misleading or worse\textquotedblright\ (p. 127).

The literature has developed numerous tests for the ex-post detection of
changepoints, and we refer to the reviews by \citet{horvath2014extensions}
and \citet{casini2019structural}, for an insightful analysis of the state of
the art and of future directions. Despite its practical importance, the
issue of real-time detection is somewhat underdeveloped compared to ex-post
testing. In a seminal contribution, \citet{chu1996monitoring} develop a
methodology to monitor (\ref{regression}), based on the CUSUM process of the
residuals. The intuition is that, given an estimate of $\beta _{0}$ (say $%
\widetilde{\beta }$) computed during a training sample in which no breaks
were detected, the \textquotedblleft prediction errors\textquotedblright\ $%
y_{t}-\mathbf{x}_{t}^{\prime }\widetilde{\beta }$ calculated for $t>m$ will
fluctuate around zero under the null of no changepoint, whereas they will
have a bias in the presence of a changepoint. Thus, the corresponding CUSUM
process will either fluctuate around zero with increasing variance as $t$
increases, or it will have a drift after a changepoint. Hence, if, at each $%
t $, the CUSUM\ process stays within a boundary, the null of no break will
not be rejected; conversely, a break will be flagged at the first crossing
of such a boundary. \citet{chu1996monitoring} derive the weak limit of the
CUSUM process. This allows to compute critical values, and ultimately to
control the procedure-wise Type I Error. In addition to ensuring size
control, the timely detection of a changepoint is also very important; %
\citet{horvath2004monitoring} and \citet{horvath2007sequential} propose a
family of boundary functions which afford a faster break detection when a
changepoint is present. These boundary functions, which represent a major
simplification of the ones studied in \citet{chu1996monitoring}, are based
on a parameterisation which depends on a user-chosen quantity (say, $0\leq
\eta \leq 1/2$) which determines the weights assigned to the CUSUM
fluctuations: as $\eta $ increases, the weight also increases, and therefore
higher power/faster detection may be expected. Several recent contributions
also propose statistics based on the CUSUM process - as a leading example,
which is drawing more and more attention in the literature, %
\citet{kirch2018modified}, \citet{yu2020note}, \citet{romano2023fast}, and %
\citet{HT2023}, study a standardised version of the CUSUM process (called
the \textquotedblleft Page-CUSUM\textquotedblright ) based on comparing the
average during the training period with the average computed during the
monitoring period using a moving window, and choosing the length of such a
window in order to maximise the discrepancy between the two averages in
order to ensure an \textquotedblleft optimal\textquotedblright\ detection
delay when a changepoint is present. Whilst different from our approach,
even in this case the CUSUM process is suitably weighted in order to enance
fast changepoint detection.

\medskip

\textit{Hypotheses of interest and the main contributions of this paper}

\medskip

As mentioned above, we study sequential detection of changepoints in (\ref%
{regression}), allowing for the presence of exogenous and dynamic
regressors, and for a very general form of serial dependence in both $%
\mathbf{x}_{t}$ and $\epsilon _{t}$. After observing the model over a
training sample of size $m$ where no changes in the regression coefficients
have occurred, we test for the null hypothesis of no change occurring at
each point in time, i.e. 
\begin{equation}
H_{0}:\beta _{t}=\beta _{0}\text{ for all }t\geq m+1.  \label{null}
\end{equation}

We make two main contributions. Firstly, we design a class of heavily
weighted CUSUM-type statistics which ensure fast detection for early
occurring changepoints. We call these \textit{R\'{e}nyi statistics},
inspired by \citet{horvath2020new} who, building on an idea by %
\citet{renyi1953theory}, propose a similar idea for the purpose of offline
changepoint detection. We develop the full-blown asymptotics under the null,
and study the consistency, and the limiting behaviour of the detection
delay, under the alternative. We show that R\'{e}nyi statistics deliver an
excellent performance in the presence of early occurring breaks. On the
other hand, when breaks occur at a later stage, heavily weighted statistics
offer a worse performance compared to CUSUM-type statistics using lighter
weights. Indeed, no weighing scheme can ensure uniformly fast changepoint
detection, and different values of $\eta $ yield faster detection for
different break locations (see also \citealp{kirch2022asymptotic}, and %
\citealp{kirch2022sequential}). Thus, as a second contribution, we build a
composite online detection scheme. Our method is inspired by the popular
\textquotedblleft majority vote\textquotedblright\ classifiers (see the
review by \citealp{dietterich2000ensemble}), where different classification
rules are combined into one rule based e.g. on the mode (or
\textquotedblleft majority vote\textquotedblright ) thereof. We propose a
\textquotedblleft veto-based\textquotedblright\ changepoint learning
procedure, where, in essence, at each point in time $t\geq m+1$ the CUSUM\
process is simultaneously weighted using several weighing schemes, with a
changepoint being flagged up as long as the null is rejected for at least
one set of weights. Heuristically, this procedure should yield the fastest
changepoint detection, irrespective of the location of the changepoint,
whilst still offering size control. Referring to the papers by %
\citet{kirch2018modified}, \citet{yu2020note}, and \citet{romano2023fast}
mentioned above, in essence the sequential detection methodology they
investigate is based on using a combination of CUSUM functionals with
different windows, which is highly effective but comes at a relatively high
computational cost. In our case, we similarly propose the combination of
several statistics constructed with different specifications, but at a much
lower computational cost.

\medskip

The remainder of the paper is organised as follows. We present our model and
the main assumptions, in Section \ref{model}. R\'{e}nyi statistics are
studied in Section \ref{SecLRM}; veto-based statistics are in Section \ref%
{veto}. We report a comprehensive Monte Carlo exercise in Section \ref{Sim}.
Section \ref{conclusions} concludes.

NOTATION. Throughout the paper, positive, finite constants are denoted as $%
c_{0}$, $c_{1}$, ... and their value can change from line to line. We denote
vectors using lower case bold face (e.g. $\mathbf{x}$) and their $j$-th
coordinate as $x_{j}$; matrices are denoted using upper case bold face (e.g. 
$\mathbf{A}$) and their element in position $\left( j,h\right) $ is denoted
as $\mathbf{A}_{j,h}$. We use $\left\Vert \mathbf{x}\right\Vert $ to
indicate the Euclidean norm of $\mathbf{x}$; $\left\Vert \mathbf{A}%
\right\Vert _{F}$ to indicate the Frobenius norm of a matrix $\mathbf{A}$;
and $s_{i}\left( \mathbf{A}\right) $ is the $i$-th largest singular value of 
$\mathbf{A}$, with the largest and smallest singular values denoted as $%
s_{\max }\left( \mathbf{A}\right) $\ and $s_{\min }\left( \mathbf{A}\right) $%
\ respectively. We use: \textquotedblleft $\rightarrow $\textquotedblright\
to denote the ordinary limit; \textquotedblleft $\overset{a.s.}{\rightarrow }
$\textquotedblright\ for almost sure convergence; \textquotedblleft $\overset%
{\mathcal{P}}{\rightarrow }$\textquotedblright\ for convergence in
probability; \textquotedblleft $\overset{\mathcal{D}}{\rightarrow }$%
\textquotedblright\ for weak convergence; \textquotedblleft $\overset{%
\mathcal{D}}{=}$\textquotedblright\ to indicate equality in distribution.
Given a random variable $X$, we denote its $L_{\nu }$-norm as $\left\vert
X\right\vert _{\nu }=\left( E\left\vert X\right\vert ^{\nu }\right) ^{1/\nu
} $. We use the symbol $\Omega \left( \cdot \right) $ to indicate the exact
order of magnitude of a sequence, i.e. $s_{m}=\Omega \left( m^{\lambda
}\right) $ means that $\lim_{m\rightarrow \infty }m^{-\lambda }s_{m}=c_{0}$
with $0<c_{0}<\infty $. Finally, $\left\{ W\left( u\right) ,u\geq 0\right\} $%
\ denotes a standard Wiener process. Other, relevant notation is introduced
later on in the paper.

\section{Model and assumptions\label{model}}

In this section, we introduce our main model and assumptions. Our workhorse
model is the linear regression model of equation (\ref{regression}), viz.%
\begin{equation*}
y_{t}=\mathbf{x}_{t}^{\prime }\beta _{t}+\epsilon _{t},
\end{equation*}%
where $\mathbf{x}_{t}$ is a $d$-dimensional vector of covariates with
coordinates denoted as $x_{j,t}$, $1\leq j\leq d$, including a constant -
viz. $x_{1,t}=1$, $t\geq 1$. In (\ref{regression}), the regression
coefficients $\beta _{t}$ may be time-varying. We begin with considering a
static model, where no lagged values of $y_{t}$ are used; in Section \ref%
{dynamic}, we extend our results to the dynamic model. As a final remark, (%
\ref{regression}) naturally nests the standard location model, where $d=1$,
and sequential monitoring boils down to monitoring for shifts in the mean of 
$y_{t}$.

\medskip

We now list the main assumptions. We begin by providing a definition of $%
L_{\nu }$\textit{-decomposable Bernoulli shifts}.

\begin{definition}
\label{bernoulli}The sequence $\left\{ m_{t},-\infty <t<\infty \right\} $
forms an $L_{\nu }$-decomposable Bernoulli shift if and only if $%
m_{t}=h\left( \eta _{t},\eta _{t-1},...\right) $, where: $\left\{ \eta
_{t},-\infty <t<\infty \right\} $ is an \textit{i.i.d.} sequence with values
in a measurable space $S$; $h\left( \cdot \right) :S^{\mathbb{N}}\rightarrow 
\mathbb{R}
$ is a non random measurable function; $\left\vert m_{t}\right\vert _{\nu
}<\infty $; and $\left\vert m_{t}-\widetilde{m}_{t,\ell }\right\vert _{\nu
}\leq c_{0}\ell ^{-a}$, for some $c_{0}>0$ and $a>0$, where $\widetilde{m}%
_{t,\ell }=h\left( \eta _{t},...,\eta _{t-\ell +1},\widetilde{\eta }_{t-\ell
,t,\ell },\right. $ $\left. \widetilde{\eta }_{t-\ell -1,t,\ell }...\right) $%
, with $\left\{ \widetilde{\eta }_{s,t,\ell },-\infty <s,\ell ,t<\infty
\right\} $ \textit{i.i.d.} copies of $\eta _{0}$, independent of $\left\{
\eta _{t},-\infty <t<\infty \right\} $.
\end{definition}

Since the seminal works by \citet{wu2005nonlinear} and %
\citet{berkes2011split}, decomposable Bernoulli shifts have proven a
convenient way to model dependent time series, mainly due to their
generality and to the fact that it is much easier to verify whether a
sequence forms a decomposable Bernoulli shift than e.g. verifying mixing
conditions. Virtually all the most common DGPs in econometrics and
statistics can be shown to satisfy Definition \ref{bernoulli}: %
\citet{liu2009strong}, \textit{inter alia}, provide various theoretical
results, and numerous examples including ARMA models, ARCH/GARCH sequences,
and other nonlinear time series models (such as e.g. random coefficient
autoregressive models and threshold models). Moreover, it is easy to verify
that if $\left\{ m_{t},-\infty <t<\infty \right\} $ is an $L_{\nu }$%
-decomposable Bernoulli shift, then $\left\{ m_{t}^{\kappa },-\infty
<t<\infty \right\} $ is an $L_{\nu /\kappa }$-decomposable Bernoulli shift.
Hence, the weak dependence of Definition \ref{bernoulli} is sufficiently
flexible to allow to study changes in the mean, and also in higher order
moments, of $m_{t}$.

\medskip

We are now ready to present our assumptions.

\begin{assumption}
\label{b-shifts}(i) $\left\{ \epsilon _{t},-\infty <t<\infty \right\} $
forms an $L_{4}$-decomposable Bernoulli shift with $E\left( \epsilon
_{t}\right) =0$, $\left\vert \epsilon _{t}\right\vert _{2}>0$, and $a>2$;
(ii) for $2\leq j\leq d$, $\left\{ x_{j,t},-\infty <t<\infty \right\} $
forms an $L_{4}$-decomposable Bernoulli shift with $a>2$.
\end{assumption}

\begin{assumption}
\label{exogeneity}(i) $E\left( \epsilon _{t}\right) =0$; (ii) $E\left(
x_{j,t}\epsilon _{t}\right) =0$, for all $2\leq j\leq d$; (iii) $E\left( 
\mathbf{x}_{t}\mathbf{x}_{t}^{\prime }\right) =\mathbf{C}$, where $\mathbf{C}
$ is such that $0<s_{\min }\left( \mathbf{C}\right) \leq s_{\max }\left( 
\mathbf{C}\right) <\infty $.
\end{assumption}

Assumption \ref{b-shifts} requires that both the error term $\epsilon _{t}$
and the regressors $\mathbf{x}_{t}$ admit at least four moments. In essence,
this assumption entails the validity of various limit theorems, and in
particular of a uniform estimate of the convergence rate (also known as
\textquotedblleft strong approximation\textquotedblright ) of the invariance
principle (see Lemma \ref{sip}). In principle, it would be possible to
reduce the moment existence requirement, upon assuming independence between $%
\epsilon _{t}$ and $\mathbf{x}_{t}$. As far as serial dependence is
concerned, the requirement that $a>2$ is quite mild; indeed the DGPs
typically used in statistics and econometrics have an \textit{exponential}
rate of decay for $\left\vert m_{t}-\widetilde{m}_{t,\ell }\right\vert _{\nu
}$. As far as Assumption \ref{exogeneity} is concerned,\ parts \textit{(i)}
and \textit{(ii)} are standard; part \textit{(iii)} rules out collinearity,
by stating that the smallest eigenvalue of $\mathbf{C}$ is bounded away from
zero.

\section{Sequential detection of changepoints\label{SecLRM}}

Our statistics are based on comparing the regression coefficients estimated
at each point in time over the monitoring horizon against a benchmark
estimate obtained during a training period $1\leq t\leq m$ in which no break
was observed.

\begin{assumption}
\label{contamination}$\beta _{t}=\beta _{0}$, for all $1\leq t\leq m$.
\end{assumption}

After $m$, the monitoring procedure starts, for the null hypothesis of (\ref%
{null}). Let $\widehat{\beta }_{m}$ be the LS\ estimator using data in the
training sample, viz.%
\begin{equation}
\widehat{\beta }_{m}=\left( \sum_{t=1}^{m}\mathbf{x}_{t}\mathbf{x}%
_{t}^{\prime }\right) ^{-1}\left( \sum_{t=1}^{m}\mathbf{x}_{t}y_{t}\right) ,
\label{ols}
\end{equation}%
and define the LS residuals 
\begin{equation}
\widehat{\epsilon }_{t}=y_{t}-\mathbf{x}_{t}^{\prime }\widehat{\beta }_{m},%
\text{ for }t\geq m+1.  \label{residuals}
\end{equation}%
Heuristically, if no changepoint is present, $\widehat{\beta }_{m}$ is a
valid estimator of $\beta _{0}$ throughout the monitoring period; hence, it
can be expected that $\widehat{\epsilon }_{t}$ will fluctuate across zero
for all $t\geq m+1$. Conversely, in the presence of a break, this creates a
bias term in $\widehat{\epsilon }_{t}$. Hence, a natural way of testing for
changepoints is to use, as detector, a CUSUM-type statistic defined as the
partial sums process of the residuals over the monitoring horizon, i.e.%
\begin{equation}
Q\left( m,k\right) =\sum_{t=m+1}^{m+k}\widehat{\epsilon }_{t},\text{ \ \ for 
}k\geq 1.  \label{qmk}
\end{equation}

We assume that the monitoring does not go on forever, but it is carried out
over a horizon of length, say, $T_{m}$, after which it is terminated.

\begin{assumption}
\label{horizon}$\lim_{m\rightarrow \infty }T_{m}=\infty $.
\end{assumption}

Note that, in Assumption \ref{horizon}, the monitoring horizon can be
\textquotedblleft long\textquotedblright , i.e. we allow for $T_{m}=\Omega
\left( m^{\lambda }\right) $ with $\lambda \geq 1$, but it can also be
\textquotedblleft short\textquotedblright , i.e. $T_{m}=o\left( m\right) $;
the latter case is seldom considered in the literature. Under such a
closed-ended monitoring scheme, our null hypothesis of interest in (\ref%
{null}) becomes 
\begin{equation}
H_{0}:\beta _{t}=\beta _{0}\text{ for all }m+1\leq t\leq m+T_{m}.
\label{monitor-null}
\end{equation}%
In practice, $Q\left( m;k\right) $ is compared against a function of $k$,
and, as soon as it exceeds it, a changepoint is detected. Such a function is
called the \textit{boundary function}; its specification is important in
order to ensure the procedure-wise control of Type I Errors under $H_{0}$
and, on the other hand, to ensure power and a timely detection in the
presence of a break. \citet{horvath2004monitoring} and %
\citet{horvath2007sequential} suggest using, in (\ref{decision}), a family
of boundary function indexed by a user-chosen parameter $0\leq \eta \leq 1/2$%
, defined as%
\begin{equation}
g_{\eta }(m,k)=\sigma m^{1/2}\left( 1+\frac{k}{m}\right) \left( \frac{k}{m+k}%
\right) ^{\eta },  \label{boundary}
\end{equation}%
where 
\begin{equation}
\sigma ^{2}=\lim_{m\rightarrow \infty }E\left(
m^{-1/2}\sum_{t=1}^{m}\epsilon _{t}\right) ^{2},  \label{lrunv}
\end{equation}%
is the long-run variance of $\left\{ \epsilon _{t},-\infty <t<\infty
\right\} $. Hence, a changepoint is marked according to the decision rule%
\begin{equation}
\tau _{m}=%
\begin{cases}
\inf \{k\geq 1:|Q\left( m;k\right) |\geq c_{\alpha ,\eta }g_{\eta }\left(
m;k\right) \}, \\ 
T_{m},\ \text{if}\ |Q\left( m;k\right) |\leq c_{\alpha ,\eta }g_{\eta
}\left( m;k\right) \ \text{for\ all}\ 1\leq k\leq T_{m},%
\end{cases}
\label{decision}
\end{equation}%
where $c_{\alpha ,\eta }$ is a critical value for a pre-specified,
user-defined nominal size $\alpha $. In (\ref{boundary}), intuitively, the
term $m^{1/2}\left( 1+k/m\right) $ is a norming sequence which ensures that $%
Q\left( m;k\right) $ is properly normed; conversely, $\left( k/\left(
m+k\right) \right) ^{\eta }$ is a weight function. Heuristically, the larger 
$\eta $, the smaller the boundary function $g_{\eta }(m,k)$, and therefore
the higher the detection ability of our monitoring scheme. %
\citet{aue2004delay} and \citet{aue2009delay} prove that this intuition is
correct, by showing that the detection delay in the presence of a break is
inversely related to $\eta $. Building on this intuition, here we study (\ref%
{boundary}) with $\eta \in (1/2,\infty )$.

\subsection{Asymptotics under the null for R\'{e}nyi statistics\label%
{nulldist}}

In order to study the limiting distribution under the null, note that (\ref%
{decision}) states that the event $\left\{ \tau _{m}<T_{m}\right\} $ is
tantamount to having 
\begin{equation*}
\max_{1\leq k\leq T_{m}}\left\vert Q\left( m;k\right) \right\vert \geq
c_{\alpha ,\eta }g_{\eta }(m;k).
\end{equation*}%
Using the Law of the Iterated Logarithm, it can be shown that $\max_{1\leq
k\leq T_{m}}\left\vert Q\left( m;k\right) \right\vert \geq c_{\alpha ,\eta
}g_{\eta }(m;k)\overset{a.s.}{\rightarrow }\infty $; however, we can derive
a well-defined limit for 
\begin{equation*}
\max_{a_{m}\leq k\leq T_{m}}\frac{\left\vert Q\left( m;k\right) \right\vert 
}{g_{\eta }(m;k)}\geq c_{\alpha ,\eta },
\end{equation*}%
for a user-specified trimming sequence $a_{m}$ such that 
\begin{equation}
\lim_{m\rightarrow \infty }a_{m}=\infty \text{, \ \ and \ \ }%
\lim_{m\rightarrow \infty }\frac{a_{m}}{\min \left\{ m,T_{m}\right\} }=0.
\label{a-m}
\end{equation}%
In (\ref{a-m}), $a_{m}$ needs to diverge as $m\rightarrow \infty $, but this
can happen at an arbitrarily slow rate.\ From an operational viewpoint, this
entails that our monitoring scheme cannot start straight after $m$, but only
after $m+a_{m}$ periods, so that (\ref{decision}) modifies into%
\begin{equation}
\tau _{m}=%
\begin{cases}
\inf \{k\geq a_{m}:|Q\left( m;k\right) |\geq c_{\alpha ,\eta }g_{\eta
}(m;k)\}, \\ 
T_{m},\ \text{if}\ |Q\left( m;k\right) |\leq c_{\alpha ,\eta }g_{\eta
}(m;k)\ \text{for\ all}\ a_{m}\leq k\leq T_{m};%
\end{cases}
\label{delay-trim}
\end{equation}%
of course, $a_{m}$ cannot be longer than $T_{m}$.

Define the norming sequence $r_{m}=a_{m}/\left( a_{m}+m\right) $. The
following theorem provides the limiting distribution of our statistics under
the null hypothesis.

\begin{theorem}
\label{nul-distr}Under Assumptions \ref{b-shifts}-\ref{horizon} and (\ref%
{a-m}), as $m\rightarrow \infty $ with $d=O\left( m^{1/4}\right) $, it holds
that, for all $\eta >1/2$ 
\begin{equation*}
r_{m}^{\eta -1/2}\max_{a_{m}\leq k\leq T_{m}}\frac{\left\vert
Q(m;k)\right\vert }{g_{\eta }(m;k)}\overset{\mathcal{D}}{\rightarrow }%
\sup_{1\leq u<\infty }\frac{\left\vert W(u)\right\vert }{u^{\eta }}.
\end{equation*}
\end{theorem}

Theorem \ref{nul-distr} provides the weak limit of the ratio between the
detector $Q(m;k)$ and the boundary function $g_{\eta }\left( m;k\right) $,
thus offering a way of computing asymptotic critical values. We are not
aware of any closed form expression for the quantiles of $\sup_{1\leq
u<\infty }u^{-\eta }|W(u)|$, which therefore must be computed by simulation.
However, when $\eta \leq 1$, exploiting the scale transformation of the
Wiener process we obtain 
\begin{equation*}
\sup_{1\leq u<\infty }\frac{|W(u)|}{u^{\eta }}\overset{\mathcal{D}}{=}%
\sup_{1\leq u<\infty }\frac{|W(u^{-1})|}{\left( u\right) ^{\eta -1}}\overset{%
\mathcal{D}}{=}\sup_{0\leq s\leq 1}\frac{|W(s)|}{s^{1-\eta }};
\end{equation*}%
in such a case, the quantiles of $\sup_{1\leq u<\infty }u^{-\eta }|W(u)|$
are the same as those in Table 1 in \citet{horvath2004monitoring}, computed
for the case $1-\eta $.

As far as $a_{m}$ is concerned, it would be desirable to have $a_{m}$ as
small as possible, so as to ensure a fast detection of breaks occurring
close to the beginning of the monitoring horizon. Also, Theorem \ref%
{nul-distr} requires the restriction $d=O\left( m^{1/4}\right) $. This can
be read in two ways: on the one hand, given $m$, it poses a limit on the
dimension of $\mathbf{x}_{t}$ the size of (\ref{regression}), indicating
that some dimension reduction must be carried out prior to monitoring (\ref%
{regression}). On the other hand, the size of the training sample $m$ can be
viewed as a tuning parameter whose choice depends, inter alia, on $d$.

\medskip

We conclude by noting that, in order to make our statistics feasible, an
estimator of the long-run variance $\sigma ^{2}$ is needed. Based on
Assumption \ref{b-shifts}, we propose a standard,
weighted-sum-of-covariances estimator, computed using the full training
sample. Seeing as no breaks occur between $1\leq t\leq m$, such an estimator
can be expected to be consistent under both the null and the alternative.
Let $\widehat{\gamma }_{j}=m^{-1}\sum_{t=j+1}^{m}\widehat{\epsilon }_{t}%
\widehat{\epsilon }_{t-j}$; we propose the following Bartlett-type estimator%
\begin{equation}
\widehat{\sigma }_{m}^{2}=\widehat{\gamma }_{0}+2\sum_{j=1}^{H}\left( 1-%
\frac{j}{H+1}\right) \widehat{\gamma }_{j},  \label{sig-hat}
\end{equation}%
where $H$ is a bandwidth parameter. Other estimators are possible, and we
refer to \citet{andrews1991heteroskedasticity} and \citet{casini2023theory}\
for useful references.

\begin{lemma}
\label{lrv}Under Assumptions \ref{b-shifts}-\ref{contamination}, it holds
that 
\begin{equation*}
\widehat{\sigma }_{m}^{2}-\sigma ^{2}=O_{P}\left( \frac{1}{H}\right)
+O_{P}\left( \frac{dH}{m^{1/2}}\right) .
\end{equation*}
\end{lemma}

Lemma \ref{lrv} indicates that, up to an appropriate choice of the bandwidth 
$H$, $\widehat{\sigma }_{m}^{2}$ is a consistent estimator of $\sigma ^{2}$.
The lemma suggests that a possible choice, designed to balance the
traditional \textquotedblleft bias\textquotedblright\ term of order $%
O_{P}\left( 1/H\right) $ and the \textquotedblleft
variance\textquotedblright\ term of order $O_{P}\left( dH/m^{1/2}\right) $,
is $H=O\left( d^{-1/2}m^{1/4}\right) $.

\subsection{Monitoring schemes based on R\'{e}nyi statistics under
alternatives\label{alternative}}

Under the alternative, we assume that there is a changepoint at $k^{\ast }$:%
\begin{equation}
H_{A}:\left\{ 
\begin{array}{cl}
\beta _{t}=\beta _{0} & 1\leq t<m+k^{\ast }, \\ 
\beta _{t}=\beta _{A} & m+k^{\ast }\leq t\leq m+T_{m}.%
\end{array}%
\right.  \label{alt-changepoint}
\end{equation}%
We allow the break size to depend on the sample size $m$, thus defining $%
\Delta _{m}=\beta _{A}-\beta _{0}$.

\begin{theorem}
\label{power}We assume that Assumptions \ref{b-shifts}-\ref{horizon}, (\ref%
{a-m}) and $\mathbf{c}_{1}^{\prime }\Delta _{m}\neq 0$, are satisfied.
Assume further that either: (i) $k^{\ast }=o\left( a_{m}\right) $ and%
\begin{equation}
\liminf_{m\rightarrow \infty }a_{m}^{1/2}\mathbf{c}_{1}^{\prime }\Delta
_{m}=\infty ;  \label{regime-1}
\end{equation}%
or (ii) $\liminf_{m\rightarrow \infty }a_{m}^{-1}k^{\ast }>0$, $%
\liminf_{m\rightarrow \infty }T_{m}/k^{\ast }>1$\ and 
\begin{equation}
\liminf_{m\rightarrow \infty }r_{m}^{\eta -1/2}m^{1/2}\left( \frac{k^{\ast }%
}{m+k^{\ast }}\right) ^{1-\eta }\mathbf{c}_{1}^{\prime }\Delta _{m}=\infty .
\label{regime-2}
\end{equation}%
Then it holds that, for all $\eta >1/2$, $\lim_{m\rightarrow \infty }P\left(
\tau _{m}<T_{m}|H_{A}\right) =1$.
\end{theorem}

Part \textit{(i)} of Theorem \ref{power} states that our statistics have
power versus a changepoint occurring close to the beginning of the
monitoring horizon. Interestingly, under (\ref{regime-1}), breaks can be
\textquotedblleft small\textquotedblright , and the case $\left\Vert \Delta
_{m}\right\Vert \rightarrow 0$ as $m\rightarrow \infty $\ can be considered;
however, breaks cannot be \textquotedblleft too small\textquotedblright .
Conceptually similar results are derived in \citet{horvath2022changepoint}
for the case of offline changepoint detection. According to part \textit{(ii)%
} of the theorem, tests also have power versus later occurring breaks, i.e.
under $\liminf_{m\rightarrow \infty }a_{m}^{-1}k^{\ast }>0$ and (\ref%
{regime-2}). For example, considering the case of $k^{\ast }=O\left(
m\right) $, condition (\ref{regime-2}) boils down to 
\begin{equation}
\liminf_{m\rightarrow \infty }a_{m}^{\eta -1/2}\left( k^{\ast }\right)
^{1-\eta }\left\vert \mathbf{c}_{1}^{\prime }\Delta _{m}\right\vert =\infty ,
\label{cond-comment-1}
\end{equation}%
which indicates that, whenever $\eta \leq 1$, detection will take place.
However, upon carefully studying the proof of Theorem 2.2 in %
\citet{horvath2004monitoring}, it follows that whenever $\eta <1/2$
changepoints are detected as long as $m^{1/2}\left\vert \mathbf{c}%
_{1}^{\prime }\Delta _{m}\right\vert \rightarrow \infty $, which is more
easily satisfied than (\ref{cond-comment-1}), and allows for smaller
changepoint magnitudes. Hence, whilst ensuring power versus early
changepoints, the use of heavy weights in the boundary function $g_{\eta
}\left( m;k\right) $ reduces the ability to detect changepoints occurring
later on during the monitoring horizon. This can be seen by considering the
case of $\eta >1$ in (\ref{cond-comment-1}), or the case where $%
m^{-1}k^{\ast }\rightarrow \infty $. In the latter case, condition \textit{%
(ii)} is tantamount to%
\begin{equation}
\liminf_{m\rightarrow \infty }a_{m}^{\eta -1/2}m^{1-\eta }\left\vert \mathbf{%
c}_{1}^{\prime }\Delta _{m}\right\vert =\infty ,  \label{cond-comment-2}
\end{equation}%
where again values of $\eta >1$ are detrimental to the power of the
monitoring scheme. Finally, although we have focussed on the case of an
abrupt break, following our proofs it is possible to verify that our
monitoring schemes have power also in the presence of smooth breaks.

\medskip

Theorem \ref{power} indicates that our statistics seem particularly suited
to the detection of early changepoints. We now study this case more
in-depth, considering the asymptotics for the detection delay in the
presence of early occurring changes.

\begin{assumption}
\label{delay-assumption}(i) $k^{\ast }=o\left( a_{m}\right) $; (ii) $%
\left\Vert \mathbf{c}_{1}\right\Vert =\Omega \left( d^{1/2}\right) $; (iii)
as $m\rightarrow \infty $, $a_{m}^{1/2}d^{1/2}\left\Vert \Delta
_{m}\right\Vert \rightarrow \infty $.
\end{assumption}

Assumption \ref{delay-assumption} deals with the location and the size of
the changepoint. According to part \textit{(i)}, the break is assumed to
occur early, before the start of the monitoring procedure. This can be read
in conjunction with similar results derived in \citet{aue2004delay} and %
\citet{aue2009delay}, where the asymptotic distribution of $\tau _{m}$ is
derived under the assumption of an early break.

\begin{theorem}
\label{delay}Under Assumptions \ref{b-shifts}-\ref{delay-assumption} and (%
\ref{a-m}), it holds that $\tau _{m}\overset{\mathcal{P}}{\rightarrow }a_{m}$%
.
\end{theorem}

Theorem \ref{delay} states that, if a break occurs prior to the trimming
point $a_{m}$, then it is detected immediately at the start of the
monitoring procedure. This result offers an important insight on the
relevance and practical use of our statistics. The theory derived in %
\citet{aue2004delay} and \citet{aue2008monitoring} indicates that using $%
\eta \leq 1/2$ ensures detection of changepoints occurring not too soon
after $m$; in particular, using $\eta =1/2$ ensures a short delay - of order 
$O\left( \sqrt{\ln \ln m}\right) $ - in detecting breaks occurring $o\left( 
\sqrt{\ln \ln m}\right) $ periods from $m$. Based on this and on Theorem \ref%
{delay}, in the case of our statistics, the delay is of order $a_{m}$ for
breaks occurring $o\left( a_{m}\right) $ periods from $m$: upon choosing $%
a_{m}=o\left( \sqrt{\ln \ln m}\right) $, this entails that early occurring
breaks are found even more quickly than when using $\eta =1/2$, with the
added advantage that one does not need to use the asymptotic critical values
computed for the case $\eta =1/2$, where convergence is notoriously slow.

\subsection{Extensions to the case of dynamic regressors\label{dynamic}}

For the ease of exposition, in the above we have considered the case of
exogenous regressors only. We now extend our results to the case of a
dynamic model:%
\begin{equation}
y_{t}=\mathbf{z}_{t}^{\prime }\beta _{t}^{Z}+\mathbf{y}_{p,t}^{\prime }\beta
_{t}^{Y}+\epsilon _{t}=\mathbf{x}_{t}^{\prime }\beta _{t}+\epsilon _{t},
\label{dyn-reg}
\end{equation}%
where $\mathbf{y}_{p,t}^{\prime }=\left( y_{t-1},...,y_{t-p}\right) ^{\prime
}$, $\beta _{t}^{Y}=\left( \rho _{1},...,\rho _{p}\right) ^{\prime }$, $%
y_{0} $ is an initial condition, and $\mathbf{z}_{t}$ is a $\left(
d-p\right) $-dimensional vector of exogenous regressors, and $\mathbf{x}%
_{t}=\left( \mathbf{z}_{t}^{\prime },\mathbf{y}_{p,t}^{\prime }\right)
^{\prime }$ is - by construction - of dimension $d$. With this notation, we
are able to define the null hypothesis of interest (\ref{null}), the
alternative hypothesis (\ref{alt-changepoint}), and the monitoring
statistics, exactly in the same way as above.

\begin{assumption}
\label{ass-dyn}(i) $\left\vert y_{0}\right\vert _{4}<\infty $; (ii) the
polynomial equation $1-\sum_{j=1}^{p}\rho _{j}a^{j}=0$, with $a\in \mathbb{C}
$, has all roots outside the unit circle.
\end{assumption}

\begin{theorem}
\label{dynamic-reg}Under e assume that the data are generated by (\ref%
{dyn-reg}) and Assumption \ref{ass-dyn}, the results in Theorems \ref%
{nul-distr}, \ref{power} and \ref{delay}, and Lemma \ref{lrv}, hold under
the same assumptions.
\end{theorem}

\section{Veto-based changepoint learning\label{veto}}

Our theory, and the theory developed in other contributions (e.g. %
\citealp{aue2004delay} and \citealp{aue2008monitoring}), suggests that
different values of $\eta $ work better for different changepoint locations.
Hence, any sequential monitoring methodology based on a specific $\eta $ can
be expected to offer superior detection timeliness in some cases, and to be
dominated by different choices of $\eta $ in other cases; this issue is also
echoed in a recent contribution by \citet{kirch2022sequential}.

In this section, we consider a combination of several detection rules. Upon
inspecting our proofs and the proofs in \citet{horvath2004monitoring}, it
can be shown that test statistics using $\eta <1/2$ and $\eta >1/2$ are
asymptotically independent of each other,\footnote{%
Intuitively, this is due to the facts that: (1) the limiting distribution
when $\eta <1/2$ is determined by the \textquotedblleft
central\textquotedblright\ part of $\left\{ W\left( u\right) ,u\geq
0\right\} $, whereas, when using $\eta >1/2$, the limit is completely
determined by the \textquotedblleft early\textquotedblright\ part of $%
\left\{ W\left( u\right) ,u\geq 0\right\} $ (see also the comments at the
end of the proof of Theorem \ref{nul-distr}); and (2) that the Wiener
process has independent increments.} although of course two or more
statistics with different $\eta >1/2$ are not independent of each other (the
same holds for statistics with different $\eta <1/2$). Hence, it is in
principle possible to combine different test statistics, although using a
combination of tests, e.g. with the Bonferroni correction, is bound to be
subject to the typical criticisms of being overconservative and, therefore,
liable to not detect (timely, or at all) a changepoint. Hence, in the light
of the considerations above, and in order to ensure the fast changepoint
detection whilst controlling for false positives, we propose a different,
composite sequential detection scheme. We base our methodology on a
\textquotedblleft veto-based\textquotedblright\ approach, where several
statistics (corresponding to different choices of $\eta $) are employed
simultaneously, and the null is rejected as long as at least one statistic
exceeds its critical value. Heuristically, the veto-based decision rule
offers a compromise between the various choices of $\eta $, ensuring that
detection delay is good irrespective of the changepoint location.

Recall the detector $Q\left( m;k\right) $ defined in (\ref{qmk}). For each
choice of $\eta _{j}$, $j=1,2,...,J$, we have a corresponding boundary
function $g_{\eta _{j}}(m;k)$, and a critical value at nominal level $\alpha 
$, $c_{\alpha ,\eta _{j}}$, with rejection occurring at the first value of $%
k\geq 1$ such that $\left\vert Q\left( m;k\right) \right\vert >c_{\alpha
,\eta _{j}}g_{\eta _{j}}(m;k)$. We combine several statistics using $0\leq
\eta _{j}\leq 1$; recall also that, when $\eta _{j}>1/2$, the monitoring
starts after $a_{m}$ periods from the beginning of the $1\leq k\leq T_{m}$
horizon. Hence, the veto-based procedure stops at $\widetilde{\tau }_{m}$
defined as%
\begin{equation}
\widetilde{\tau }_{m}=%
\begin{cases}
\inf \{k\geq 1:|Q\left( m;k\right) |\geq C_{\alpha }\min_{1\leq j\leq
J}c_{\alpha ,\eta _{j}}d_{a_{m},\eta }g_{\eta _{j}}\left( m;k\right) \}, \\ 
T_{m},\ \text{if}\ |Q\left( m;k\right) |\leq C_{\alpha }\min_{1\leq j\leq
J}c_{\alpha ,\eta _{j}}d_{a_{m},\eta }g_{\eta _{j}}\left( m;k\right) \ \text{%
for\ all}\ 1\leq k\leq T_{m},%
\end{cases}
\label{tilde-tau}
\end{equation}%
where%
\begin{equation*}
\widetilde{r}_{m}=\left\{ 
\begin{array}{ll}
1 & \text{if }\eta _{j}<1/2, \\ 
r_{m} & \text{if }\eta _{j}\geq 1/2,%
\end{array}%
\right.
\end{equation*}%
and%
\begin{equation*}
d_{a_{m},\eta _{j}}=\left\{ 
\begin{array}{ll}
1 & \text{if }k\geq a_{m}\text{ \ or \ }\eta _{j}<1/2, \\ 
\infty & \text{if }k<a_{m}\text{ \ and \ }\eta _{j}>1/2.%
\end{array}%
\right.
\end{equation*}%
In (\ref{tilde-tau}), $C_{\alpha }$ is the procedure-wise critical value,
such that $\lim_{m\rightarrow \infty }P\left( \widetilde{\tau }%
_{m}<T_{m}|H_{0}\right) =\alpha $.

In order to implement (\ref{tilde-tau}), we require an asymptotic
approximation of the procedure-wise critical value $C_{\alpha }$. Define $%
\widetilde{\eta }_{j}=\eta _{j}I\left( \eta _{j}\leq 1/2\right) +\left(
1-\eta _{j}\right) I\left( \eta _{j}>1/2\right) $, and modify Assumption \ref%
{horizon} to consider a long monitoring horizon.

\begin{assumption}
\label{horizon copy(1)}(i) $\lim_{m\rightarrow \infty }T_{m}=\infty $; (ii) $%
\liminf_{m\rightarrow \infty }T_{m}/m>0$; (iii) $T_{m}=\Omega \left(
m^{\lambda }\right) $, with $\lambda \geq 1$.
\end{assumption}

Assumption \ref{horizon copy(1)} is required to derive the asymptotics of
monitoring statistics with $\eta <1/2$ (see \citealp{horvath2004monitoring}).

\medskip

The following results characterise the null distribution, and the power
under alternatives, of veto-based statistics.

\begin{theorem}
\label{veto-prop}We assume that Assumptions \ref{b-shifts}-\ref%
{contamination} and \ref{horizon copy(1)}, (\ref{a-m}) and $d=o\left(
m^{1/4}\right) $, are satisfied. Then, under the null, as $m\rightarrow
\infty $\ it holds that%
\begin{equation}
\max_{1\leq k\leq T_{m}}\frac{\left\vert Q\left( m;k\right) \right\vert }{%
\min_{1\leq j\leq J}c_{\alpha ,\eta _{j}}\widetilde{r}_{m}^{1/2-\eta
_{j}}d_{a_{m},\eta }g_{\eta _{j}}(m;k)}\overset{\mathcal{D}}{\rightarrow }%
\sup_{0<u<1}\frac{\left\vert W\left( u\right) \right\vert }{\min_{1\leq
j\leq J}c_{\alpha ,\eta _{j}}u^{\widetilde{\eta }_{j}}}.  \label{null-veto}
\end{equation}
\end{theorem}

\begin{theorem}
\label{veto-power}We assume that the assumptions of Theorem \ref{veto-prop}
hold, and further that either: (i) $k^{\ast }=o\left( a_{m}\right) $ and (%
\ref{regime-1}); or (ii) $\liminf_{m\rightarrow \infty }a_{m}^{-1}k^{\ast
}>0 $ and (\ref{regime-2}) hold; or (iii) that $k^{\ast }=O\left( m\right) $
and%
\begin{equation}
\liminf_{m\rightarrow \infty }\max_{1\leq j\leq J}\frac{\left( k^{\ast
}\right) ^{1-\eta _{j}}}{\widetilde{r}_{m}^{1/2-\eta _{j}}m^{1/2-\eta _{j}}}%
\left\vert \mathbf{c}_{1}^{\prime }\Delta _{m}\right\vert =\infty .
\label{veto-diverge}
\end{equation}%
Then it holds that%
\begin{equation}
\lim_{m\rightarrow \infty }P\left( \widetilde{\tau }_{m}<T_{m}|H_{A}\right)
=1.  \label{pow-veto}
\end{equation}
\end{theorem}

Some comments are in order. Theorem \ref{veto-prop} stipulates that,
asymptotically%
\begin{equation}
P\left( \sup_{0<u<1}\frac{\left\vert W\left( u\right) \right\vert }{%
\min_{1\leq j\leq J}c_{\alpha ,\eta _{j}}u^{\widetilde{\eta }_{j}}}\leq
C_{\alpha }\right) =\alpha ;  \label{asy-veto}
\end{equation}%
hence, $C_{\alpha }$ can be obtained by standard Monte Carlo techniques.

According to Theorem \ref{veto-power}, the veto-based decision rule does
have power versus changepoints; note that the proposition studies only the
case $k^{\ast }=O\left( m\right) $ - i.e. breaks occurring not too late
during the monitoring period - but it could be extended to the case where $%
k^{\ast }$ is \textquotedblleft bigger\textquotedblright\ than $m$ after
elementary, if tedious, algebra. Condition (\ref{veto-diverge}) is a
high-level requirement which can be more easily interpreted by considering
some leading examples. To begin with, we investigate the case $k^{\ast
}=c_{0}m$, corresponding to a break occurring later on during the monitoring
horizon. In such a case, a possible question is how small can the size of
the break, $\left\vert \mathbf{c}_{1}^{\prime }\Delta _{m}\right\vert $, be.
When $\eta _{j}<1/2$, the non-centrality parameter in (\ref{veto-diverge})
is proportional to $m^{1/2}\left\vert \mathbf{c}_{1}^{\prime }\Delta
_{m}\right\vert $, which entails that nontrivial power is found versus
breaks of magnitude at least $O\left( m^{-1/2}\right) $; note that this
holds irrespective of the actual value of $\eta _{j}$. On the other hand,
when $\eta _{j}>1/2$, (\ref{veto-diverge}) is proportional to $a_{m}^{\eta
_{j}-1/2}m^{1-\eta _{j}}\left\vert \mathbf{c}_{1}^{\prime }\Delta
_{m}\right\vert $, suggesting that the larger the choice of $\eta _{j}$, the
larger $\left\vert \mathbf{c}_{1}^{\prime }\Delta _{m}\right\vert $ needs to
be to be detected. In this case, the veto-based rule will be triggered by
the statistics with $\eta _{j}<1/2$, thus ensuring power versus small breaks.

In conclusion, we would like to point out that an alternative approach to
the \textquotedblleft classical\textquotedblright\ CUSUM detectors employed
here, is the so-called Page-CUSUM detectors (see e.g. %
\citealp{fremdt2015page}, \citealp{kirch2018modified}, and %
\citealp{romano2023fast}), defined as 
\begin{equation}
Q^{\dagger }(m;k)=\max_{1\leq \ell \leq k}\left\vert \sum_{i=m+\ell }^{m+k}%
\frac{\left( y_{i}-\widehat{\beta }_{m}y_{i-1}\right) y_{i-1}}{1+y_{i-1}^{2}}%
\right\vert .  \label{page}
\end{equation}%
\citet{HT2023} consider the weighted version of (\ref{page}). Intuitively,
at each point in time $1\leq k\leq T_{m}$, the detector is set equal to the
\textquotedblleft worst case\textquotedblright\ observed up to $k$. This
should make the detector $Q^{\dagger }(m;k)$ more sensitive, and more prone
to using only the observations which, up to time $k$, are affected by a
possible changepoint. Indeed, the simulations in \citet{HT2023} show that
the use of the Page-CUSUM detector does offer some improvement on the
standard CUSUM, but the main gains can be ascribed to the use of heavier and
heavier weighing schemes. Furthermore, the use of Page-CUSUM methodologies
is marred by their computational complexity (see however the solution
proposed in a recent contribution by \citealp{romano2023fast}). In this
respect, the veto-based approach suggested in this section is similar, in
spirit, to (\ref{page}); instead of using the \textquotedblleft worst case
scenario\textquotedblright\ across $k$, we use the (computationally simpler)
\textquotedblleft worst case scenario\textquotedblright\ across several
values of $\eta $.

\section{Simulations\label{Sim}}

We report some Monte Carlo evidence to evaluate the finite sample
performance of our methodology, and to offer guidelines on the choice of
some specifications, chiefly $a_{m}$.\footnote{%
Further Monte Carlo evidence is reported in Section \ref{furtherMC} in the
Supplement.} All our experiments are based on the following design%
\begin{equation}
y_{t}=\mathbf{x}_{t}^{\prime }\beta _{t}+\rho y_{t-1}+\epsilon _{t},
\label{dgp}
\end{equation}%
where we allow for the possible presence of a dynamic term. Under the null,
we generate $\epsilon _{t}$ as \textit{i.i.d.} $N\left( 0,1\right) $ for $%
1\leq t\leq m+T_{m}$, and each coordinate of $\beta _{0}$ as $\beta
_{0,j}\sim 1+\sigma _{\beta }N\left( 0,1\right) $, independently across $%
1\leq j\leq d$. Our results are derived under $\rho =0.5$ and $\sigma
_{\beta }=0.5$, but unreported experiments showed that the value of $\sigma
_{\beta }$ does not alter our results. We allow for serial dependence in the 
$d$-dimensional regressor $\mathbf{x}_{t}$, generating its (demeaned)
coordinates as%
\begin{equation*}
x_{j,t}=\phi x_{j,t-1}+e_{j,t},
\end{equation*}%
with $\phi =0.5$, and $e_{j,t}\sim N\left( 0,1\right) $, \textit{i.i.d.}
across $2\leq j\leq d$ and $1\leq t\leq m+T_{m}$ (recall that $x_{1,t}=1$).
In the case of a static regression - i.e., when $\rho =0$ in (\ref{dgp})\ -
we allow for the regression error $\epsilon _{t}$ to be serially dependent:%
\begin{equation*}
\epsilon _{t}=\theta \epsilon _{t-1}+w_{t},
\end{equation*}%
with $w_{t}\sim N\left( 0,1\right) $ \textit{i.i.d.} across $1\leq t\leq
m+T_{m}$. We have used $\theta =0.5$; again, unreported simulations show
that altering this specification does not have a significant impact on our
result. We estimate the long-run variance using (\ref{sig-hat}) with
bandwidth $H=\left\lfloor m^{2/5}\right\rfloor $. Finally, we set $T_{m}=m$
across all experiments; in the Supplement, we further assess the sensitivity
of our result to the length of the monitoring horizon $T_{m}$.

As far as implementation is concerned, we consider several choices for the
trimming sequence $a_{m}$, studying how these affect the empirical rejection
frequencies and the delays. We also consider several values of $\eta $, in
order to assess the impact of this specification on the performance of our
methodologies. By way of comparison, we also use $\eta \leq 1/2$; in this
case, we have used the critical values in Table 1 in %
\citet{horvath2004monitoring}. All the results reported here are obtained
for $d=2$ (that is, one exogenous regressor and the constant, in addition to
the lagged dependent variable) and a nominal level $\alpha =0.05$. Finally,
all simulations were carried out with $2,500$ replications.

\medskip

\begin{table*}[t]
\caption{{\protect\footnotesize {Empirical rejection frequencies under $H_0$}%
}}
\label{tab:Size1}\centering
{\footnotesize {\ }}
\par
{\scriptsize {\ }}
\par
{\scriptsize 
\begin{tabular}{llccccccccc}
\multicolumn{1}{c}{} & \multicolumn{1}{c}{$a_m$} & \multicolumn{3}{c}{$\ln
\ln m$} & \multicolumn{3}{c}{$\ln m$} & \multicolumn{3}{c}{$\ln^2 m$} \\ 
\multicolumn{2}{c}{} & \multicolumn{3}{c}{} & \multicolumn{3}{c}{} & 
\multicolumn{3}{c}{} \\ 
\cmidrule(lr){3-5}\cmidrule(lr){6-8}\cmidrule(lr){9-11} &  &  &  &  &  &  & 
&  &  &  \\ 
~ & $m$ & 300 & 500 & 1000 & 300 & 500 & 1000 & 300 & 500 & 1000 \\ 
\multicolumn{1}{c}{$\eta$} & \multicolumn{3}{c}{} & \multicolumn{3}{c}{} & 
\multicolumn{3}{c}{} &  \\ 
\multicolumn{2}{c}{} & \multicolumn{3}{c}{} & \multicolumn{3}{c}{} & 
\multicolumn{3}{c}{} \\ 
0.51 & ~ & 0.051 & 0.047 & 0.040 & 0.046 & 0.042 & 0.034 & 0.032 & 0.028 & 
0.029 \\ 
0.55 & ~ & 0.051 & 0.046 & 0.034 & 0.057 & 0.051 & 0.044 & 0.046 & 0.046 & 
0.046 \\ 
0.65 & ~ & 0.039 & 0.035 & 0.028 & 0.052 & 0.050 & 0.043 & 0.062 & 0.056 & 
0.054 \\ 
0.75 & ~ & 0.036 & 0.030 & 0.025 & 0.047 & 0.043 & 0.038 & 0.062 & 0.052 & 
0.052 \\ 
0.85 & ~ & 0.039 & 0.031 & 0.028 & 0.048 & 0.048 & 0.040 & 0.066 & 0.055 & 
0.055 \\ 
1 & ~ & 0.044 & 0.033 & 0.030 & 0.046 & 0.048 & 0.041 & 0.066 & 0.051 & 0.054
\\ 
\multicolumn{11}{c}{} \\ 
\midrule \bottomrule &  &  &  &  &  &  &  &  &  & 
\end{tabular}
}
\par
{\scriptsize \ }
\par
{\scriptsize {\footnotesize 
\begin{tablenotes}
      \tiny
            \item The table contains the empirical rejection frequencies under the null of no changepoint for different sample sizes and different values of $\eta$, for the case of the dynamic regression model (\ref{dgp}), with $d=2$ - i.e., one exogenous regressor and the constant. The specifications of the model are described in the main text.
            
\end{tablenotes}
} }
\end{table*}

\medskip

In Table \ref{tab:Size1}, we report the empirical rejection frequencies
under the null of no break. In the context of sequential monitoring the
notion of \textquotedblleft size control\textquotedblright\ is different,
and the nominal level $\alpha $ represents an upper bound for the
probability of a procedure-wise Type I Error.\footnote{%
As \citet{horvath2007sequential} put it, \textquotedblleft \lbrack t]he goal
is to keep the probability of false rejection below $\alpha $ rather than to
make it close to $\alpha $\textquotedblright .} Our empirical rejection
frequencies broadly stay beneath the $5\%$ upper bound (with some exceptions
for \textquotedblleft large\textquotedblright\ values of $a_{m}$ and small $%
m $), and decline as $m$ increases, indicating that procedure-wise size
control is ensured. The choice of $a_{m}$ is important in determining the
empirical rejection frequencies; as our theory predicts, as $m$ increases
all empirical rejection frequencies fall below the $5\%$ level, but large
values of $a_{m}$ - e.g. corresponding to the choice $a_{m}=\ln ^{2}m$ -
require comparatively higher sample sizes to ensure proper size control. In
Section \ref{furtherMC} in the Supplement, we report more Monte Carlo
evidence, considering also several values of $\eta \leq 1/2$ by way of
comparison, on: empirical rejection frequencies under the null with a larger
number of regressors (Tables \ref{tab:SizeDyn4}-\ref{tab:SizeDyn8}), showing
that size control, while still guaranteed in general, tends to worsen as $d$
increases; and empirical rejection frequencies under the null using a static
regression model, with various $d$ (Tables \ref{tab:SizeSt2}-\ref%
{tab:SizeSt8}), showing that, in this case, the sequential detection
procedure becomes more conservative \textit{ceteris paribus} compared to a
dynamic model.

\medskip

\begin{table*}[t]
\caption{{\protect\footnotesize {Median delays under an early occurring
changepoint}}}
\label{tab:DelayE1}\centering
{\footnotesize {\ }}
\par
{\scriptsize {\ }}
\par
{\scriptsize 
\begin{tabular}{llllllllllllllll}
\multicolumn{16}{c}{Median Delay for Early Changepoint, $d=2$} \\ 
\multicolumn{16}{c}{} \\ 
$a_m$ & ~ & $\eta$ & 0 & 0.15 & 0.25 & 0.35 & 0.45 & 0.49 & 0.5 & 0.51 & 0.55
& 0.65 & 0.75 & 0.85 & 1 \\ 
\multicolumn{16}{c}{} \\ 
\multicolumn{16}{c}{\ ~~~~~~~~$m=300$~~~~~} \\ 
\multicolumn{16}{c}{} \\ 
$\ln \ln m$ & ~ & ~ & 27 & 18 & 13 & 8 & 5 & 4 & 4 & 4 & 3.5 & 3 & 3 & 4 & 1
\\ 
$\ln m$ & ~ & ~ & 28 & 18 & 14 & 10 & 7 & 7 & 7 & 7 & 6 & 6 & 6 & 6 & 6 \\ 
$\ln^2 m$ & ~ & ~ & 30 & 23 & 19 & 16 & 14 & 15 & 16 & 14 & 13 & 12 & 11 & 11
& 11 \\ 
\multicolumn{16}{c}{} \\ \hline
\multicolumn{16}{c}{} \\ 
\multicolumn{16}{c}{\ ~~~~~~~~$m=500$~~~~~} \\ 
\multicolumn{16}{c}{} \\ 
$\ln \ln m$ &  & ~ & 35 & 22 & 15 & 9 & 5 & 4 & 4 & 4 & 3 & 3 & 3 & 3 & 1 \\ 
$\ln m$ & ~ & ~ & 35 & 22 & 16 & 11 & 8 & 8 & 8 & 7 & 7 & 6 & 6 & 6 & 6.5 \\ 
$\ln^2 m$ & ~ & ~ & 38 & 28 & 23 & 19 & 16 & 16 & 17 & 16 & 14 & 13 & 12 & 12
& 12 \\ 
\multicolumn{16}{c}{} \\ \hline
\multicolumn{16}{c}{} \\ 
\multicolumn{16}{c}{\ ~~~~~~~~$m=1000$~~~~~} \\ 
\multicolumn{16}{c}{} \\ 
$\ln \ln m$ & ~ & ~ & 49 & 29 & 19 & 10 & 6 & 4 & 5 & 4 & 3 & 3 & 3 & 4 & 1
\\ 
$\ln m$ & ~ & ~ & 50 & 30 & 20 & 13 & 8 & 8 & 8 & 7 & 6 & 6 & 6 & 6 & 6.5 \\ 
$\ln ^2 m$ & ~ & ~ & 51 & 35 & 28 & 21 & 17 & 17 & 18 & 17 & 15 & 14 & 13 & 
13 & 13 \\ 
\multicolumn{16}{c}{} \\ 
\bottomrule \bottomrule &  &  &  &  &  &  &  &  &  &  &  &  &  &  & 
\end{tabular}
}
\par
{\scriptsize \ }
\par
{\scriptsize {\footnotesize 
\begin{tablenotes}
      \tiny
            \item The table contains the median delay under the alternative of an \textit{early} occurring changepoint, for different sample sizes, different values of $\eta$, and different trimming sequence $a_m$, in the case of the dynamic regression model (\ref{dgp}), with $d=2$ - i.e., one exogenous regressor and the constant. The specifications of the model are described in the main text.
            
\end{tablenotes}
} }
\end{table*}

\medskip

\begin{table*}[b]
\caption{{\protect\footnotesize {Median delays under a late occurring
changepoint}}}
\label{tab:DelayL1}\centering
{\footnotesize {\ }}
\par
{\scriptsize {\ }}
\par
{\scriptsize 
\begin{tabular}{llllllllllllllll}
\multicolumn{16}{c}{} \\ 
$a_m$ & ~ & $\eta$ & 0 & 0.15 & 0.25 & 0.35 & 0.45 & 0.49 & 0.5 & 0.51 & 0.55
& 0.65 & 0.75 & 0.85 & 1 \\ 
\multicolumn{16}{c}{} \\ 
\multicolumn{16}{c}{\ ~~~~~~~~$m=300$~~~~~} \\ 
\multicolumn{16}{c}{} \\ 
$\ln \ln m$ & ~ & ~ & 37 & 31 & 29 & 27 & 26 & 28 & 30 & 29 & 32 & 49 & 88 & 
154 & NA \\ 
$\ln m$ & ~ & ~ & 38 & 32 & 29 & 27 & 27 & 28 & 30 & 29 & 29 & 37 & 51 & 74
& 140 \\ 
$\ln^2 m$ & ~ & ~ & 39 & 34 & 32 & 30 & 31 & 32 & 34 & 32 & 31 & 32 & 35 & 39
& 48 \\ 
\multicolumn{16}{c}{} \\ \hline
\multicolumn{16}{c}{} \\ 
\multicolumn{16}{c}{\ ~~~~~~~~$m=500$~~~~~} \\ 
\multicolumn{16}{c}{} \\ 
$\ln \ln m$ & ~ & ~ & 42 & 34 & 30 & 27 & 25 & 26 & 28 & 27 & 30 & 46 & 84 & 
202 & NA \\ 
$\ln m$ & ~ & ~ & 43 & 35 & 31 & 27 & 26 & 27 & 29 & 27 & 28 & 35 & 46.5 & 66
& 152 \\ 
$\ln^2 m$ & ~ & ~ & 45 & 37 & 34 & 31 & 30 & 31 & 33 & 31 & 29 & 30 & 33 & 37
& 45 \\ 
\multicolumn{16}{c}{} \\ \hline
\multicolumn{16}{c}{} \\ 
\multicolumn{16}{c}{\ ~~~~~~~~$m=1000$~~~~~} \\ 
\multicolumn{16}{c}{} \\ 
$\ln \ln m$ & ~ & ~ & 54 & 40 & 33 & 28 & 24 & 25 & 27 & 25 & 28 & 43 & 80 & 
195 & NA \\ 
$\ln m$ & ~ & ~ & 54 & 40 & 33 & 28 & 25 & 25 & 27 & 26 & 26 & 32 & 44 & 64
& 154 \\ 
$\ln^2 m$ & ~ & ~ & 56 & 43 & 37 & 32 & 29 & 30 & 32 & 30 & 28 & 28 & 31 & 34
& 41 \\ 
\multicolumn{16}{c}{} \\ 
\bottomrule \bottomrule &  &  &  &  &  &  &  &  &  &  &  &  &  &  & 
\end{tabular}
}
\par
{\scriptsize \ }
\par
{\scriptsize {\footnotesize 
\begin{tablenotes}
      \tiny
            \item The table contains the median delay under the alternative of a \textit{late} occurring changepoint, for different sample sizes, different values of $\eta$, and different trimming sequence $a_m$, in the case of the dynamic regression model (\ref{dgp}), with $d=2$ - i.e., one exogenous regressor and the constant. The specifications of the model are described in the main text. 
            \item The symbol \textquotedblleft NA\textquotedblright indicates that no changepoint was detected.
            
\end{tablenotes}
} }
\end{table*}

\medskip

In Tables \ref{tab:DelayE1}-\ref{tab:DelayL1}, we report the detection delay
in the presence of a changepoint, considering two types of alternatives: an
early break, occurring exactly $a_{m}$ periods after the start of the
monitoring horizon, and a late break, occurring at $k^{\ast }=m+a_{m}$. We
compare detection delays across several values of $\eta $, and for different
choices of $a_{m}$. In the case of early occurring changepoints, sequential
monitoring with $\eta >1/2$ ensures a superior performance compared to $\eta
\leq 1/2$, as predicted by Theorem \ref{delay}. This confirms that using
heavy weights ensures a very fast detection of early occurring breaks. This
is essentially true for all values of $\eta $, with the choice $\eta =1$
yielding the best results. Interestingly, the choice $a_{m}=\ln \ln m$
delivers the fastest detection, as predicted by the theory, and despite the
fact that in this case the test is undersized. However, on the other hand,
the results in Table \ref{tab:DelayE1} cannot be directly compared across $%
a_{m}$, since the location of changepoints is different; for example, when $%
a_{m}=\ln ^{2}m$ and therefore also $k^{\ast }=\ln ^{2}m$, this corresponds
to a break occurring early, but not \textquotedblleft very
early\textquotedblright . In this case, detection delays are actually quite
good: Theorem \ref{delay} predicts that breaks are detected after $a_{m}=\ln
^{2}m$ steps, but the results in Table \ref{tab:DelayE1} indicate that such
detection occurs much more quickly (for example, when $m=1000$, it holds
that $\ln ^{2}m\approx 48$, but detection takes place after only $12$
periods for large values of $\eta $). This result should further be read in
conjunction with the fact that, according to Table \ref{tab:DelayE1}, the
choice $a_{m}=\ln ^{2}m$ results in oversizement for small samples. As far
as early break detection is concerned, results are reversed in the presence
of a late occurring break, as Table \ref{tab:DelayL1}\ shows. In such a
case, delays improve as $\eta $ moves towards $1/2$ - both from the left and
the right - but heavily weighted statistics perform substantially worse than
when using $\eta \leq 1/2$; this is especially true when considering the
delays under very early occurring breaks in Tables \ref{tab:DelayE1}, under $%
a_{m}=\ln \ln m$ and even $a_{m}=\ln m$.

In Section \ref{furtherMC} in the Supplement, we report a full-blown of
statistics on detection delays under both alternatives of an early occurring
break and a late one, for various values of $d$, and for both a dynamic
(Tables \ref{tab:DelayE2}-\ref{tab:DelayLF3}) and a static (Tables \ref%
{tab:DelayES1}-\ref{tab:DelayLS3}) regression, essentially confirming the
results discussed above.

\bigskip

The distilled essence of Tables \ref{tab:DelayE1}-\ref{tab:DelayL1} is that
- as far as detection delay is concerned - the \textquotedblleft
optimal\textquotedblright\ weight for the CUSUM process, which ensures the
fastest detection irrespective of the changepoint location, does not exist.
Hence, using an agnostic approach such as the one proposed in Section \ref%
{veto}, based on combining various statistics, may be preferable. In the
next set of experiments, we consider several combinations of weighted CUSUM
statistics; in particular, we use a scheme with one lightly weighted and one
heavily weighted statistics (denoted as $\mathcal{V}_{2}$, and based on
using $\eta =0.2$ and $\eta =0.85$); two lightly weighted and one heavily
weighted statistics (denoted as $\mathcal{V}_{3}$, and based on using $\eta
=0.2$, $\eta =0.3$ and $\eta =0.85$); and two lightly weighted and three
heavily weighted statistics (denoted as $\mathcal{V}_{5}$, and based on
using $\eta =0.2$, $\eta =0.45$, $\eta =0.65$, $\eta =0.85$, and $\eta =0.9$%
).\footnote{%
Results with other combinations of $\eta $ confirm the findings reported
here, and are available upon request.} We use the same DGP and
specifications as above; by way of comparison, in Table \ref{tab:VetoSize}
we also report empirical rejection frequencies under the null, and
descriptive statistics for the detection delay under the alternative, also
for the cases of individual weighted CUSUM statistics based on $\eta =0.25$
and $\eta =0.75$.

\medskip

\begin{table*}[h!]
\caption{{\protect\footnotesize {Empirical rejection frequencies under $H_0$
- veto-based sequential detection}}}
\label{tab:VetoSize}\centering
{\footnotesize {\ }}
\par
{\scriptsize {\ }}
\par
{\scriptsize 
\begin{tabular}{llccccccccc}
\multicolumn{11}{c}{} \\ 
\multicolumn{1}{c}{} & \multicolumn{1}{c}{$a_m$} & \multicolumn{3}{c}{$\ln
\ln m$} & \multicolumn{3}{c}{$\ln m$} & \multicolumn{3}{c}{$\ln^2 m$} \\ 
\multicolumn{2}{c}{} & \multicolumn{3}{c}{} & \multicolumn{3}{c}{} & 
\multicolumn{3}{c}{} \\ 
\cmidrule(lr){3-5}\cmidrule(lr){6-8}\cmidrule(lr){9-11} &  &  &  &  &  &  & 
&  &  &  \\ 
~ & $m$ & 300 & 500 & 1000 & 300 & 500 & 1000 & 300 & 500 & 1000 \\ 
\multicolumn{1}{c}{$\eta$} & \multicolumn{3}{c}{} & \multicolumn{3}{c}{} & 
\multicolumn{3}{c}{} &  \\ 
\multicolumn{2}{c}{} & \multicolumn{3}{c}{} & \multicolumn{3}{c}{} & 
\multicolumn{3}{c}{} \\ 
0.25 & ~ & 0.020 & 0.020 & 0.019 & 0.020 & 0.020 & 0.019 & 0.020 & 0.020 & 
0.019 \\ 
0.75 & ~ & 0.036 & 0.030 & 0.025 & 0.047 & 0.043 & 0.038 & 0.062 & 0.052 & 
0.052 \\ 
$\mathcal{V}_2$ & ~ & 0.052 & 0.047 & 0.040 & 0.061 & 0.058 & 0.052 & 0.070
& 0.062 & 0.060 \\ 
$\mathcal{V}_3$ & ~ & 0.058 & 0.054 & 0.049 & 0.064 & 0.063 & 0.059 & 0.070
& 0.064 & 0.068 \\ 
$\mathcal{V}_5$ & ~ & 0.057 & 0.048 & 0.044 & 0.056 & 0.052 & 0.049 & 0.050
& 0.044 & 0.045 \\ 
\multicolumn{11}{c}{} \\ 
\midrule \bottomrule &  &  &  &  &  &  &  &  &  & 
\end{tabular}
}
\par
{\scriptsize \ }
\par
{\scriptsize {\footnotesize 
\begin{tablenotes}
      \tiny
            \item The table contains the empirical rejection frequencies under the null of no changepoint for different veto-based monitoring schemes, for the case of the dynamic regression model (\ref{dgp}), with $d=2$ - i.e., one exogenous regressor and the constant. The specifications of the model are described in the main text.
            
\end{tablenotes}
} }
\end{table*}

\medskip

\begin{table*}[h!]
\caption{{\protect\footnotesize {Median delays under $H_{A}$ - veto-based
sequential detection}}}
\label{tab:VetoDelays}\centering
{\footnotesize {\ }}
\par
{\scriptsize {\ }}
\par
{\scriptsize 
\begin{tabular}{lllllllllllllll}
\multicolumn{15}{c}{} \\ 
\multicolumn{4}{c}{} & \multicolumn{5}{c}{Early Break} &  & 
\multicolumn{5}{c}{Late Break} \\ 
\cmidrule(lr){4-9}\cmidrule(lr){10-15} &  &  &  &  &  &  &  &  &  &  &  &  & 
&  \\ 
\multicolumn{15}{c}{} \\ 
$a_m$ &  & $\eta$ & ~ & 0.25 & 0.75 & $\mathcal{V}_{2}$ & $\mathcal{V}_{3} $
& $\mathcal{V}_{5}$ & ~ & 0.25 & 0.75 & $\mathcal{V}_2$ & $\mathcal{V}_{3} $
& $\mathcal{V}_{5}$ \\ 
\multicolumn{15}{c}{} \\ 
$\ln \ln m$ & Min & ~ & ~ & 3 & 1 & 1 & 1 & 1 & ~ & 3 & 25 & 5 & 3 & 2 \\ 
~ & Q1 & ~ & ~ & 10 & 1 & 1 & 1 & 1 & ~ & 22 & 66 & 24 & 21 & 20 \\ 
~ & MED & ~ & ~ & 15 & 3 & 4 & 4 & 4 & ~ & 30 & 84 & 32 & 29 & 28 \\ 
~ & ARL & ~ & ~ & 16.31 & 8.05 & 9.74 & 8.25 & 6.15 & ~ & 31.58 & 87.92 & 
33.52 & 29.95 & 29.19 \\ 
~ & Q3 & ~ & ~ & 21 & 10 & 17 & 14 & 9 & ~ & 39 & 105 & 41 & 37 & 36.75 \\ 
~ & Max & ~ & ~ & 70 & 123 & 71 & 68 & 66 & ~ & 94 & 221 & 100 & 93 & 94 \\ 
~ & ~ & ~ & ~ & ~ & ~ & ~ & ~ & ~ & ~ & ~ & ~ & ~ & ~ &  \\ 
$\ln m$ & Min & ~ & ~ & 3 & 1 & 1 & 1 & 1 & ~ & 3 & 8 & 4 & 2 & 1 \\ 
~ & Q1 & ~ & ~ & 12 & 3 & 3 & 3 & 4 & ~ & 23 & 36 & 24 & 21 & 21 \\ 
~ & MED & ~ & ~ & 16 & 6 & 6 & 6 & 7 & ~ & 31 & 46.5 & 33 & 29 & 29 \\ 
~ & ARL & ~ & ~ & 17.92 & 8.09 & 8.81 & 8.77 & 8.48 & ~ & 32.25 & 48.93 & 
34.23 & 30.71 & 30.00 \\ 
~ & Q3 & ~ & ~ & 22 & 10 & 11 & 11 & 11 & ~ & 40 & 60 & 42 & 38 & 38 \\ 
~ & Max & ~ & ~ & 76 & 83 & 78 & 65 & 65 & ~ & 96 & 136 & 98 & 95 & 96 \\ 
~ & ~ & ~ & ~ & ~ & ~ & ~ & ~ & ~ & ~ & ~ & ~ & ~ & ~ &  \\ 
$\ln^2 m$ & Min & ~ & ~ & 3 & 1 & 1 & 1 & 1 & ~ & 1 & 1 & 3 & 2 & 2 \\ 
~ & Q1 & ~ & ~ & 17 & 8 & 7 & 7 & 8 & ~ & 25 & 24 & 27 & 24 & 24 \\ 
~ & MED & ~ & ~ & 23 & 12 & 12 & 12 & 13 & ~ & 34 & 33 & 36 & 32 & 32 \\ 
~ & ARL & ~ & ~ & 24.13 & 13.94 & 13.88 & 13.99 & 14.91 & ~ & 35.39 & 34.87
& 37.20 & 34.09 & 34.25 \\ 
~ & Q3 & ~ & ~ & 30 & 19 & 19 & 19 & 20 & ~ & 44 & 44 & 46 & 42 & 43 \\ 
~ & Max & ~ & ~ & 73 & 60 & 60 & 62 & 60 & ~ & 100 & 103 & 101 & 99 & 100 \\ 
\multicolumn{15}{c}{} \\ 
\bottomrule \bottomrule &  &  &  &  &  &  &  &  &  &  &  &  &  & 
\end{tabular}
}
\par
{\scriptsize \ }
\par
{\scriptsize {\footnotesize 
\begin{tablenotes}
      \tiny
            \item The table contains several descriptive statistics for the delay under alternative veto-based monitoring schemes, in the case of the dynamic regression model (\ref{dgp}), with $d=2$ - i.e., one exogenous regressor and the constant. The specifications of the model are described in the main text. 
            \item The descriptive statistics reported in the table are: the minimum and maximum values ($\min$ and $\max$ respectively), the first and third quartiles (Q1 and Q3 respectively), the median delay (denoted as MED), and the average delay (denoted as ARL).
            
\end{tablenotes}
} }
\end{table*}

\medskip

Table \ref{tab:VetoSize} states that veto-based detection schemes are able
to afford procedure-wise Type I Error control, especially when $a_{m}$ is
\textquotedblleft small\textquotedblright\ - corresponding to $a_{m}=\ln \ln
m$; indeed, results are very similar to the ones obtained using the
individual, heavily weighted CUSUM statistics. As far as power and timely
detection under alternatives are concerned, Table \ref{tab:VetoDelays}
indicates that veto-based statistic guarantee fast detection, both in the
presence of an early break and of a late break. In the former case, the
veto-based procedure delivers essentially the same performance as using a
single, heavily weighted CUSUM statistic, improving on lightly weighted
CUSUM statistics which are less effective; conversely, in the latter case,
the veto-based rules achieve a similar result as lightly weighted CUSUM
statistics, whereas the heavily weighted ones perform poorly in comparison\
- thus confirming the intuition that veto-based detection rules should
combine \textquotedblleft the best of both worlds\textquotedblright . These
results hold, broadly, for all combinations considered, although $\mathcal{V}%
_{3}$ and $\mathcal{V}_{5}$ fare better than $\mathcal{V}_{2}$, suggesting
that combining more than two CUSUM\ statistics would result in better
detection in all cases considered. As noted in Tables \ref{tab:DelayE1}-\ref%
{tab:DelayL1}, when using heavily weighted statistics, the best results as
far as detection delays are concerned are found when using $a_{m}=\ln \ln m$%
. This finding is important, because the choice of $a_{m}=\ln \ln m$ also
corresponds to the best size control under the null. In the Supplement, we
report more Monte Carlo evidence focusing on the case of a dynamic
regression, studying: empirical rejection frequencies under the null with
various values of $d$\ (Tables \ref{tab:VetoApp1} and \ref{tab:VetoApp2}),
showing that, as above, size control worsens as $d$ increases, although
using $a_{m}=\ln \ln m$ still offers the best results in all cases
considered; and detection delays under both an early and a late occurring
break with various values of $d$ (Tables \ref{tab:VetoApp3} and \ref%
{tab:VetoApp4}), showing that these improve as $d$ increase, mainly as a
consequence of the size/power trade-off. 

\section{Conclusions and discussion\label{conclusions}}

In this paper, we investigate the well-known, and highly important, issue of
(fast) online changepoint detection. We focus on a linear regression model,
in the possible presence of weak dependence, which also nests the cases of
changepoint detection in the location of a time series $\left\{
y_{t},-\infty <t<\infty \right\} $; moreover, we use a flexible definition
of dependence which allows to readily extend our methods to test for changes
in the higher order moments of $\left\{ y_{t},-\infty <t<\infty \right\} $.

We make two contributions. First, we introduce a class of statistics, called 
\textit{R\'{e}nyi statistics}\ and based on the heavily weighted CUSUM
process of the regression residuals, which are specifically designed to
detect early occurring changepoints. Second, we propose a composite,
veto-based, test statistic which combines different weighted versions of the
CUSUM process of the residuals, designed to ensure timely detection of
changepoints occurring at any point in time over the monitoring horizon. Our
Monte Carlo evidence shows that our statistics offer good control of the
probability of the procedure-wise Type I Error under the null, also
affording timely detection under the alternative. Whilst we focus, as far as
the empirical illustration is concerned, on a macroeconomic example, we
would like to point out that our methodology can be applied to a wide
variety of contexts in virtually all applied sciences.

Several interesting questions remain open and are worth pursuing. In the
construction of heavily weighted CUSUM statistics, the only tuning parameter
is the sequence $a_{m}$ defined in (\ref{a-m}). Our simulations suggest that
a good choice to ensure both size control and timely detction is $a_{m}=\ln
\ln m$ in all cases considered; however, having an automated, data-driven
selection rule would be desirable. Heuristically, our proofs (and our
simulations) show that the quality of the asymptotic distribution under the
null is determined by two factors: the speed of convergence in the Gaussian
approximation (see Lemma \ref{sip}), and the impact of $a_{m}$ in the
limiting behaviour of $\max_{a_{m}/m\leq u\leq T_{m}/m}\left\vert W\left(
u/\left( 1+u\right) \right) \right\vert /\left( u/\left( 1+u\right) \right)
^{\eta }$ (see the proof of Theorem \ref{nul-distr}). Seeing as $a_{m}$
plays a role only in determining the limit of a heavily weighted Wiener
process, a possible suggestion to the applied user would be to choose the
optimal $a_{m}$ by running a simulation with \textit{i.i.d. }$N\left(
0,1\right) $ data, generated with training and monitoring samples of size $m$
and $T_{m}$ respectively, determining which choice of $a_{m}$ yields the
best size control without having an overly conservative procedure. As a
second important extension, recently the literature has considered
extensions of online change-point detection methods to different data
structures such as e.g. network data. \citet{yu2020note} and %
\citet{dubey2021online} study online changepoint detection in networks, and %
\citet{wang2021optimal} extend the CUSUM to the network case. Our
methodologies could lend themselves to being extended to other types of data
structures. These issues have a high priority on the authors' research
agenda.

{\footnotesize {\ 
\bibliographystyle{chicago}
\bibliography{biblio}
} }

\newpage

\clearpage
\renewcommand*{\thesection}{\Alph{section}}

\setcounter{section}{0} \setcounter{subsection}{-1} %
\setcounter{subsubsection}{-1} \setcounter{equation}{0} \setcounter{lemma}{0}
\setcounter{theorem}{0} \renewcommand{\theassumption}{A.\arabic{assumption}} 
\renewcommand{\thetheorem}{A.\arabic{theorem}} \renewcommand{\thelemma}{A.%
\arabic{lemma}} \renewcommand{\theproposition}{A.\arabic{proposition}} %
\renewcommand{\thecorollary}{A.\arabic{corollary}} \renewcommand{%
\theequation}{A.\arabic{equation}}

\section{Further Monte Carlo evidence\label{furtherMC}}

\subsection{Empirical rejection frequencies under the null\label%
{size_further}}


\begin{table*}[b]
\caption{{\protect\footnotesize {Empirical rejection frequencies under $H_0$%
, dynamic regression with $d=4$}}}
\label{tab:SizeDyn4}\centering
\par
{\scriptsize {\ }}
\par
{\scriptsize 
\begin{tabular}{llccccccccc}
\multicolumn{11}{c}{} \\ 
\multicolumn{1}{c}{} & \multicolumn{1}{c}{$a_m$} & \multicolumn{3}{c}{$\ln
\ln m$} & \multicolumn{3}{c}{$\ln m$} & \multicolumn{3}{c}{$\ln^2 m$} \\ 
\multicolumn{2}{c}{} & \multicolumn{3}{c}{} & \multicolumn{3}{c}{} & 
\multicolumn{3}{c}{} \\ 
\cmidrule(lr){3-5}\cmidrule(lr){6-8}\cmidrule(lr){9-11} &  &  &  &  &  &  & 
&  &  &  \\ 
~ & $m$ & 300 & 500 & 1000 & 300 & 500 & 1000 & 300 & 500 & 1000 \\ 
\multicolumn{1}{c}{$\eta$} & \multicolumn{3}{c}{} & \multicolumn{3}{c}{} & 
\multicolumn{3}{c}{} &  \\ 
\multicolumn{2}{c}{} & \multicolumn{3}{c}{} & \multicolumn{3}{c}{} & 
\multicolumn{3}{c}{} \\ 
0.51 & ~ & 0.059 & 0.047 & 0.038 & 0.052 & 0.040 & 0.038 & 0.038 & 0.033 & 
0.029 \\ 
0.55 & ~ & 0.055 & 0.043 & 0.033 & 0.064 & 0.052 & 0.041 & 0.059 & 0.050 & 
0.050 \\ 
0.65 & ~ & 0.042 & 0.032 & 0.026 & 0.059 & 0.052 & 0.038 & 0.079 & 0.057 & 
0.052 \\ 
0.75 & ~ & 0.036 & 0.031 & 0.023 & 0.052 & 0.046 & 0.034 & 0.078 & 0.059 & 
0.051 \\ 
0.85 & ~ & 0.039 & 0.036 & 0.028 & 0.049 & 0.047 & 0.038 & 0.078 & 0.065 & 
0.056 \\ 
1 & ~ & 0.042 & 0.042 & 0.031 & 0.046 & 0.045 & 0.040 & 0.072 & 0.062 & 0.050
\\ 
\multicolumn{11}{c}{} \\ 
\midrule \bottomrule &  &  &  &  &  &  &  &  &  & 
\end{tabular}
}
\par
{\scriptsize \ }
\par
{\scriptsize {\footnotesize 
\begin{tablenotes}
      \tiny
            \item The table contains the empirical rejection frequencies under the null of no changepoint for different sample sizes and different values of $\eta$, for the case of the dynamic regression model (\ref{dgp}), and $d=4$ - i.e. 3 exogenous regressors and the constant. The other specifications of the model are the same as for Table \ref{tab:Size1}. 
            
\end{tablenotes}
} }
\end{table*}

\medskip

\begin{table*}[b]
\caption{{\protect\footnotesize {Empirical rejection frequencies under $H_0$%
, dynamic regression with $d=8$}}}
\label{tab:SizeDyn8}\centering
\par
{\scriptsize {\ }}
\par
{\scriptsize 
\begin{tabular}{llccccccccc}
\multicolumn{11}{c}{} \\ 
\multicolumn{1}{c}{} & \multicolumn{1}{c}{$a_m$} & \multicolumn{3}{c}{$\ln
\ln m$} & \multicolumn{3}{c}{$\ln m$} & \multicolumn{3}{c}{$\ln^2 m$} \\ 
\multicolumn{2}{c}{} & \multicolumn{3}{c}{} & \multicolumn{3}{c}{} & 
\multicolumn{3}{c}{} \\ 
\cmidrule(lr){3-5}\cmidrule(lr){6-8}\cmidrule(lr){9-11} &  &  &  &  &  &  & 
&  &  &  \\ 
~ & $m$ & 300 & 500 & 1000 & 300 & 500 & 1000 & 300 & 500 & 1000 \\ 
\multicolumn{1}{c}{$\eta$} & \multicolumn{3}{c}{} & \multicolumn{3}{c}{} & 
\multicolumn{3}{c}{} &  \\ 
\multicolumn{2}{c}{} & \multicolumn{3}{c}{} & \multicolumn{3}{c}{} & 
\multicolumn{3}{c}{} \\ 
0.51 & ~ & 0.069 & 0.062 & 0.055 & 0.063 & 0.059 & 0.051 & 0.042 & 0.040 & 
0.043 \\ 
0.55 & ~ & 0.065 & 0.060 & 0.045 & 0.076 & 0.071 & 0.058 & 0.067 & 0.058 & 
0.058 \\ 
0.65 & ~ & 0.046 & 0.038 & 0.027 & 0.070 & 0.067 & 0.048 & 0.086 & 0.075 & 
0.068 \\ 
0.75 & ~ & 0.040 & 0.034 & 0.028 & 0.062 & 0.056 & 0.042 & 0.085 & 0.072 & 
0.066 \\ 
0.85 & ~ & 0.041 & 0.035 & 0.027 & 0.065 & 0.058 & 0.046 & 0.088 & 0.069 & 
0.067 \\ 
1 & ~ & 0.042 & 0.036 & 0.028 & 0.060 & 0.059 & 0.045 & 0.079 & 0.063 & 0.062
\\ 
\multicolumn{11}{c}{} \\ 
\midrule \bottomrule &  &  &  &  &  &  &  &  &  & 
\end{tabular}
}
\par
{\scriptsize \ }
\par
{\scriptsize {\footnotesize 
\begin{tablenotes}
      \tiny
            \item The table contains the empirical rejection frequencies under the null of no changepoint for different sample sizes and different values of $\eta$, for the case of the dynamic regression model (\ref{dgp}), and $d=8$ - i.e. 7 exogenous regressors and the constant. The other specifications of the model are the same as for Table \ref{tab:Size1}. 
            
\end{tablenotes}
} }
\end{table*}

\medskip

\begin{table*}[h!]
\caption{{\protect\footnotesize {Empirical rejection frequencies under $H_0$%
, static regression with $d=2$}}}
\label{tab:SizeSt2}\centering
\par
{\footnotesize {\ }}
\par
{\scriptsize {\ }}
\par
{\scriptsize 
\begin{tabular}{llccccccccc}
\multicolumn{11}{c}{} \\ 
\multicolumn{1}{c}{} & \multicolumn{1}{c}{$a_m$} & \multicolumn{3}{c}{$\ln
\ln m$} & \multicolumn{3}{c}{$\ln m$} & \multicolumn{3}{c}{$\ln^2 m$} \\ 
\multicolumn{2}{c}{} & \multicolumn{3}{c}{} & \multicolumn{3}{c}{} & 
\multicolumn{3}{c}{} \\ 
\cmidrule(lr){3-5}\cmidrule(lr){6-8}\cmidrule(lr){9-11} &  &  &  &  &  &  & 
&  &  &  \\ 
~ & $m$ & 300 & 500 & 1000 & 300 & 500 & 1000 & 300 & 500 & 1000 \\ 
\multicolumn{1}{c}{$\eta$} & \multicolumn{3}{c}{} & \multicolumn{3}{c}{} & 
\multicolumn{3}{c}{} &  \\ 
\multicolumn{2}{c}{} & \multicolumn{3}{c}{} & \multicolumn{3}{c}{} & 
\multicolumn{3}{c}{} \\ 
0.51 & ~ & 0.044 & 0.035 & 0.034 & 0.053 & 0.044 & 0.039 & 0.045 & 0.038 & 
0.034 \\ 
0.55 & ~ & 0.034 & 0.024 & 0.018 & 0.057 & 0.047 & 0.037 & 0.058 & 0.055 & 
0.052 \\ 
0.65 & ~ & 0.006 & 0.004 & 0.002 & 0.043 & 0.039 & 0.031 & 0.072 & 0.068 & 
0.057 \\ 
0.75 & ~ & 0.001 & 0.000 & 0.000 & 0.034 & 0.028 & 0.025 & 0.072 & 0.061 & 
0.054 \\ 
0.85 & ~ & 0.001 & 0.000 & 0.000 & 0.032 & 0.029 & 0.025 & 0.074 & 0.062 & 
0.057 \\ 
1 & ~ & 0.002 & 0.000 & 0.000 & 0.029 & 0.028 & 0.024 & 0.069 & 0.059 & 0.056
\\ 
\multicolumn{11}{c}{} \\ 
\midrule \bottomrule &  &  &  &  &  &  &  &  &  & 
\end{tabular}
}
\par
{\scriptsize \ }
\par
{\scriptsize {\footnotesize \ } }
\par
{\scriptsize {\footnotesize 
\begin{tablenotes}
      \tiny
            \item The table contains the empirical rejection frequencies under the null of no changepoint for different sample sizes and different values of $\eta$, for the case of the static regression model (\ref{dgp}), and $d=2$ - i.e. one exogenous regressor and the constant. The other specifications of the model are the same as for Table \ref{tab:Size1}. 
            
\end{tablenotes}
} }
\end{table*}

\medskip

\begin{table*}[h!]
\caption{{\protect\footnotesize {Empirical rejection frequencies under $H_0$%
, static regression with $d=4$}}}
\label{tab:SizeSt4}\centering
\par
{\footnotesize {\ }}
\par
{\scriptsize {\ }}
\par
{\scriptsize 
\begin{tabular}{llccccccccc}
\multicolumn{11}{c}{} \\ 
\multicolumn{1}{c}{} & \multicolumn{1}{c}{$a_m$} & \multicolumn{3}{c}{$\ln
\ln m$} & \multicolumn{3}{c}{$\ln m$} & \multicolumn{3}{c}{$\ln^2 m$} \\ 
\multicolumn{2}{c}{} & \multicolumn{3}{c}{} & \multicolumn{3}{c}{} & 
\multicolumn{3}{c}{} \\ 
\cmidrule(lr){3-5}\cmidrule(lr){6-8}\cmidrule(lr){9-11} &  &  &  &  &  &  & 
&  &  &  \\ 
~ & $m$ & 300 & 500 & 1000 & 300 & 500 & 1000 & 300 & 500 & 1000 \\ 
\multicolumn{1}{c}{$\eta$} & \multicolumn{3}{c}{} & \multicolumn{3}{c}{} & 
\multicolumn{3}{c}{} &  \\ 
\multicolumn{2}{c}{} & \multicolumn{3}{c}{} & \multicolumn{3}{c}{} & 
\multicolumn{3}{c}{} \\ 
0.51 & ~ & 0.051 & 0.040 & 0.034 & 0.057 & 0.045 & 0.042 & 0.051 & 0.040 & 
0.039 \\ 
0.55 & ~ & 0.038 & 0.027 & 0.019 & 0.064 & 0.048 & 0.042 & 0.073 & 0.057 & 
0.059 \\ 
0.65 & ~ & 0.009 & 0.006 & 0.005 & 0.050 & 0.034 & 0.026 & 0.089 & 0.071 & 
0.059 \\ 
0.75 & ~ & 0.002 & 0.001 & 0.001 & 0.035 & 0.030 & 0.021 & 0.083 & 0.068 & 
0.055 \\ 
0.85 & ~ & 0.002 & 0.001 & 0.001 & 0.034 & 0.030 & 0.022 & 0.083 & 0.072 & 
0.058 \\ 
1 & ~ & 0.002 & 0.000 & 0.001 & 0.031 & 0.026 & 0.021 & 0.078 & 0.070 & 0.053
\\ 
\multicolumn{11}{c}{} \\ 
\midrule \bottomrule &  &  &  &  &  &  &  &  &  & 
\end{tabular}
}
\par
{\scriptsize \ }
\par
{\scriptsize {\footnotesize \ } }
\par
{\scriptsize {\footnotesize 
\begin{tablenotes}
      \tiny
            \item The table contains the empirical rejection frequencies under the null of no changepoint for different sample sizes and different values of $\eta$, for the case of the static regression model (\ref{dgp}), and $d=4$ - i.e. 3 exogenous regressors and the constant. The other specifications of the model are the same as for Table \ref{tab:Size1}. 
            
\end{tablenotes}
} }
\end{table*}

\medskip

\begin{table*}[h!]
\caption{{\protect\footnotesize {Empirical rejection frequencies under $H_0$%
, static regression with $d=8$}}}
\label{tab:SizeSt8}\centering
\par
{\footnotesize {\ }}
\par
{\scriptsize {\ 
\begin{tabular}{llccccccccc}
\multicolumn{11}{c}{} \\ 
\multicolumn{1}{c}{} & \multicolumn{1}{c}{$a_m$} & \multicolumn{3}{c}{$\ln
\ln m$} & \multicolumn{3}{c}{$\ln m$} & \multicolumn{3}{c}{$\ln^2 m$} \\ 
\multicolumn{2}{c}{} & \multicolumn{3}{c}{} & \multicolumn{3}{c}{} & 
\multicolumn{3}{c}{} \\ 
\cmidrule(lr){3-5}\cmidrule(lr){6-8}\cmidrule(lr){9-11} &  &  &  &  &  &  & 
&  &  &  \\ 
~ & $m$ & 300 & 500 & 1000 & 300 & 500 & 1000 & 300 & 500 & 1000 \\ 
\multicolumn{1}{c}{$\eta$} & \multicolumn{3}{c}{} & \multicolumn{3}{c}{} & 
\multicolumn{3}{c}{} &  \\ 
\multicolumn{2}{c}{} & \multicolumn{3}{c}{} & \multicolumn{3}{c}{} & 
\multicolumn{3}{c}{} \\ 
0.51 & ~ & 0.063 & 0.054 & 0.055 & 0.070 & 0.060 & 0.062 & 0.058 & 0.051 & 
0.051 \\ 
0.55 & ~ & 0.050 & 0.037 & 0.031 & 0.078 & 0.069 & 0.058 & 0.081 & 0.076 & 
0.069 \\ 
0.65 & ~ & 0.012 & 0.007 & 0.002 & 0.059 & 0.049 & 0.032 & 0.100 & 0.084 & 
0.073 \\ 
0.75 & ~ & 0.004 & 0.003 & 0.000 & 0.044 & 0.040 & 0.028 & 0.100 & 0.079 & 
0.070 \\ 
0.85 & ~ & 0.004 & 0.001 & 0.000 & 0.044 & 0.041 & 0.028 & 0.100 & 0.078 & 
0.070 \\ 
1 & ~ & 0.004 & 0.002 & 0.000 & 0.043 & 0.040 & 0.027 & 0.091 & 0.072 & 0.065
\\ 
\multicolumn{11}{c}{} \\ 
\midrule \bottomrule &  &  &  &  &  &  &  &  &  & 
\end{tabular}
}}
\par
{\scriptsize \ }
\par
{\scriptsize {\footnotesize \ } }
\par
{\scriptsize {\footnotesize 
\begin{tablenotes}
      \tiny
            \item The table contains the empirical rejection frequencies under the null of no changepoint for different sample sizes and different values of $\eta$, for the case of the static regression model (\ref{dgp}), and $d=8$ - i.e. 7 exogenous regressors and the constant. The other specifications of the model are the same as for Table \ref{tab:Size1}. 
            
\end{tablenotes}
} }
\end{table*}

\clearpage

\newpage

\subsection{Detection delays\label{delays_further}}

\begin{table*}[b]
\caption{{\protect\footnotesize {Median delays under an early occurring
changepoint - dynamic regression}}}
\label{tab:DelayE2}\centering
\par
{\scriptsize {\ }}
\par
{\scriptsize 
\begin{tabular}{llllllllllllllll}
\multicolumn{16}{c}{} \\ 
$a_m$ & ~ & $\eta$ & 0 & 0.15 & 0.25 & 0.35 & 0.45 & 0.49 & 0.5 & 0.51 & 0.55
& 0.65 & 0.75 & 0.85 & 1 \\ 
\multicolumn{16}{c}{} \\ 
\multicolumn{16}{c}{\ ~~~~~~~~$m=300$~~~~~} \\ 
\multicolumn{16}{c}{} \\ 
$\ln \ln m$ & ~ & ~ & 19 & 12 & 8 & 5 & 3 & 2 & 2 & 2 & 2 & 2 & 2 & 1 & 1 \\ 
$\ln m$ & ~ & ~ & 20 & 13 & 9 & 7 & 5 & 4.5 & 5 & 4 & 4 & 4 & 3 & 3 & 3 \\ 
$\ln^ m$ & ~ & ~ & 21 & 16 & 13 & 11 & 10 & 10 & 11 & 10 & 9 & 8 & 8 & 7 & 7
\\ 
\multicolumn{16}{c}{} \\ \hline
\multicolumn{16}{c}{} \\ 
\multicolumn{16}{c}{\ ~~~~~~~~$m=500$~~~~~} \\ 
\multicolumn{16}{c}{} \\ 
$\ln \ln m$ & ~ & ~ & 25 & 15 & 10 & 6 & 3 & 2 & 3 & 2 & 2 & 2 & 2 & 1 & 1
\\ 
$\ln m$ & ~ & ~ & 25 & 15 & 11 & 7 & 5 & 5 & 5 & 5 & 4 & 4 & 4 & 4 & 4 \\ 
$\ln^2 m$ & ~ & ~ & 26 & 19 & 15 & 13 & 11 & 11 & 12 & 11 & 10 & 8 & 8 & 8 & 
8 \\ 
\multicolumn{16}{c}{} \\ \hline
\multicolumn{16}{c}{} \\ 
\multicolumn{16}{c}{\ ~~~~~~~~$m=1000$~~~~~} \\ 
\multicolumn{16}{c}{} \\ 
$\ln \ln m$ & ~ & ~ & 35 & 20 & 12 & 7 & 3 & 3 & 3 & 2 & 2 & 2 & 2 & 2 & 1
\\ 
$\ln m$ & ~ & ~ & 35 & 20 & 13 & 8 & 6 & 5 & 5 & 5 & 4 & 4 & 4 & 4 & 4 \\ 
$\ln^2 m$ & ~ & ~ & 36 & 25 & 19 & 15 & 12 & 12 & 13 & 12 & 11 & 10 & 9 & 9
& 9 \\ 
\multicolumn{16}{c}{} \\ 
\bottomrule \bottomrule &  &  &  &  &  &  &  &  &  &  &  &  &  &  & 
\end{tabular}
}
\par
{\scriptsize \ }
\par
{\scriptsize {\footnotesize 
\begin{tablenotes}
      \tiny
            \item The table contains the median delay under the alternative of an \textit{early} occurring changepoint, for different sample sizes, different values of $\eta$, and different trimming sequence $a_m$, in the case of the dynamic regression model (\ref{dgp}), with $d=4$ - i.e., 3 exogenous regressors and the constant. The specifications of the model are described in the main text.
            
\end{tablenotes}
} }
\end{table*}

\medskip

\begin{table*}[h!]
\caption{{\protect\footnotesize {Median delays under an early occurring
changepoint - dynamic regression}}}
\label{tab:DelayE3}\centering
\par
{\scriptsize {\ }}
\par
{\scriptsize 
\begin{tabular}{llllllllllllllll}
\multicolumn{16}{c}{} \\ 
$a_m$ & ~ & $\eta$ & 0 & 0.15 & 0.25 & 0.35 & 0.45 & 0.49 & 0.5 & 0.51 & 0.55
& 0.65 & 0.75 & 0.85 & 1 \\ 
\multicolumn{16}{c}{} \\ 
\multicolumn{16}{c}{\ ~~~~~~~~$m=300$~~~~~} \\ 
\multicolumn{16}{c}{} \\ 
$\ln \ln m$ & ~ & ~ & 14 & 8 & 5 & 3 & 2 & 2 & 2 & 1 & 1 & 1 & 1 & 1 & 1 \\ 
$\ln m$ & ~ & ~ & 14 & 9 & 6 & 5 & 3 & 3 & 3 & 3 & 3 & 2 & 2 & 2 & 2 \\ 
$\ln^2 m$ & ~ & ~ & 16 & 12 & 10 & 8 & 7 & 7 & 8 & 7 & 7 & 6 & 6 & 5 & 5 \\ 
\multicolumn{16}{c}{} \\ \hline
\multicolumn{16}{c}{} \\ 
\multicolumn{16}{c}{\ ~~~~~~~~$m=500$~~~~~} \\ 
\multicolumn{16}{c}{} \\ 
$\ln \ln m$ & ~ & ~ & 18 & 10 & 6 & 3 & 2 & 2 & 2 & 1 & 1 & 1 & 1 & 1 & 1 \\ 
$\ln m$ & ~ & ~ & 18 & 11 & 8 & 5 & 4 & 4 & 4 & 3 & 3 & 3 & 3 & 3 & 3 \\ 
$\ln^2 m$ & ~ & ~ & 19 & 14 & 11 & 9 & 8 & 8 & 8 & 8 & 7 & 6 & 6 & 6 & 6 \\ 
\multicolumn{16}{c}{} \\ \hline
\multicolumn{16}{c}{} \\ 
\multicolumn{16}{c}{\ ~~~~~~~~$m=1000$~~~~~} \\ 
\multicolumn{16}{c}{} \\ 
$\ln \ln m$ & ~ & ~ & 26 & 14 & 8 & 4 & 2 & 2 & 2 & 1 & 1 & 1 & 1 & 1 & 1 \\ 
$\ln m$ & ~ & ~ & 26 & 15 & 10 & 6 & 4 & 3 & 4 & 3 & 3 & 3 & 3 & 3 & 3 \\ 
$\ln^2 m$ & ~ & ~ & 28 & 18 & 14 & 11 & 9 & 9 & 10 & 9 & 8 & 7 & 7 & 6 & 6
\\ 
\multicolumn{16}{c}{} \\ 
\bottomrule \bottomrule &  &  &  &  &  &  &  &  &  &  &  &  &  &  & 
\end{tabular}
}
\par
{\scriptsize {\footnotesize 
\begin{tablenotes}
      \tiny
            \item The table contains the median delay under the alternative of an \textit{early} occurring changepoint, for different sample sizes, different values of $\eta$, and different trimming sequence $a_m$, in the case of the dynamic regression model (\ref{dgp}), with $d=8$ - i.e., 7 exogenous regressors and the constant. The specifications of the model are described in the main text.
            
\end{tablenotes}
} }
\end{table*}

\medskip

\begin{table*}[h!]
\caption{{\protect\footnotesize {Median delays under a late occurring
changepoint - dynamic regression}}}
\label{tab:DelayL2}\centering
\par
{\scriptsize {\ 
\begin{tabular}{llllllllllllllll}
\multicolumn{16}{c}{} \\ 
$a_m$ & ~ & $\eta$ & 0 & 0.15 & 0.25 & 0.35 & 0.45 & 0.49 & 0.5 & 0.51 & 0.55
& 0.65 & 0.75 & 0.85 & 1 \\ 
\multicolumn{16}{c}{} \\ 
\multicolumn{16}{c}{\ ~~~~~~~~$m=300$~~~~~} \\ 
\multicolumn{16}{c}{} \\ 
$\ln \ln m$ & ~ & ~ & 25 & 21 & 19 & 18 & 18 & 19 & 20 & 19 & 21 & 31 & 53 & 
103 & 163.5 \\ 
$\ln m$ & ~ & ~ & 26 & 22 & 20 & 18 & 18 & 19 & 20 & 19 & 20 & 24 & 32 & 45
& 88 \\ 
$\ln^2 m$ & ~ & ~ & 28 & 24 & 22 & 21 & 21 & 22 & 23 & 22 & 21 & 21 & 23 & 26
& 31 \\ 
\multicolumn{16}{c}{} \\ \hline
\multicolumn{16}{c}{} \\ 
\multicolumn{16}{c}{\ ~~~~~~~~$m=500$~~~~~} \\ 
\multicolumn{16}{c}{} \\ 
$\ln \ln m$ & ~ & ~ & 29 & 23 & 21 & 18 & 17 & 18 & 19 & 18 & 20 & 30 & 52 & 
103 & 305.5 \\ 
$\ln m$ & ~ & ~ & 30 & 24 & 21 & 18 & 17 & 18 & 19 & 18 & 19 & 23 & 30 & 41
& 77 \\ 
$\ln^2 m$ & ~ & ~ & 31 & 26 & 23 & 21 & 20 & 21 & 23 & 21 & 20 & 20 & 22 & 24
& 29 \\ 
\multicolumn{16}{c}{} \\ \hline
\multicolumn{16}{c}{} \\ 
\multicolumn{16}{c}{\ ~~~~~~~~$m=1000$~~~~~} \\ 
\multicolumn{16}{c}{} \\ 
$\ln \ln m$ & ~ & ~ & 38 & 28 & 23 & 19 & 17 & 17 & 18.5 & 18 & 19 & 29 & 50
& 99 & 599 \\ 
$\ln m $ & ~ & ~ & 38 & 28 & 24 & 20 & 17 & 18 & 19 & 18 & 18 & 22 & 29 & 41
& 77 \\ 
$\ln^2 m$ & ~ & ~ & 39 & 30 & 26 & 22 & 20 & 21 & 22 & 21 & 20 & 20 & 21 & 23
& 27 \\ 
\multicolumn{16}{c}{} \\ 
\bottomrule \bottomrule &  &  &  &  &  &  &  &  &  &  &  &  &  &  & 
\end{tabular}
}}
\par
{\scriptsize \ }
\par
{\scriptsize {\footnotesize 
\begin{tablenotes}
      \tiny
            \item The table contains the median delay under the alternative of a \textit{late} occurring changepoint, for different sample sizes, different values of $\eta$, and different trimming sequence $a_m$, in the case of the dynamic regression model (\ref{dgp}), with $d=4$ - i.e., 3 exogenous regressors and the constant. The specifications of the model are described in the main text.
            
\end{tablenotes}
} }
\end{table*}

\medskip

\begin{table*}[h!]
\caption{{\protect\footnotesize {Median delays under a late occurring
changepoint - dynamic regression}}}
\label{tab:DelayL3}\centering
{\footnotesize {\ }}
\par
{\scriptsize {\ 
\begin{tabular}{llllllllllllllll}
\multicolumn{16}{c}{} \\ 
$a_m$ & ~ & $\eta$ & 0 & 0.15 & 0.25 & 0.35 & 0.45 & 0.49 & 0.5 & 0.51 & 0.55
& 0.65 & 0.75 & 0.85 & 1 \\ 
\multicolumn{16}{c}{} \\ 
\multicolumn{16}{c}{\ ~~~~~~~~$m=300$~~~~~} \\ 
\multicolumn{16}{c}{} \\ 
$\ln \ln m$ & ~ & ~ & 18 & 16 & 14 & 13 & 13 & 13 & 14 & 14 & 15 & 22 & 36 & 
64 & 155 \\ 
$\ln m$ & ~ & ~ & 18 & 16 & 14 & 13 & 13 & 14 & 14 & 14 & 14 & 17 & 22 & 30.5
& 54 \\ 
$\ln^2 m$ & ~ & ~ & 20 & 17 & 16 & 15 & 15 & 16 & 17 & 16 & 15 & 15 & 17 & 18
& 22 \\ 
\multicolumn{16}{c}{} \\ \hline
\multicolumn{16}{c}{} \\ 
\multicolumn{16}{c}{\ ~~~~~~~~$m=500$~~~~~} \\ 
\multicolumn{16}{c}{} \\ 
$\ln \ln m$ & ~ & ~ & 22 & 18 & 16 & 14 & 13 & 13 & 14 & 14 & 15 & 22 & 35 & 
64 & 248 \\ 
$\ln m$ & ~ & ~ & 22 & 18 & 16 & 14 & 13 & 14 & 15 & 14 & 14 & 17 & 21 & 29
& 49 \\ 
$\ln^2 m$ & ~ & ~ & 23 & 19 & 17 & 16 & 15 & 16 & 17 & 16 & 15 & 15 & 16 & 18
& 21 \\ 
\multicolumn{16}{c}{} \\ \hline
\multicolumn{16}{c}{} \\ 
\multicolumn{16}{c}{\ ~~~~~~~~$m=1000$~~~~~} \\ 
\multicolumn{16}{c}{} \\ 
$\ln \ln m$ & ~ & ~ & 29 & 21 & 18 & 14 & 13 & 13 & 14 & 13 & 14 & 22 & 35 & 
64 & 346 \\ 
$\ln m$ & ~ & ~ & 29 & 21 & 18 & 15 & 13 & 13 & 14 & 13 & 13 & 16 & 21 & 29
& 50 \\ 
$\ln^2 m$ & ~ & ~ & 30 & 23 & 20 & 17 & 16 & 16 & 17 & 16 & 15 & 15 & 16 & 17
& 20 \\ 
\multicolumn{16}{c}{} \\ 
\bottomrule \bottomrule &  &  &  &  &  &  &  &  &  &  &  &  &  &  & 
\end{tabular}
} }
\par
{\scriptsize {\footnotesize 
\begin{tablenotes}
      \tiny
            \item The table contains the median delay under the alternative of a \textit{late} occurring changepoint, for different sample sizes, different values of $\eta$, and different trimming sequence $a_m$, in the case of the dynamic regression model (\ref{dgp}), with $d=8$ - i.e., 7 exogenous regressors and the constant. The specifications of the model are described in the main text.
            
\end{tablenotes}
} }
\end{table*}

\medskip

\begin{table*}[h!]
\caption{{\protect\footnotesize {Delays under an early occurring changepoint
- dynamic regression, full descriptive statistics}}}
\label{tab:DelayEF1}\centering
{\footnotesize {\ }}
\par
{\scriptsize {\ }}
\par
{\scriptsize 
\begin{tabular}{llllllllllllllll}
\multicolumn{16}{c}{Summary Statistics for Delays, Early Break, $m=500$, $%
d=2 $} \\ 
\multicolumn{16}{c}{} \\ 
$a_m$ & ~ & $\eta$ & 0 & 0.15 & 0.25 & 0.35 & 0.45 & 0.49 & 0.5 & 0.51 & 0.55
& 0.65 & 0.75 & 0.85 & 1 \\ 
\multicolumn{16}{c}{} \\ 
$\ln \ln m$ & Min & ~ & 11 & 6 & 3 & 1 & 1 & 1 & 1 & 1 & 1 & 1 & 1 & 1 & 1
\\ 
~ & Q1 & ~ & 29 & 17 & 10 & 6 & 3 & 2 & 2 & 2 & 2 & 1 & 1 & 1 & 1 \\ 
~ & MED & ~ & 35 & 22 & 15 & 9 & 5 & 4 & 4 & 4 & 3 & 3 & 3 & 3 & 1 \\ 
~ & ARL & ~ & 36.25 & 23.19 & 16.31 & 10.72 & 6.97 & 6.15 & 6.42 & 5.84 & 
5.49 & 5.99 & 8.05 & 15.75 & 4.03 \\ 
~ & Q3 & ~ & 42 & 28 & 21 & 14 & 9 & 8 & 9 & 8 & 7 & 8 & 10 & 16 & 3 \\ 
~ & Max & ~ & 103 & 83 & 70 & 66 & 65 & 65 & 66 & 66 & 66 & 84 & 123 & 430 & 
437 \\ 
\multicolumn{16}{c}{} \\ \hline
\multicolumn{16}{c}{} \\ 
$\ln m$ & Min & ~ & 13 & 5 & 3 & 2 & 1 & 1 & 1 & 1 & 1 & 1 & 1 & 1 & 1 \\ 
~ & Q1 & ~ & 29 & 17 & 12 & 8 & 5 & 5 & 5 & 4 & 4 & 3 & 3 & 3 & 3 \\ 
~ & MED & ~ & 35 & 22 & 16 & 11 & 8 & 8 & 8 & 7 & 7 & 6 & 6 & 6 & 6.5 \\ 
~ & ARL & ~ & 36.61 & 24.18 & 17.92 & 12.89 & 9.82 & 9.32 & 9.76 & 8.93 & 
8.23 & 7.84 & 8.09 & 8.69 & 11.49 \\ 
~ & Q3 & ~ & 43 & 29 & 22 & 16 & 13 & 12 & 13 & 11 & 11 & 10 & 10 & 11 & 13
\\ 
~ & Max & ~ & 99 & 83 & 76 & 65 & 63 & 63 & 64 & 63 & 63 & 65 & 83 & 104 & 
310 \\ 
\multicolumn{16}{c}{} \\ \hline
\multicolumn{16}{c}{} \\ 
$\ln^2 m$ & Min & ~ & 11 & 5 & 3 & 2 & 1 & 1 & 1 & 1 & 1 & 1 & 1 & 1 & 1 \\ 
~ & Q1 & ~ & 31 & 21 & 17 & 13 & 11 & 11 & 12 & 11 & 9 & 8 & 8 & 7 & 7 \\ 
~ & MED & ~ & 38 & 28 & 23 & 19 & 16 & 16 & 17 & 16 & 14 & 13 & 12 & 12 & 12
\\ 
~ & ARL & ~ & 39.32 & 29.26 & 24.13 & 19.95 & 17.51 & 17.65 & 18.75 & 17.20
& 15.61 & 14.23 & 13.94 & 13.70 & 13.92 \\ 
~ & Q3 & ~ & 46 & 35 & 30 & 25 & 22 & 23 & 24 & 22 & 20 & 19 & 19 & 19 & 19
\\ 
~ & Max & ~ & 90 & 80 & 73 & 63 & 60 & 62 & 67 & 62 & 60 & 59 & 60 & 60 & 70
\\ 
\multicolumn{16}{c}{} \\ 
\bottomrule \bottomrule &  &  &  &  &  &  &  &  &  &  &  &  &  &  & 
\end{tabular}
}
\par
{\scriptsize \ }
\par
{\scriptsize {\footnotesize 
\begin{tablenotes}
      \tiny
            \item The table contains several descriptive statistics for the delay under the alternative of an \textit{early} occurring changepoint, for $m=500$ and different values of $\eta$ and different trimming sequence $a_m$, in the case of the dynamic regression model (\ref{dgp}), with $d=2$ - i.e., one exogenous regressor and the constant. The specifications of the model are described in the main text. 
            \item The descriptive statistics reported in the table are: the minimum and maximum values ($\min$ and $\max$ respectively), the first and third quartiles (Q1 and Q3 respectively), the average delay (denoted as ARL), and the median delay (denoted as MED).
            
\end{tablenotes}
} }
\end{table*}

\medskip

\begin{table*}[h!]
\caption{{\protect\footnotesize {Delays under an early occurring changepoint
- dynamic regression, full descriptive statistics}}}
\label{tab:DelayEF2}\centering
{\footnotesize {\ }}
\par
{\scriptsize {\ }}
\par
{\scriptsize 
\begin{tabular}{llllllllllllllll}
\multicolumn{16}{c}{} \\ 
$a_m$ & ~ & $\eta$ & 0 & 0.15 & 0.25 & 0.35 & 0.45 & 0.49 & 0.5 & 0.51 & 0.55
& 0.65 & 0.75 & 0.85 & 1 \\ 
\multicolumn{16}{c}{} \\ 
$\ln \ln m$ & Min & ~ & 8 & 3 & 2 & 1 & 1 & 1 & 1 & 1 & 1 & 1 & 1 & 1 & 1 \\ 
~ & Q1 & ~ & 19 & 10 & 6 & 3 & 2 & 1 & 1 & 1 & 1 & 1 & 1 & 1 & 1 \\ 
~ & MED & ~ & 25 & 15 & 10 & 6 & 3 & 2 & 3 & 2 & 2 & 2 & 2 & 1 & 1 \\ 
~ & ARL & ~ & 26.32 & 16.35 & 11.47 & 7.41 & 4.77 & 4.13 & 4.35 & 3.93 & 3.69
& 3.79 & 4.31 & 5.53 & 12.08 \\ 
~ & Q3 & ~ & 32 & 21 & 15 & 10 & 6 & 5 & 6 & 5 & 4 & 4 & 4 & 5 & 5 \\ 
~ & Max & ~ & 81 & 53 & 51 & 45 & 42 & 42 & 44 & 42 & 44 & 47 & 70 & 130 & 
481 \\ 
\multicolumn{16}{c}{} \\ \hline
\multicolumn{16}{c}{} \\ 
$\ln m$ & Min & ~ & 8 & 4 & 2 & 1 & 1 & 1 & 1 & 1 & 1 & 1 & 1 & 1 & 1 \\ 
~ & Q1 & ~ & 19 & 11 & 7 & 5 & 3 & 3 & 3 & 3 & 3 & 2 & 2 & 2 & 2 \\ 
~ & MED & ~ & 25 & 15 & 11 & 7 & 5 & 5 & 5 & 5 & 4 & 4 & 4 & 4 & 4 \\ 
~ & ARL & ~ & 26.36 & 17.03 & 12.61 & 9.20 & 7.03 & 6.60 & 6.93 & 6.38 & 5.89
& 5.55 & 5.59 & 5.74 & 6.47 \\ 
~ & Q3 & ~ & 32 & 21 & 16 & 12 & 9 & 9 & 9 & 8 & 8 & 7 & 7 & 7 & 8 \\ 
~ & Max & ~ & 80 & 55 & 47 & 41 & 38 & 38 & 39 & 38 & 37 & 34 & 41 & 43 & 109
\\ 
\multicolumn{16}{c}{} \\ \hline
\multicolumn{16}{c}{} \\ 
$\ln^2 m$ & Min & ~ & 8 & 4 & 1 & 1 & 1 & 1 & 1 & 1 & 1 & 1 & 1 & 1 & 1 \\ 
~ & Q1 & ~ & 21 & 14 & 11 & 8.25 & 7 & 7 & 8 & 7 & 6 & 5 & 5 & 5 & 5 \\ 
~ & MED & ~ & 26 & 19 & 15 & 13 & 11 & 11 & 12 & 11 & 10 & 8 & 8 & 8 & 8 \\ 
~ & ARL & ~ & 28.19 & 20.86 & 17.19 & 14.19 & 12.46 & 12.49 & 13.22 & 12.17
& 11.14 & 10.14 & 9.83 & 9.59 & 9.61 \\ 
~ & Q3 & ~ & 34 & 26 & 21 & 18 & 16 & 16 & 17 & 16 & 15 & 13 & 13 & 13 & 13
\\ 
~ & Max & ~ & 82 & 73 & 69 & 58 & 58 & 58 & 64 & 58 & 51 & 50 & 50 & 51 & 66
\\ 
\multicolumn{16}{c}{} \\ 
\bottomrule \bottomrule &  &  &  &  &  &  &  &  &  &  &  &  &  &  & 
\end{tabular}
}
\par
{\scriptsize \ }
\par
{\scriptsize {\footnotesize 
\begin{tablenotes}
      \tiny
            \item The table contains several descriptive statistics for the delay under the alternative of an \textit{early} occurring changepoint, for $m=500$ and different values of $\eta$ and different trimming sequence $a_m$, in the case of the dynamic regression model (\ref{dgp}), with $d=4$ - i.e., 3 exogenous regressors and the constant. The specifications of the model are described in the main text; see the notes to Table \ref{tab:DelayEF1} for an explanation of the descriptive statistics. 
                        
\end{tablenotes}
} }
\end{table*}

\medskip

\begin{table*}[h!]
\caption{{\protect\footnotesize {Delays under an early occurring changepoint
- dynamic regression, full descriptive statistics}}}
\label{tab:DelayEF3}\centering
{\footnotesize {\ }}
\par
{\scriptsize {\ }}
\par
{\scriptsize 
\begin{tabular}{llllllllllllllll}
\multicolumn{16}{c}{} \\ 
$a_m$ & ~ & $\eta$ & 0 & 0.15 & 0.25 & 0.35 & 0.45 & 0.49 & 0.5 & 0.51 & 0.55
& 0.65 & 0.75 & 0.85 & 1 \\ 
\multicolumn{16}{c}{} \\ 
$\ln \ln m$ & Min & ~ & 5 & 2 & 1 & 1 & 1 & 1 & 1 & 1 & 1 & 1 & 1 & 1 & 1 \\ 
~ & Q1 & ~ & 14 & 7 & 4 & 2 & 1 & 1 & 1 & 1 & 1 & 1 & 1 & 1 & 1 \\ 
~ & MED & ~ & 18 & 10 & 6 & 3 & 2 & 2 & 2 & 1 & 1 & 1 & 1 & 1 & 1 \\ 
~ & ARL & ~ & 20.06 & 12.17 & 8.34 & 5.29 & 3.32 & 2.91 & 2.99 & 2.75 & 2.56
& 2.53 & 2.74 & 3.12 & 6.01 \\ 
~ & Q3 & ~ & 24 & 15 & 11 & 7 & 4 & 3 & 4 & 3 & 3 & 3 & 3 & 3 & 3 \\ 
~ & Max & ~ & 80 & 78 & 50 & 45 & 31 & 31 & 31 & 31 & 31 & 50 & 52 & 58 & 415
\\ 
\multicolumn{16}{c}{} \\ \hline
\multicolumn{16}{c}{} \\ 
$\ln m$ & Min & ~ & 5 & 3 & 1 & 1 & 1 & 1 & 1 & 1 & 1 & 1 & 1 & 1 & 1 \\ 
~ & Q1 & ~ & 14 & 8 & 5 & 3 & 2 & 2 & 2 & 2 & 2 & 2 & 1 & 1 & 1 \\ 
~ & MED & ~ & 18 & 11 & 8 & 5 & 4 & 4 & 4 & 3 & 3 & 3 & 3 & 3 & 3 \\ 
~ & ARL & ~ & 20.27 & 13.02 & 9.64 & 6.91 & 5.27 & 4.93 & 5.16 & 4.78 & 4.35
& 4.01 & 3.99 & 4.03 & 4.28 \\ 
~ & Q3 & ~ & 25 & 16 & 12 & 9 & 7 & 6 & 6 & 6 & 5 & 5 & 5 & 5 & 5 \\ 
~ & Max & ~ & 79 & 74 & 71 & 61 & 61 & 61 & 61 & 61 & 41 & 42 & 47 & 73 & 100
\\ 
\multicolumn{16}{c}{} \\ \hline
\multicolumn{16}{c}{} \\ 
$\ln^2 m$ & Min & ~ & 3 & 2 & 1 & 1 & 1 & 1 & 1 & 1 & 1 & 1 & 1 & 1 & 1 \\ 
~ & Q1 & ~ & 14 & 10 & 8 & 6 & 5 & 5 & 5 & 5 & 4 & 4 & 3 & 3 & 3 \\ 
~ & MED & ~ & 19 & 14 & 11 & 9 & 8 & 8 & 8 & 8 & 7 & 6 & 6 & 6 & 6 \\ 
~ & ARL & ~ & 21.38 & 15.96 & 13.29 & 11.07 & 9.86 & 9.87 & 10.30 & 9.57 & 
8.65 & 7.86 & 7.60 & 7.36 & 7.31 \\ 
~ & Q3 & ~ & 26 & 20 & 17 & 14 & 13 & 13 & 13 & 12 & 11 & 10 & 10 & 10 & 10
\\ 
~ & Max & ~ & 73 & 68 & 66 & 62 & 62 & 62 & 62 & 62 & 62 & 57 & 62 & 62 & 63
\\ 
\multicolumn{16}{c}{} \\ 
\bottomrule \bottomrule &  &  &  &  &  &  &  &  &  &  &  &  &  &  & 
\end{tabular}
}
\par
{\scriptsize \ }
\par
{\scriptsize {\footnotesize 
\begin{tablenotes}
      \tiny
            \item The table contains several descriptive statistics for the delay under the alternative of an \textit{early} occurring changepoint, for $m=500$ and different values of $\eta$ and different trimming sequence $a_m$, in the case of the dynamic regression model (\ref{dgp}), with $d=8$ - i.e., 7 exogenous regressors and the constant. The specifications of the model are described in the main text; see the notes to Table \ref{tab:DelayEF1} for an explanation of the descriptive statistics.             
\end{tablenotes}
} }
\end{table*}

\medskip

\begin{table*}[h!]
\caption{{\protect\footnotesize {Delays under a late occurring changepoint -
dynamic regression, full descriptive statistics}}}
\label{tab:DelayLF1}\centering
{\footnotesize {\ }}
\par
{\scriptsize {\ 
\begin{tabular}{llllllllllllllll}
\multicolumn{16}{c}{} \\ 
$a_m$ & ~ & $\eta$ & 0 & 0.15 & 0.25 & 0.35 & 0.45 & 0.49 & 0.5 & 0.51 & 0.55
& 0.65 & 0.75 & 0.85 & 1 \\ 
\multicolumn{16}{c}{} \\ 
$\ln \ln m$ & Min & ~ & 11 & 6 & 3 & 2 & 1 & 1 & 2 & 2 & 3 & 9 & 25 & 53 & NA
\\ 
~ & Q1 & ~ & 34 & 26 & 22 & 19 & 18 & 19 & 20 & 19 & 22 & 36 & 66 & 156 & NA
\\ 
~ & MED & ~ & 42 & 34 & 30 & 27 & 25 & 26 & 28 & 27 & 30 & 46 & 84 & 202 & NA
\\ 
~ & ARL & ~ & 43.81 & 35.52 & 31.58 & 28.20 & 26.80 & 27.82 & 29.66 & 28.76
& 31.59 & 48.18 & 87.92 & 211.47 & NA \\ 
~ & Q3 & ~ & 52 & 43 & 39 & 35 & 34 & 35 & 37 & 36 & 39 & 59 & 105 & 259 & NA
\\ 
~ & Max & ~ & 110 & 101 & 94 & 84 & 83 & 93 & 96 & 94 & 101 & 128 & 221 & 400
& NA \\ 
\multicolumn{16}{c}{} \\ \hline
\multicolumn{16}{c}{} \\ 
$\ln m$ & Min & ~ & 7 & 5 & 3 & 1 & 1 & 1 & 1 & 1 & 1 & 4 & 8 & 14 & 36 \\ 
~ & Q1 & ~ & 34 & 26 & 23 & 20 & 18 & 19 & 21 & 20 & 20 & 25 & 36 & 51 & 115
\\ 
~ & MED & ~ & 43 & 35 & 31 & 27 & 26 & 27 & 29 & 27 & 28 & 35 & 46.5 & 66 & 
152 \\ 
~ & ARL & ~ & 44.27 & 36.15 & 32.25 & 28.83 & 27.45 & 28.57 & 30.56 & 29.01
& 29.56 & 36.24 & 48.93 & 69.85 & 164.58 \\ 
~ & Q3 & ~ & 53 & 44 & 40 & 36 & 35 & 36 & 38 & 36 & 37 & 45 & 60 & 85 & 
202.5 \\ 
~ & Max & ~ & 106 & 99 & 96 & 91 & 91 & 95 & 98 & 95 & 96 & 106 & 136 & 194
& 395 \\ 
\multicolumn{16}{c}{} \\ \hline
\multicolumn{16}{c}{} \\ 
$\ln^2 m$ & Min & ~ & 8 & 2 & 1 & 1 & 1 & 1 & 2 & 1 & 1 & 1 & 1 & 3 & 2 \\ 
~ & Q1 & ~ & 35 & 28 & 25 & 22 & 21 & 22 & 24 & 22 & 21 & 22 & 24 & 27 & 33
\\ 
~ & MED & ~ & 45 & 37 & 34 & 31 & 30 & 31 & 33 & 31 & 29 & 30 & 33 & 37 & 45
\\ 
~ & ARL & ~ & 46.40 & 39.02 & 35.39 & 32.33 & 31.38 & 32.90 & 35.11 & 32.86
& 31.28 & 31.92 & 34.87 & 38.58 & 47.22 \\ 
~ & Q3 & ~ & 56 & 48 & 44 & 41 & 40 & 41 & 44 & 41 & 40 & 40 & 44 & 48 & 58
\\ 
~ & Max & ~ & 111 & 102 & 100 & 98 & 97 & 99 & 102 & 99 & 98 & 99 & 103 & 111
& 135 \\ 
\multicolumn{16}{c}{} \\ 
\bottomrule \bottomrule &  &  &  &  &  &  &  &  &  &  &  &  &  &  & 
\end{tabular}
}}
\par
{\scriptsize \ }
\par
{\scriptsize {\footnotesize 
\begin{tablenotes}
      \tiny
            \item The table contains several descriptive statistics for the delay under the alternative of a \textit{late} occurring changepoint, for $m=500$ and different values of $\eta$ and different trimming sequence $a_m$, in the case of the dynamic regression model (\ref{dgp}), with $d=2$ - i.e., one exogenous regressor and the constant. The specifications of the model are described in the main text; see the notes to Table \ref{tab:DelayEF1} for an explanation of the descriptive statistics.             

\end{tablenotes}
} }
\end{table*}

\medskip

\begin{table*}[h!]
\caption{{\protect\footnotesize {Delays under a late occurring changepoint -
dynamic regression, full descriptive statistics}}}
\label{tab:DelayLF2}\centering
{\footnotesize {\ }}
\par
{\scriptsize {\ }}
\par
{\scriptsize 
\begin{tabular}{llllllllllllllll}
\multicolumn{16}{c}{} \\ 
$a_m$ & ~ & $\eta$ & 0 & 0.15 & 0.25 & 0.35 & 0.45 & 0.49 & 0.5 & 0.51 & 0.55
& 0.65 & 0.75 & 0.85 & 1 \\ 
\multicolumn{16}{c}{} \\ 
$\ln \ln m$ & Min & ~ & 3 & 2 & 1 & 1 & 1 & 1 & 1 & 1 & 1 & 2 & 10 & 26 & 156
\\ 
~ & Q1 & ~ & 22 & 17 & 15 & 13 & 12 & 12 & 13 & 13 & 14 & 22 & 40 & 79 & 259
\\ 
~ & MED & ~ & 29 & 23 & 21 & 18 & 17 & 18 & 19 & 18 & 20 & 30 & 52 & 103 & 
305.5 \\ 
~ & ARL & ~ & 31.12 & 25.24 & 22.41 & 19.97 & 18.99 & 19.62 & 20.83 & 20.22
& 22.06 & 32.17 & 54.31 & 108.51 & 304.61 \\ 
~ & Q3 & ~ & 38 & 31 & 28 & 25 & 24 & 25 & 26 & 25 & 28 & 40 & 66 & 132 & 356
\\ 
~ & Max & ~ & 91 & 82 & 76 & 73 & 73 & 73 & 75 & 74 & 81 & 100 & 174 & 310 & 
400 \\ 
\multicolumn{16}{c}{} \\ \hline
\multicolumn{16}{c}{} \\ 
$\ln m$ & Min & ~ & 4 & 1 & 1 & 1 & 1 & 1 & 1 & 1 & 1 & 1 & 3 & 11 & 20 \\ 
~ & Q1 & ~ & 23 & 17 & 15 & 13 & 12 & 12 & 13 & 13 & 13 & 16 & 22 & 31 & 59
\\ 
~ & MED & ~ & 30 & 24 & 21 & 18 & 17 & 18 & 19 & 18 & 19 & 23 & 30 & 41 & 77
\\ 
~ & ARL & ~ & 31.41 & 25.48 & 22.63 & 20.15 & 19.20 & 19.82 & 21.18 & 20.09
& 20.43 & 24.57 & 32.30 & 43.90 & 82.97 \\ 
~ & Q3 & ~ & 39 & 32 & 28 & 26 & 24 & 25 & 27 & 26 & 26 & 31 & 41 & 54 & 102
\\ 
~ & Max & ~ & 87 & 76 & 74 & 68 & 67 & 72 & 74 & 73 & 73 & 86 & 98 & 142 & 
258 \\ 
\multicolumn{16}{c}{} \\ \hline
\multicolumn{16}{c}{} \\ 
$\ln^2 m$ & Min & ~ & 3 & 2 & 1 & 1 & 1 & 1 & 2 & 1 & 1 & 1 & 1 & 2 & 3 \\ 
~ & Q1 & ~ & 24 & 19 & 16 & 14 & 14 & 14.25 & 16 & 14 & 14 & 14 & 15 & 17 & 
21 \\ 
~ & MED & ~ & 31 & 26 & 23 & 21 & 20 & 21 & 23 & 21 & 20 & 20 & 22 & 24 & 29
\\ 
~ & ARL & ~ & 32.89 & 27.51 & 24.88 & 22.61 & 21.95 & 22.90 & 24.48 & 22.77
& 21.64 & 21.92 & 23.78 & 25.98 & 31.28 \\ 
~ & Q3 & ~ & 40 & 34 & 31 & 29 & 28 & 29 & 31 & 29 & 28 & 28 & 30 & 33 & 39
\\ 
~ & Max & ~ & 92 & 84 & 83 & 77 & 76 & 82 & 83 & 82 & 77 & 77 & 84 & 88 & 97
\\ 
\multicolumn{16}{c}{} \\ 
\bottomrule \bottomrule &  &  &  &  &  &  &  &  &  &  &  &  &  &  & 
\end{tabular}
}
\par
{\scriptsize {\footnotesize 
\begin{tablenotes}
      \tiny
            \item The table contains several descriptive statistics for the delay under the alternative of a \textit{late} occurring changepoint, for $m=500$ and different values of $\eta$ and different trimming sequence $a_m$, in the case of the dynamic regression model (\ref{dgp}), with $d=4$ - i.e., 3 exogenous regressors and the constant. The specifications of the model are described in the main text; see the notes to Table \ref{tab:DelayEF1} for an explanation of the descriptive statistics.             
\end{tablenotes}
} }
\end{table*}

\medskip

\begin{table*}[h!]
\caption{{\protect\footnotesize {Delays under a late occurring changepoint -
dynamic regression, full descriptive statistics}}}
\label{tab:DelayLF3}\centering
{\footnotesize {\ }}
\par
{\scriptsize {\ }}
\par
{\scriptsize 
\begin{tabular}{llllllllllllllll}
\multicolumn{16}{c}{} \\ 
$a_m$ & ~ & $\eta$ & 0 & 0.15 & 0.25 & 0.35 & 0.45 & 0.49 & 0.5 & 0.51 & 0.55
& 0.65 & 0.75 & 0.85 & 1 \\ 
\multicolumn{16}{c}{} \\ 
$\ln \ln m$ & Min & ~ & 4 & 1 & 1 & 1 & 1 & 1 & 1 & 1 & 1 & 3 & 7 & 16 & 65
\\ 
~ & Q1 & ~ & 16 & 12 & 11 & 9 & 8 & 9 & 9 & 9 & 10 & 15 & 26 & 49 & 186 \\ 
~ & MED & ~ & 22 & 18 & 16 & 14 & 13 & 13 & 14 & 14 & 15 & 22 & 35 & 64 & 248
\\ 
~ & ARL & ~ & 24.25 & 19.80 & 17.46 & 15.63 & 14.85 & 15.37 & 16.23 & 15.78
& 17.18 & 24.39 & 38.73 & 68.92 & 250.56 \\ 
~ & Q3 & ~ & 30 & 25 & 22 & 20 & 19 & 20 & 21 & 20 & 22 & 30 & 48 & 84 & 314
\\ 
~ & Max & ~ & 81 & 79 & 69 & 68 & 58 & 68 & 69 & 69 & 69 & 95 & 129 & 220 & 
400 \\ 
\multicolumn{16}{c}{} \\ \hline
\multicolumn{16}{c}{} \\ 
$\ln m$ & Min & ~ & 3 & 2 & 1 & 1 & 1 & 1 & 1 & 1 & 1 & 1 & 2 & 5 & 10 \\ 
~ & Q1 & ~ & 16 & 12 & 11 & 9 & 8 & 9 & 10 & 9 & 9 & 11 & 15 & 21 & 37 \\ 
~ & MED & ~ & 22 & 18 & 16 & 14 & 13 & 14 & 15 & 14 & 14 & 17 & 21 & 29 & 49
\\ 
~ & ARL & ~ & 24.48 & 20.02 & 17.86 & 16.06 & 15.26 & 15.77 & 16.70 & 15.98
& 16.19 & 19.19 & 24.33 & 32.08 & 54.21 \\ 
~ & Q3 & ~ & 30 & 25 & 23 & 21 & 20 & 21 & 22 & 21 & 21 & 24 & 31 & 40 & 67
\\ 
~ & Max & ~ & 90 & 78 & 65 & 64 & 64 & 64 & 65 & 64 & 65 & 80 & 93 & 124 & 
198 \\ 
\multicolumn{16}{c}{} \\ \hline
\multicolumn{16}{c}{} \\ 
$\ln^2 m$ & Min & ~ & 3 & 2 & 1 & 1 & 1 & 1 & 1 & 1 & 1 & 1 & 1 & 1 & 2 \\ 
~ & Q1 & ~ & 17 & 13 & 12 & 10 & 10 & 11 & 11 & 10 & 10 & 10 & 11 & 12 & 14
\\ 
~ & MED & ~ & 23 & 19 & 17 & 16 & 15 & 16 & 17 & 16 & 15 & 15 & 16 & 18 & 21
\\ 
~ & ARL & ~ & 25.58 & 21.35 & 19.42 & 17.81 & 17.23 & 17.95 & 18.97 & 17.82
& 17.01 & 17.22 & 18.50 & 20.04 & 23.54 \\ 
~ & Q3 & ~ & 32 & 27 & 24 & 23 & 22 & 23 & 24 & 23 & 22 & 22 & 23 & 25 & 30
\\ 
~ & Max & ~ & 111 & 90 & 89 & 89 & 89 & 89 & 90 & 89 & 89 & 89 & 90 & 111 & 
114 \\ 
\multicolumn{16}{c}{} \\ 
\bottomrule \bottomrule &  &  &  &  &  &  &  &  &  &  &  &  &  &  & 
\end{tabular}
}
\par
{\scriptsize \ }
\par
{\scriptsize {\footnotesize 
\begin{tablenotes}
      \tiny
            \item The table contains several descriptive statistics for the delay under the alternative of a \textit{late} occurring changepoint, for $m=500$ and different values of $\eta$ and different trimming sequence $a_m$, in the case of the dynamic regression model (\ref{dgp}), with $d=8$ - i.e., 7 exogenous regressors and the constant. The specifications of the model are described in the main text; see the notes to Table \ref{tab:DelayEF1} for an explanation of the descriptive statistics. 
                        
\end{tablenotes}
} }
\end{table*}

\medskip

\begin{table*}[h!]
\caption{{\protect\footnotesize {Delays under an early occurring changepoint
- static regression, full descriptive statistics}}}
\label{tab:DelayES1}\centering
{\footnotesize {\ }}
\par
{\footnotesize {\ }}
\par
{\scriptsize {\ }}
\par
{\scriptsize 
\begin{tabular}{llllllllllllllll}
\multicolumn{16}{c}{} \\ 
$a_m$ & ~ & $\eta$ & 0 & 0.15 & 0.25 & 0.35 & 0.45 & 0.49 & 0.5 & 0.51 & 0.55
& 0.65 & 0.75 & 0.85 & 1 \\ 
\multicolumn{16}{c}{} \\ 
$\ln \ln m$ & Min & ~ & 25 & 9 & 6 & 3 & 1 & 1 & 1 & 1 & 1 & 1 & 1 & 1 & 1
\\ 
~ & Q1 & ~ & 57 & 37 & 25 & 16 & 9 & 7 & 8 & 7 & 6 & 7 & 12 & 2 & 1 \\ 
~ & MED & ~ & 71 & 47 & 35 & 23 & 15 & 14 & 15 & 14 & 13 & 18 & 45.5 & 8 & 1
\\ 
~ & ARL & ~ & 72.90 & 50.66 & 38.21 & 27.42 & 19.58 & 18.40 & 19.73 & 18.08
& 18.35 & 27.62 & 80.74 & 62.63 & 1.48 \\ 
~ & Q3 & ~ & 85 & 62 & 48 & 35 & 26 & 24.25 & 26 & 24 & 24 & 38 & 117.75 & 
57.75 & 1 \\ 
~ & Max & ~ & 178 & 134 & 130 & 126 & 125 & 126 & 128 & 127 & 130 & 219 & 499
& 499 & 9 \\ 
\multicolumn{16}{c}{} \\ \hline
\multicolumn{16}{c}{} \\ 
$\ln m$ & Min & ~ & 25 & 14 & 5 & 2 & 1 & 1 & 1 & 1 & 1 & 1 & 1 & 1 & 1 \\ 
~ & Q1 & ~ & 58 & 38 & 27 & 18 & 12 & 12 & 12 & 11 & 10 & 9 & 9 & 10 & 7 \\ 
~ & MED & ~ & 71 & 49 & 37 & 27 & 20 & 18 & 20 & 18 & 16 & 16 & 19 & 23 & 18
\\ 
~ & ARL & ~ & 73.82 & 52.21 & 40.55 & 30.09 & 23.63 & 22.88 & 24.70 & 22.38
& 20.90 & 22.11 & 27.91 & 44.15 & 49.91 \\ 
~ & Q3 & ~ & 87 & 63 & 51 & 38 & 31 & 30 & 33 & 30 & 28 & 30 & 37 & 53 & 51
\\ 
~ & Max & ~ & 191 & 163 & 128 & 125 & 124 & 125 & 127 & 125 & 126 & 137 & 276
& 478 & 491 \\ 
\multicolumn{16}{c}{} \\ \hline
\multicolumn{16}{c}{} \\ 
$\ln^2 m$ & Min & ~ & 25 & 8 & 4 & 1 & 1 & 1 & 2 & 1 & 1 & 1 & 1 & 1 & 1 \\ 
~ & Q1 & ~ & 61 & 44 & 35 & 28 & 23 & 24 & 26 & 23 & 20 & 17 & 17 & 16 & 16
\\ 
~ & MED & ~ & 76 & 57 & 48 & 40 & 35 & 36 & 39 & 35 & 32 & 28 & 28 & 28 & 31
\\ 
~ & ARL & ~ & 78.96 & 60.40 & 51.00 & 42.74 & 38.24 & 39.11 & 42.07 & 38.19
& 34.59 & 31.95 & 32.45 & 33.13 & 37.38 \\ 
~ & Q3 & ~ & 94 & 74 & 64 & 55 & 50 & 51 & 55 & 50 & 46 & 43 & 44 & 45 & 50
\\ 
~ & Max & ~ & 208 & 180 & 178 & 165 & 164 & 177 & 179 & 177 & 165 & 177 & 180
& 208 & 257 \\ 
\multicolumn{16}{c}{} \\ 
\bottomrule \bottomrule &  &  &  &  &  &  &  &  &  &  &  &  &  &  & 
\end{tabular}
}
\par
{\scriptsize \ }
\par
{\scriptsize {\footnotesize \ } }
\par
{\scriptsize {\footnotesize 
\begin{tablenotes}
      \tiny
            \item The table contains several descriptive statistics for the delay under the alternative of an \textit{early} occurring changepoint, for $m=500$ and different values of $\eta$ and different trimming sequence $a_m$, in the case of the static regression model (\ref{dgp}), with $d=2$ - i.e., one exogenous regressor and the constant. The specifications of the model are described in the main text; see the notes to Table \ref{tab:DelayEF1} for an explanation of the descriptive statistics. 
            
\end{tablenotes}
} }
\end{table*}

\medskip

\begin{table*}[h!]
\caption{{\protect\footnotesize {Delays under an early occurring changepoint
- static regression, full descriptive statistics}}}
\label{tab:DelayES2}\centering
{\footnotesize {\ }}
\par
{\footnotesize {\ }}
\par
{\scriptsize {\ }}
\par
{\scriptsize 
\begin{tabular}{llllllllllllllll}
\multicolumn{16}{c}{} \\ 
$a_m$ & ~ & $\eta$ & 0 & 0.15 & 0.25 & 0.35 & 0.45 & 0.49 & 0.5 & 0.51 & 0.55
& 0.65 & 0.75 & 0.85 & 1 \\ 
\multicolumn{16}{c}{} \\ 
$\ln \ln m$ & Min & ~ & 18 & 8 & 4 & 2 & 1 & 1 & 1 & 1 & 1 & 1 & 1 & 1 & 1
\\ 
~ & Q1 & ~ & 38 & 23 & 15 & 9 & 4.75 & 4 & 4 & 3 & 3 & 3 & 3 & 2 & 1 \\ 
~ & MED & ~ & 48 & 31 & 22 & 14 & 9 & 7 & 8 & 7 & 6 & 7 & 10 & 14 & 1 \\ 
~ & ARL & ~ & 50.60 & 33.56 & 24.89 & 17.06 & 11.83 & 10.73 & 11.41 & 10.43
& 10.19 & 12.59 & 21.79 & 61.36 & 2.61 \\ 
~ & Q3 & ~ & 60 & 41 & 32 & 22 & 16 & 14.25 & 15 & 14 & 14 & 17 & 29 & 73.75
& 3 \\ 
~ & Max & ~ & 143 & 119 & 111 & 89 & 87 & 88 & 89 & 88 & 89 & 153 & 281 & 500
& 49 \\ 
\multicolumn{16}{c}{} \\ \hline
\multicolumn{16}{c}{} \\ 
$\ln m$ & Min & ~ & 16 & 9 & 4 & 2 & 1 & 1 & 1 & 1 & 1 & 1 & 1 & 1 & 1 \\ 
~ & Q1 & ~ & 39 & 24 & 16 & 10 & 7 & 7 & 7 & 6 & 5 & 5 & 5 & 5 & 5 \\ 
~ & MED & ~ & 49 & 32 & 24 & 17 & 12 & 11 & 12 & 11 & 10 & 9 & 10 & 10 & 12
\\ 
~ & ARL & ~ & 50.90 & 34.71 & 26.30 & 19.42 & 14.90 & 14.30 & 15.13 & 13.88
& 12.93 & 12.90 & 14.31 & 16.74 & 33.51 \\ 
~ & Q3 & ~ & 60 & 43 & 33 & 26 & 20 & 19 & 20 & 19 & 17 & 17 & 19 & 21 & 34
\\ 
~ & Max & ~ & 139 & 116 & 107 & 86 & 84 & 84 & 86 & 85 & 85 & 107 & 148 & 189
& 476 \\ 
\multicolumn{16}{c}{} \\ \hline
\multicolumn{16}{c}{} \\ 
$\ln^2 m$ & Min & ~ & 15 & 4 & 2 & 1 & 1 & 1 & 1 & 1 & 1 & 1 & 1 & 1 & 1 \\ 
~ & Q1 & ~ & 41 & 29 & 22 & 18 & 15 & 15 & 16 & 14 & 12 & 11 & 10 & 9 & 9 \\ 
~ & MED & ~ & 52 & 38 & 31 & 25 & 22 & 22 & 24 & 22 & 19 & 17 & 17 & 16 & 17
\\ 
~ & ARL & ~ & 53.94 & 40.48 & 33.74 & 28.18 & 24.96 & 25.24 & 26.98 & 24.56
& 22.11 & 20.20 & 19.82 & 19.69 & 20.56 \\ 
~ & Q3 & ~ & 65 & 50 & 42 & 36 & 32 & 33 & 35 & 32 & 29 & 27 & 27 & 26.25 & 
28 \\ 
~ & Max & ~ & 148 & 129 & 125 & 116 & 115 & 116 & 124 & 116 & 93 & 93 & 94 & 
107 & 113 \\ 
\multicolumn{16}{c}{} \\ 
\bottomrule \bottomrule &  &  &  &  &  &  &  &  &  &  &  &  &  &  & 
\end{tabular}
}
\par
{\scriptsize {\footnotesize \ } }
\par
{\scriptsize {\footnotesize 
\begin{tablenotes}
      \tiny
            \item The table contains several descriptive statistics for the delay under the alternative of an \textit{early} occurring changepoint, for $m=500$ and different values of $\eta$ and different trimming sequence $a_m$, in the case of the static regression model (\ref{dgp}), with $d=4$ - i.e., 3 exogenous regressors and the constant. The specifications of the model are described in the main text; see the notes to Table \ref{tab:DelayEF1} for an explanation of the descriptive statistics. 
            
\end{tablenotes}
} }
\end{table*}

\medskip

\begin{table*}[h!]
\caption{{\protect\footnotesize {Delays under an early occurring changepoint
- static regression, full descriptive statistics}}}
\label{tab:DelayES3}\centering
{\footnotesize {\ }}
\par
{\footnotesize {\ }}
\par
{\scriptsize {\ 
\begin{tabular}{llllllllllllllll}
\multicolumn{16}{c}{} \\ 
$a_m$ & ~ & $\eta$ & 0 & 0.15 & 0.25 & 0.35 & 0.45 & 0.49 & 0.5 & 0.51 & 0.55
& 0.65 & 0.75 & 0.85 & 1 \\ 
\multicolumn{16}{c}{} \\ 
$\ln \ln m$ & Min & ~ & 11 & 4 & 2 & 1 & 1 & 1 & 1 & 1 & 1 & 1 & 1 & 1 & 1
\\ 
~ & Q1 & ~ & 27 & 15 & 9 & 5 & 2 & 2 & 2 & 2 & 1 & 1 & 1 & 1 & 1 \\ 
~ & MED & ~ & 35 & 21 & 14 & 9 & 5 & 4 & 4 & 4 & 3 & 3 & 3 & 3 & 1 \\ 
~ & ARL & ~ & 37.27 & 23.90 & 17.12 & 11.47 & 7.72 & 6.94 & 7.29 & 6.60 & 
6.29 & 6.88 & 9.13 & 18.17 & 5.04 \\ 
~ & Q3 & ~ & 45 & 30 & 22 & 15 & 10 & 9 & 9 & 8 & 8 & 8 & 11 & 16 & 3 \\ 
~ & Max & ~ & 108 & 96 & 93 & 87 & 86 & 86 & 87 & 86 & 88 & 98 & 140 & 433 & 
207 \\ 
\multicolumn{16}{c}{} \\ \hline
\multicolumn{16}{c}{} \\ 
$\ln m$ & Min & ~ & 9 & 5 & 2 & 1 & 1 & 1 & 1 & 1 & 1 & 1 & 1 & 1 & 1 \\ 
~ & Q1 & ~ & 28 & 16 & 11 & 7 & 5 & 4 & 4 & 4 & 3 & 3 & 3 & 3 & 3 \\ 
~ & MED & ~ & 35 & 23 & 16 & 11 & 8 & 7 & 8 & 7 & 6 & 6 & 6 & 6 & 6 \\ 
~ & ARL & ~ & 37.68 & 25.15 & 18.67 & 13.77 & 10.51 & 9.92 & 10.42 & 9.56 & 
8.85 & 8.45 & 8.87 & 9.47 & 12.52 \\ 
~ & Q3 & ~ & 45 & 31.25 & 24 & 18 & 14 & 13 & 14 & 12 & 11 & 11 & 11 & 12 & 
14 \\ 
~ & Max & ~ & 112 & 99 & 92 & 85 & 84 & 81 & 81 & 81 & 81 & 83 & 100 & 114 & 
318 \\ 
\multicolumn{16}{c}{} \\ \hline
\multicolumn{16}{c}{} \\ 
$\ln^2 m$ & Min & ~ & 6 & 3 & 2 & 1 & 1 & 1 & 1 & 1 & 1 & 1 & 1 & 1 & 1 \\ 
~ & Q1 & ~ & 29 & 20 & 16 & 12 & 10 & 10 & 11 & 10 & 8 & 7 & 7 & 6 & 6 \\ 
~ & MED & ~ & 38 & 27 & 22 & 18 & 16 & 16 & 17 & 15 & 14 & 12 & 12 & 11 & 11
\\ 
~ & ARL & ~ & 40.14 & 29.81 & 24.58 & 20.54 & 18.27 & 18.38 & 19.48 & 17.90
& 16.24 & 14.86 & 14.47 & 14.21 & 14.52 \\ 
~ & Q3 & ~ & 48 & 36 & 31 & 26 & 23 & 24 & 25 & 23 & 21 & 20 & 19 & 19 & 19
\\ 
~ & Max & ~ & 119 & 106 & 102 & 101 & 87 & 88 & 102 & 88 & 85 & 85 & 88 & 90
& 107 \\ 
\multicolumn{16}{c}{} \\ 
\bottomrule \bottomrule &  &  &  &  &  &  &  &  &  &  &  &  &  &  & 
\end{tabular}
}}
\par
{\scriptsize \ }
\par
{\scriptsize {\footnotesize \ } }
\par
{\scriptsize {\footnotesize 
\begin{tablenotes}
      \tiny
            \item The table contains several descriptive statistics for the delay under the alternative of an \textit{early} occurring changepoint, for $m=500$ and different values of $\eta$ and different trimming sequence $a_m$, in the case of the static regression model (\ref{dgp}), with $d=8$ - i.e., 7 exogenous regressors and the constant. The specifications of the model are described in the main text; see the notes to Table \ref{tab:DelayEF1} for an explanation of the descriptive statistics. 
            
\end{tablenotes}
} }
\end{table*}

\medskip

\begin{table*}[h!]
\caption{{\protect\footnotesize {Delays under a late occurring changepoint -
static regression, full descriptive statistics}}}
\label{tab:DelayLS1}\centering
{\footnotesize {\ }}
\par
{\footnotesize {\ }}
\par
{\scriptsize {\ }}
\par
{\scriptsize 
\begin{tabular}{llllllllllllllll}
\multicolumn{16}{c}{} \\ 
$a_m$ & ~ & $\eta$ & 0 & 0.15 & 0.25 & 0.35 & 0.45 & 0.49 & 0.5 & 0.51 & 0.55
& 0.65 & 0.75 & 0.85 & 1 \\ 
\multicolumn{16}{c}{} \\ 
$\ln \ln m$ & Min & ~ & 16 & 10 & 3 & 3 & 1 & 2 & 3 & 2 & 5 & 15 & 59 & 239
& NA \\ 
~ & Q1 & ~ & 66 & 52 & 45 & 39 & 37 & 39 & 43 & 40 & 47 & 83 & 196 & 285 & NA
\\ 
~ & MED & ~ & 85 & 70 & 62 & 55 & 53 & 56 & 60 & 58 & 66 & 114 & 254 & 302 & 
NA \\ 
~ & ARL & ~ & 88.03 & 72.81 & 65.32 & 58.66 & 56.29 & 59.22 & 64.19 & 61.90
& 70.04 & 121.92 & 255.88 & 301.71 & NA \\ 
~ & Q3 & ~ & 106 & 90 & 81 & 74 & 72 & 75 & 81 & 78 & 88 & 151 & 320 & 319 & 
NA \\ 
~ & Max & ~ & 211 & 196 & 193 & 188 & 188 & 194 & 197 & 196 & 211 & 360 & 400
& 363 & NA \\ 
\multicolumn{16}{c}{} \\ \hline
\multicolumn{16}{c}{} \\ 
$\ln m$ & Min & ~ & 13 & 5 & 4 & 1 & 1 & 2 & 1 & 1 & 1 & 4 & 12 & 30 & 155
\\ 
~ & Q1 & ~ & 66 & 52 & 46 & 40 & 37 & 40 & 44 & 41 & 42 & 55 & 87 & 149 & 
245.25 \\ 
~ & MED & ~ & 85 & 70 & 63 & 57 & 54 & 57 & 62 & 58 & 60 & 78 & 119 & 203 & 
316.5 \\ 
~ & ARL & ~ & 88.84 & 73.70 & 66.37 & 59.99 & 57.73 & 60.88 & 65.72 & 62.20
& 64.16 & 83.63 & 129.16 & 212.21 & 293.50 \\ 
~ & Q3 & ~ & 107 & 91 & 83 & 76 & 74 & 78 & 83 & 79 & 81.25 & 105 & 162 & 272
& 342 \\ 
~ & Max & ~ & 210 & 193 & 190 & 188 & 189 & 192 & 204 & 192 & 204 & 263 & 391
& 394 & 394 \\ 
\multicolumn{16}{c}{} \\ \hline
\multicolumn{16}{c}{} \\ 
$\ln^2 m$ & Min & ~ & 17 & 2 & 4 & 4 & 2 & 3 & 3 & 3 & 1 & 1 & 4 & 5 & 2 \\ 
~ & Q1 & ~ & 70 & 56 & 50 & 44 & 43 & 46 & 50 & 46 & 43 & 45 & 51 & 59.5 & 83
\\ 
~ & MED & ~ & 90 & 76 & 68 & 63 & 62 & 65 & 70 & 65 & 62 & 64.5 & 73 & 85 & 
117 \\ 
~ & ARL & ~ & 93.81 & 79.34 & 72.61 & 66.79 & 65.53 & 69.46 & 75.12 & 69.62
& 66.31 & 69.23 & 78.60 & 91.78 & 128.41 \\ 
~ & Q3 & ~ & 114 & 98 & 91 & 85 & 84 & 88 & 95 & 89 & 85 & 88 & 99 & 116 & 
164.75 \\ 
~ & Max & ~ & 231 & 204 & 196 & 185 & 187 & 198 & 220 & 200 & 197 & 219 & 264
& 308 & 362 \\ 
\multicolumn{16}{c}{} \\ 
\bottomrule \bottomrule &  &  &  &  &  &  &  &  &  &  &  &  &  &  & 
\end{tabular}
}
\par
{\scriptsize \ }
\par
{\scriptsize {\footnotesize \ } }
\par
{\scriptsize {\footnotesize 
\begin{tablenotes}
      \tiny
            \item The table contains several descriptive statistics for the delay under the alternative of a \textit{late} occurring changepoint, for $m=500$ and different values of $\eta$ and different trimming sequence $a_m$, in the case of the static regression model (\ref{dgp}), with $d=2$ - i.e., one exogenous regressor and the constant. The specifications of the model are described in the main text; see the notes to Table \ref{tab:DelayEF1} for an explanation of the descriptive statistics. 
            
\end{tablenotes}
} }
\end{table*}

\medskip

\begin{table*}[h!]
\caption{{\protect\footnotesize {Delays under a late occurring changepoint -
static regression, full descriptive statistics}}}
\label{tab:DelayLS2}\centering
{\footnotesize {\ }}
\par
{\footnotesize {\ }}
\par
{\scriptsize {\ 
\begin{tabular}{llllllllllllllll}
\multicolumn{16}{c}{} \\ 
$a_m$ & ~ & $\eta$ & 0 & 0.15 & 0.25 & 0.35 & 0.45 & 0.49 & 0.5 & 0.51 & 0.55
& 0.65 & 0.75 & 0.85 & 1 \\ 
\multicolumn{16}{c}{} \\ 
$\ln \ln m$ & Min & ~ & 9 & 3 & 1 & 2 & 1 & 1 & 1 & 1 & 1 & 5 & 27 & 82 & NA
\\ 
~ & Q1 & ~ & 44 & 35 & 30 & 25 & 23 & 25 & 27 & 26 & 30 & 50 & 106 & 235.5 & 
NA \\ 
~ & MED & ~ & 58 & 47 & 41 & 36 & 34 & 36 & 39 & 37 & 41 & 67 & 140 & 298 & 
NA \\ 
~ & ARL & ~ & 60.28 & 49.24 & 43.80 & 39.18 & 37.24 & 38.89 & 41.77 & 40.45
& 44.94 & 71.96 & 150.46 & 288.41 & NA \\ 
~ & Q3 & ~ & 74 & 61 & 56 & 51 & 48 & 51 & 54 & 52.75 & 58 & 89 & 185 & 348
& NA \\ 
~ & Max & ~ & 171 & 164 & 161 & 156 & 156 & 161 & 164 & 163 & 170 & 246 & 398
& 400 & NA \\ 
\multicolumn{16}{c}{} \\ \hline
\multicolumn{16}{c}{} \\ 
$\ln m$ & Min & ~ & 8 & 4 & 1 & 1 & 1 & 1 & 1 & 1 & 1 & 2 & 3 & 14 & 46 \\ 
~ & Q1 & ~ & 45 & 35 & 30 & 26 & 24 & 25 & 27 & 26 & 26 & 34 & 50 & 78 & 189
\\ 
~ & MED & ~ & 59 & 47 & 42 & 37 & 35 & 37 & 40 & 37 & 38 & 48 & 69 & 108 & 
251 \\ 
~ & ARL & ~ & 60.64 & 49.67 & 44.28 & 39.74 & 38.10 & 39.77 & 42.59 & 40.47
& 41.48 & 51.95 & 73.68 & 116.02 & 251.34 \\ 
~ & Q3 & ~ & 74 & 62 & 56 & 51 & 50 & 52 & 55 & 52 & 54 & 66 & 92 & 145 & 320
\\ 
~ & Max & ~ & 177 & 157 & 152 & 150 & 150 & 151 & 157 & 155 & 156 & 197 & 250
& 360 & 395 \\ 
\multicolumn{16}{c}{} \\ \hline
\multicolumn{16}{c}{} \\ 
$\ln^2 m$ & Min & ~ & 4 & 5 & 1 & 1 & 1 & 1 & 2 & 1 & 1 & 1 & 1 & 3 & 6 \\ 
~ & Q1 & ~ & 47 & 37 & 32 & 28 & 27 & 29 & 32 & 29 & 27 & 27 & 31 & 35 & 47
\\ 
~ & MED & ~ & 61 & 51 & 46 & 41 & 40 & 42 & 46 & 42 & 40 & 41 & 45 & 52 & 66
\\ 
~ & ARL & ~ & 63.82 & 53.57 & 48.52 & 44.08 & 43.17 & 45.29 & 48.78 & 45.29
& 43.03 & 44.13 & 48.98 & 55.40 & 71.99 \\ 
~ & Q3 & ~ & 78 & 67 & 61 & 56 & 55 & 58 & 62 & 58 & 55 & 57 & 62 & 70 & 91
\\ 
~ & Max & ~ & 187 & 167 & 161 & 154 & 155 & 164 & 175 & 165 & 161 & 167 & 189
& 203 & 244 \\ 
\multicolumn{16}{c}{} \\ 
\bottomrule \bottomrule &  &  &  &  &  &  &  &  &  &  &  &  &  &  & 
\end{tabular}
}}
\par
{\scriptsize \ }
\par
{\scriptsize {\footnotesize \ } }
\par
{\scriptsize {\footnotesize 
\begin{tablenotes}
      \tiny
            \item The table contains several descriptive statistics for the delay under the alternative of a \textit{late} occurring changepoint, for $m=500$ and different values of $\eta$ and different trimming sequence $a_m$, in the case of the static regression model (\ref{dgp}), with $d=4$ - i.e., 3 exogenous regressors and the constant. The specifications of the model are described in the main text; see the notes to Table \ref{tab:DelayEF1} for an explanation of the descriptive statistics. 
            
\end{tablenotes}
} }
\end{table*}

\medskip

\begin{table*}[h]
\caption{{\protect\footnotesize {Delays under a late occurring changepoint -
static regression, full descriptive statistics}}}
\label{tab:DelayLS3}\centering
{\footnotesize {\ }}
\par
{\footnotesize {\ }}
\par
{\scriptsize {\ }}
\par
{\scriptsize 
\begin{tabular}{llllllllllllllll}
\multicolumn{16}{c}{} \\ 
$a_m$ & ~ & $\eta$ & 0 & 0.15 & 0.25 & 0.35 & 0.45 & 0.49 & 0.5 & 0.51 & 0.55
& 0.65 & 0.75 & 0.85 & 1 \\ 
\multicolumn{16}{c}{} \\ 
$\ln \ln m$ & Min & ~ & 7 & 2 & 2 & 2 & 1 & 1 & 2 & 1 & 1 & 4 & 19 & 38 & NA
\\ 
~ & Q1 & ~ & 32 & 24 & 21 & 18 & 16 & 18 & 19 & 18 & 20 & 33 & 63 & 147 & NA
\\ 
~ & MED & ~ & 43 & 34 & 30 & 27 & 25 & 26 & 28 & 27 & 30 & 46 & 86 & 202 & NA
\\ 
~ & ARL & ~ & 45.23 & 36.93 & 32.93 & 29.49 & 28.03 & 29.11 & 31.06 & 29.95
& 33.05 & 50.07 & 92.14 & 208.36 & NA \\ 
~ & Q3 & ~ & 56 & 46 & 42 & 38 & 36 & 37.75 & 40 & 39 & 43 & 63 & 113 & 260
& NA \\ 
~ & Max & ~ & 137 & 129 & 109 & 106 & 106 & 107 & 129 & 109 & 130 & 190 & 324
& 400 & NA \\ 
\multicolumn{16}{c}{} \\ \hline
\multicolumn{16}{c}{} \\ 
$\ln^2 m$ & Min & ~ & 5 & 2 & 1 & 1 & 1 & 1 & 1 & 1 & 1 & 1 & 4 & 9 & 27 \\ 
~ & Q1 & ~ & 32 & 25 & 21 & 18 & 17 & 18 & 19 & 18 & 18 & 24 & 33 & 49 & 109
\\ 
~ & MED & ~ & 43 & 35 & 31 & 27 & 26 & 27 & 29 & 27 & 28 & 34 & 46 & 67 & 151
\\ 
~ & ARL & ~ & 45.71 & 37.55 & 33.50 & 30.04 & 28.57 & 29.79 & 32.00 & 30.29
& 30.87 & 37.95 & 51.01 & 73.49 & 165.30 \\ 
~ & Q3 & ~ & 56 & 48 & 43 & 39 & 37 & 39 & 41 & 39 & 40 & 49 & 64 & 91 & 209
\\ 
~ & Max & ~ & 141 & 132 & 127 & 124 & 123 & 125 & 131 & 127 & 131 & 143 & 205
& 274 & 394 \\ 
\multicolumn{16}{c}{} \\ \hline
\multicolumn{16}{c}{} \\ 
$\ln^2 m$ & Min & ~ & 4 & 2 & 2 & 1 & 1 & 1 & 1 & 1 & 1 & 1 & 2 & 1 & 3 \\ 
~ & Q1 & ~ & 33 & 26 & 23 & 21 & 20 & 21 & 23 & 21 & 20 & 20 & 22 & 25 & 31
\\ 
~ & MED & ~ & 45 & 37 & 33 & 30 & 29 & 30.5 & 33 & 30 & 29 & 29 & 32 & 36 & 
45 \\ 
~ & ARL & ~ & 47.87 & 40.02 & 36.50 & 33.26 & 32.43 & 33.96 & 36.31 & 33.88
& 32.35 & 33.01 & 36.03 & 39.93 & 49.44 \\ 
~ & Q3 & ~ & 59 & 51 & 47 & 43 & 42 & 44 & 47 & 44 & 42 & 43 & 47 & 51 & 62
\\ 
~ & Max & ~ & 156 & 151 & 147 & 133 & 133 & 147 & 151 & 147 & 147 & 149 & 155
& 170 & 223 \\ 
\multicolumn{16}{c}{} \\ 
\bottomrule \bottomrule &  &  &  &  &  &  &  &  &  &  &  &  &  &  & 
\end{tabular}
}
\par
{\scriptsize \ }
\par
{\scriptsize {\footnotesize \ } }
\par
{\scriptsize {\footnotesize 
\begin{tablenotes}
      \tiny
            \item The table contains several descriptive statistics for the delay under the alternative of a \textit{late} occurring changepoint, for $m=500$ and different values of $\eta$ and different trimming sequence $a_m$, in the case of the static regression model (\ref{dgp}), with $d=8$ - i.e., 7 exogenous regressors and the constant. The specifications of the model are described in the main text; see the notes to Table \ref{tab:DelayEF1} for an explanation of the descriptive statistics. 
            
\end{tablenotes}
} }
\end{table*}

\clearpage\newpage

\subsection{Veto-based changepoint detection: empirical rejection
frequencies \label{veto-erf-null}}

\medskip

\begin{table*}[h!]
\caption{{\protect\footnotesize {Empirical rejection frequencies under $H_0$
- veto-based changepoint detection with $m=500$ and $d=4$}}}
\label{tab:VetoApp1}\centering
\par
{\scriptsize {\ }}
\par
{\scriptsize 
\begin{tabular}{llccccccccc}
\multicolumn{11}{c}{} \\ 
\multicolumn{1}{c}{} & \multicolumn{1}{c}{$a_m$} & \multicolumn{3}{c}{$\ln
\ln m$} & \multicolumn{3}{c}{$\ln m$} & \multicolumn{3}{c}{$\ln^2 m$} \\ 
\multicolumn{2}{c}{} & \multicolumn{3}{c}{} & \multicolumn{3}{c}{} & 
\multicolumn{3}{c}{} \\ 
\cmidrule(lr){3-5}\cmidrule(lr){6-8}\cmidrule(lr){9-11} &  &  &  &  &  &  & 
&  &  &  \\ 
~ & $m$ & 300 & 500 & 1000 & 300 & 500 & 1000 & 300 & 500 & 1000 \\ 
\multicolumn{1}{c}{$\eta$} & \multicolumn{3}{c}{} & \multicolumn{3}{c}{} & 
\multicolumn{3}{c}{} &  \\ 
\multicolumn{2}{c}{} & \multicolumn{3}{c}{} & \multicolumn{3}{c}{} & 
\multicolumn{3}{c}{} \\ 
0.25 & ~ & 0.027 & 0.022 & 0.017 & 0.027 & 0.022 & 0.017 & 0.027 & 0.022 & 
0.017 \\ 
0.75 & ~ & 0.036 & 0.031 & 0.023 & 0.052 & 0.046 & 0.034 & 0.078 & 0.059 & 
0.051 \\ 
$\mathcal{V}_2$ & ~ & 0.060 & 0.050 & 0.040 & 0.068 & 0.060 & 0.049 & 0.085
& 0.069 & 0.062 \\ 
$\mathcal{V}_3$ & ~ & 0.068 & 0.057 & 0.046 & 0.075 & 0.066 & 0.054 & 0.084
& 0.070 & 0.065 \\ 
$\mathcal{V}_5$ & ~ & 0.064 & 0.051 & 0.042 & 0.062 & 0.057 & 0.043 & 0.057
& 0.048 & 0.044 \\ 
\multicolumn{11}{c}{} \\ 
\midrule \bottomrule &  &  &  &  &  &  &  &  &  & 
\end{tabular}
}
\par
{\scriptsize {\footnotesize 
\begin{tablenotes}
      \tiny
            \item The table contains the empirical rejection frequencies under the null of no changepoint for different sample sizes and different monitoring schemes, for the case of the dynamic regression model (\ref{dgp}), and $d=4$ - i.e. 3 exogenous regressor and the constant. The other specifications of the model are the same as for Table \ref{tab:Size1}. 
            
\end{tablenotes}
} }
\end{table*}

\begin{table*}[h]
\caption{{\protect\footnotesize {Empirical rejection frequencies under $%
H_{0} $ - veto-based changepoint detection with $m=500$ and $d=8$}}}
\label{tab:VetoApp2}\centering
\par
{\scriptsize {\ }}
\par
{\scriptsize 
\begin{tabular}{llccccccccc}
\multicolumn{11}{c}{} \\ 
\multicolumn{1}{c}{} & \multicolumn{1}{c}{$a_m$} & \multicolumn{3}{c}{$\ln
\ln m$} & \multicolumn{3}{c}{$\ln m$} & \multicolumn{3}{c}{$\ln^2 m$} \\ 
\multicolumn{2}{c}{} & \multicolumn{3}{c}{} & \multicolumn{3}{c}{} & 
\multicolumn{3}{c}{} \\ 
\cmidrule(lr){3-5}\cmidrule(lr){6-8}\cmidrule(lr){9-11} &  &  &  &  &  &  & 
&  &  &  \\ 
~ & $m$ & 300 & 500 & 1000 & 300 & 500 & 1000 & 300 & 500 & 1000 \\ 
\multicolumn{1}{c}{$\eta$} & \multicolumn{3}{c}{} & \multicolumn{3}{c}{} & 
\multicolumn{3}{c}{} &  \\ 
\multicolumn{2}{c}{} & \multicolumn{3}{c}{} & \multicolumn{3}{c}{} & 
\multicolumn{3}{c}{} \\ 
0.25 & ~ & 0.031 & 0.028 & 0.024 & 0.031 & 0.028 & 0.024 & 0.031 & 0.028 & 
0.024 \\ 
0.75 & ~ & 0.040 & 0.034 & 0.028 & 0.062 & 0.056 & 0.042 & 0.085 & 0.072 & 
0.066 \\ 
$\mathcal{V}_2$ & ~ & 0.061 & 0.052 & 0.045 & 0.080 & 0.075 & 0.062 & 0.092
& 0.080 & 0.076 \\ 
$\mathcal{V}_3$ & ~ & 0.073 & 0.064 & 0.052 & 0.091 & 0.084 & 0.068 & 0.095
& 0.085 & 0.078 \\ 
$\mathcal{V}_5$ & ~ & 0.066 & 0.062 & 0.052 & 0.070 & 0.070 & 0.062 & 0.062
& 0.058 & 0.056 \\ 
\multicolumn{11}{c}{} \\ 
\midrule \bottomrule &  &  &  &  &  &  &  &  &  & 
\end{tabular}
}
\par
{\scriptsize \ }
\par
{\scriptsize {\footnotesize 
\begin{tablenotes}
      \tiny
            \item The table contains the empirical rejection frequencies under the null of no changepoint for different sample sizes and different monitoring schemes, for the case of the dynamic regression model (\ref{dgp}), and $d=8$ - i.e. 7 exogenous regressor and the constant. The other specifications of the model are the same as for Table \ref{tab:Size1}. 
            
\end{tablenotes}
} }
\end{table*}

\clearpage\newpage

\subsection{Veto-based changepoint detection: detection delays\label%
{veto-erf-alt}}

\begin{table*}[h]
\caption{{\protect\footnotesize {Delays using veto-based changepoint
detection - dynamic regression with $m=500$ and $d=4$, full descriptive
statistics}}}
\label{tab:VetoApp3}\centering
\par
{\scriptsize {\ }}
\par
{\scriptsize 
\begin{tabular}{lllllllllllllll}
\multicolumn{15}{c}{} \\ 
\multicolumn{4}{c}{} & \multicolumn{5}{c}{Early Break} &  & 
\multicolumn{5}{c}{Late Break} \\ 
\cmidrule(lr){4-9}\cmidrule(lr){10-15} &  &  &  &  &  &  &  &  &  &  &  &  & 
&  \\ 
\multicolumn{15}{c}{} \\ 
$a_m$ &  & $\eta$ & ~ & 0.25 & 0.75 & $\mathcal{V}_{2}$ & $\mathcal{V}_{3} $
& $\mathcal{V}_{5}$ & ~ & 0.25 & 0.75 & $\mathcal{V}_2$ & $\mathcal{V}_{3} $
& $\mathcal{V}_{5}$ \\ 
\multicolumn{15}{c}{} \\ 
$\ln \ln m$ & Min & ~ & ~ & 2 & 1 & 1 & 1 & 1 & ~ & 1 & 10 & 1 & 1 & 1 \\ 
~ & Q1 & ~ & ~ & 6 & 1 & 1 & 1 & 1 & ~ & 15 & 40 & 16 & 14 & 13 \\ 
~ & MED & ~ & ~ & 10 & 2 & 1 & 1.5 & 2 & ~ & 21 & 52 & 22 & 20 & 19 \\ 
~ & ARL & ~ & ~ & 11.47 & 4.31 & 4.98 & 4.59 & 3.90 & ~ & 22.41 & 54.31 & 
23.87 & 21.33 & 20.59 \\ 
~ & Q3 & ~ & ~ & 15 & 4 & 5 & 5 & 5 & ~ & 28 & 66 & 30 & 27 & 26 \\ 
~ & Max & ~ & ~ & 51 & 70 & 52 & 46 & 44 & ~ & 76 & 174 & 78 & 74 & 74 \\ 
~ & ~ & ~ & ~ & ~ & ~ & ~ & ~ & ~ & ~ & ~ & ~ & ~ & ~ & ~ \\ 
$\ln m$ & Min & ~ & ~ & 2 & 1 & 1 & 1 & 1 & ~ & 1 & 3 & 2 & 1 & 1 \\ 
~ & Q1 & ~ & ~ & 7 & 2 & 2 & 2 & 2 & ~ & 15 & 22 & 16 & 14 & 13 \\ 
~ & MED & ~ & ~ & 11 & 4 & 4 & 4 & 4 & ~ & 21 & 30 & 22 & 20 & 19 \\ 
~ & ARL & ~ & ~ & 12.61 & 5.59 & 5.79 & 5.87 & 5.93 & ~ & 22.63 & 32.30 & 
24.08 & 21.47 & 20.82 \\ 
~ & Q3 & ~ & ~ & 16 & 7 & 7 & 7 & 8 & ~ & 28 & 41 & 30 & 27 & 26 \\ 
~ & Max & ~ & ~ & 47 & 41 & 44 & 42 & 39 & ~ & 74 & 98 & 76 & 73 & 73 \\ 
~ & ~ & ~ & ~ & ~ & ~ & ~ & ~ & ~ & ~ & ~ & ~ & ~ & ~ &  \\ 
$\ln^2 m $ & Min & ~ & ~ & 1 & 1 & 1 & 1 & 1 & ~ & 1 & 1 & 2 & 2 & 2 \\ 
~ & Q1 & ~ & ~ & 11 & 5 & 5 & 5 & 5 & ~ & 16 & 15 & 17 & 16 & 15 \\ 
~ & MED & ~ & ~ & 15 & 8 & 8 & 8 & 9 & ~ & 23 & 22 & 24 & 22 & 22 \\ 
~ & ARL & ~ & ~ & 17.19 & 9.83 & 9.71 & 9.79 & 10.36 & ~ & 24.88 & 23.78 & 
25.79 & 23.82 & 23.73 \\ 
~ & Q3 & ~ & ~ & 21 & 13 & 13 & 13 & 14 & ~ & 31 & 30 & 32 & 30 & 30 \\ 
~ & Max & ~ & ~ & 69 & 50 & 53 & 53 & 53 & ~ & 83 & 84 & 83 & 82 & 83 \\ 
\multicolumn{15}{c}{} \\ 
\bottomrule \bottomrule &  &  &  &  &  &  &  &  &  &  &  &  &  & 
\end{tabular}
}
\par
{\scriptsize \ }
\par
{\scriptsize {\footnotesize 
\begin{tablenotes}
      \tiny
            \item The table contains several descriptive statistics for the delay under the alternative, for $m=500$ and different monitoring schemes, in the case of the dynamic regression model (\ref{dgp}), with $d=4$ - i.e., 3 exogenous regressors and the constant. The specifications of the model are described in the main text; see the notes to Table \ref{tab:DelayEF1} for an explanation of the descriptive statistics. 
            
\end{tablenotes}
} }
\end{table*}

\begin{table*}[h!]
\caption{{\protect\footnotesize {Delays using veto-based changepoint
detection - dynamic regression with $m=500$ and $d=8$, full descriptive
statistics}}}
\label{tab:VetoApp4}\centering
\par
{\scriptsize {\ }}
\par
{\scriptsize 
\begin{tabular}{lllllllllllllll}
\multicolumn{15}{c}{} \\ 
\multicolumn{4}{c}{} & \multicolumn{5}{c}{Early Break} &  & 
\multicolumn{5}{c}{Late Break} \\ 
\cmidrule(lr){4-9}\cmidrule(lr){10-15} &  &  &  &  &  &  &  &  &  &  &  &  & 
&  \\ 
\multicolumn{15}{c}{} \\ 
$a_m$ &  & $\eta$ & ~ & 0.25 & 0.75 & $\mathcal{V}_{2}$ & $\mathcal{V}_{3} $
& $\mathcal{V}_{5}$ & ~ & 0.25 & 0.75 & $\mathcal{V}_2$ & $\mathcal{V}_{3} $
& $\mathcal{V}_{5}$ \\ 
\multicolumn{15}{c}{} \\ 
$\ln \ln m$ & Min & ~ & ~ & 1 & 1 & 1 & 1 & 1 & ~ & 1 & 7 & 1 & 1 & 1 \\ 
~ & Q1 & ~ & ~ & 4 & 1 & 1 & 1 & 1 & ~ & 11 & 26 & 11 & 10 & 9 \\ 
~ & MED & ~ & ~ & 6 & 1 & 1 & 1 & 1 & ~ & 16 & 35 & 17 & 15 & 14 \\ 
~ & ARL & ~ & ~ & 8.34 & 2.74 & 3.03 & 2.93 & 2.65 & ~ & 17.46 & 38.73 & 
18.70 & 16.65 & 16.07 \\ 
~ & Q3 & ~ & ~ & 11 & 3 & 3 & 3 & 3 & ~ & 22 & 48 & 24 & 21 & 21 \\ 
~ & Max & ~ & ~ & 50 & 52 & 50 & 49 & 31 & ~ & 69 & 129 & 79 & 69 & 69 \\ 
~ & ~ & ~ & ~ & ~ & ~ & ~ & ~ & ~ & ~ & ~ & ~ & ~ & ~ & ~ \\ 
$\ln m$ & Min & ~ & ~ & 1 & 1 & 1 & 1 & 1 & ~ & 1 & 2 & 1 & 1 & 1 \\ 
~ & Q1 & ~ & ~ & 5 & 1 & 1 & 1 & 2 & ~ & 11 & 15 & 12 & 10 & 9 \\ 
~ & MED & ~ & ~ & 8 & 3 & 3 & 3 & 3 & ~ & 16 & 21 & 17 & 15 & 14 \\ 
~ & ARL & ~ & ~ & 9.64 & 3.99 & 4.07 & 4.10 & 4.26 & ~ & 17.86 & 24.33 & 
18.99 & 16.97 & 16.49 \\ 
~ & Q3 & ~ & ~ & 12 & 5 & 5 & 5 & 5 & ~ & 23 & 31 & 24 & 22 & 21 \\ 
~ & Max & ~ & ~ & 71 & 47 & 70 & 62 & 61 & ~ & 65 & 93 & 65 & 65 & 65 \\ 
~ & ~ & ~ & ~ & ~ & ~ & ~ & ~ & ~ & ~ & ~ & ~ & ~ & ~ &  \\ 
$\ln^2 m$ & Min & ~ & ~ & 1 & 1 & 1 & 1 & 1 & ~ & 1 & 1 & 1 & 1 & 1 \\ 
~ & Q1 & ~ & ~ & 8 & 3 & 3 & 3 & 4 & ~ & 12 & 11 & 12 & 11 & 11 \\ 
~ & MED & ~ & ~ & 11 & 6 & 6 & 6 & 6 & ~ & 17 & 16 & 18 & 16 & 16 \\ 
~ & ARL & ~ & ~ & 13.29 & 7.60 & 7.44 & 7.50 & 7.97 & ~ & 19.42 & 18.50 & 
20.06 & 18.68 & 18.45 \\ 
~ & Q3 & ~ & ~ & 17 & 10 & 10 & 10 & 10 & ~ & 24 & 23 & 26 & 24 & 23 \\ 
~ & Max & ~ & ~ & 66 & 62 & 62 & 62 & 62 & ~ & 89 & 90 & 90 & 89 & 89 \\ 
\multicolumn{15}{c}{} \\ 
\bottomrule \bottomrule &  &  &  &  &  &  &  &  &  &  &  &  &  & 
\end{tabular}
}
\par
{\scriptsize \ }
\par
{\scriptsize {\footnotesize 
\begin{tablenotes}
      \tiny
            \item The table contains several descriptive statistics for the delay under the alternative, for $m=500$ and different monitoring schemes, in the case of the dynamic regression model (\ref{dgp}), with $d=8$ - i.e., 7 exogenous regressors and the constant. The specifications of the model are described in the main text; see the notes to Table \ref{tab:DelayEF1} for an explanation of the descriptive statistics.
            
\end{tablenotes}
} }
\end{table*}

\clearpage
\newpage

\renewcommand*{\thesection}{\Alph{section}}

\setcounter{equation}{0} \setcounter{lemma}{0} \setcounter{theorem}{0} %
\renewcommand{\theassumption}{B.\arabic{assumption}} 
\renewcommand{\thetheorem}{B.\arabic{theorem}} \renewcommand{\thelemma}{B.%
\arabic{lemma}} \renewcommand{\theproposition}{B.\arabic{proposition}} %
\renewcommand{\thecorollary}{B.\arabic{corollary}} \renewcommand{%
\theequation}{B.\arabic{equation}}

\section{Preliminary lemmas\label{lemmas}}

Henceforth, we use the following notation: $\left\lceil \varrho \right\rceil 
$ is the ceiling of a real number $\varrho $. We also use the convention
that, if a summation involves a non-integer index, this is rounded up - e.g. 
$\sum_{i=1}^{\varrho }=\sum_{i=1}^{\left\lceil \varrho \right\rceil }$ for
any $\varrho >1$.

\begin{lemma}
\label{sip}We assume that Assumption \ref{b-shifts}\textit{(i)} holds. Then,
for every $m$, two independent standard Wiener processes $\left\{
W_{1,m}\left( k\right) ,k\geq 1\right\} $ and $\left\{ W_{2,m}\left(
k\right) ,k\geq 1\right\} $ whose distribution does not depend on $m$\ can
be defined on a suitably larger probability space such that, for some $%
0<\zeta _{1}<1/2$, it holds that 
\begin{equation}
\max_{1\leq k\leq T_{m}}\frac{1}{k^{\zeta _{1}}}\left\vert
\sum_{t=m+1}^{m+k}\epsilon _{t}-\sigma W_{1,m}(k)\right\vert =O_{P}(1),
\label{SIP1}
\end{equation}%
and 
\begin{equation}
\frac{1}{m^{\zeta _{2}}}\left\vert \sum_{t=1}^{m}\epsilon _{t}-\sigma
W_{2,m}(m)\right\vert =O_{P}(1),  \label{SIP2}
\end{equation}%
where $\sigma ^{2}$ is defined in (\ref{lrunv}).
\end{lemma}

\begin{proof}
The desired result follows immediately from applying Theorem B.1 in %
\citet{aue2014dependent}; note that the proof of the theorem is based on the
blocking argument, whence the independence of $\left\{ W_{1,m}\left(
k\right) ,k\geq 1\right\} $ and $\left\{ W_{2,m}\left( k\right) ,k\geq
1\right\} $.
\end{proof}

\begin{lemma}
\label{max-ineq-1}We assume that Assumptions \ref{b-shifts} and \ref%
{exogeneity} hold. Then it holds that

\begin{equation*}
\left\Vert \frac{1}{m}\sum_{t=1}^{m}\mathbf{x}_{t}\epsilon _{t}\right\Vert
=O_{P}\left( \left( \frac{d}{m}\right) ^{1/2}\right) .
\end{equation*}
\end{lemma}

\begin{proof}
We begin by estimating $E\left\Vert \sum_{t=1}^{m}\mathbf{x}_{t}\epsilon
_{t}\right\Vert ^{2}$, noting that%
\begin{equation}
E\left\Vert \sum_{t=1}^{m}\mathbf{x}_{t}\epsilon _{t}\right\Vert
^{2}=\sum_{j=1}^{d}\left\vert \sum_{t=1}^{m}x_{j,t}\epsilon _{t}\right\vert
_{2}^{2}.  \label{desired-result-1}
\end{equation}%
We now show that Assumptions \ref{b-shifts} and \ref{exogeneity} entail that
the sequence $z_{j,t}=x_{j,t}\epsilon _{t}$ is a zero mean, $L_{2}$%
-decomposable Bernoulli shift. Indeed, Assumption \ref{exogeneity}
immediately yields $E\left( z_{j,t}\right) =0$; further, consider the
coupling constructions%
\begin{align*}
&\widetilde{x}_{j,t,t} =h_{j}^{x}\left( \eta _{j,t}^{x},...,\eta _{j,1}^{x},%
\widetilde{\eta }_{j,0,t,t}^{x},\widetilde{\eta }_{j,-1,t,t}^{x},...\right) ,
\\
&\widetilde{\epsilon }_{t,t} =h^{\epsilon }\left( \eta _{t}^{\epsilon
},...,\eta _{1}^{\epsilon },\widetilde{\eta }_{0,t,t}^{\epsilon },\widetilde{%
\eta }_{-1,t,t}^{\epsilon },...\right) ,
\end{align*}%
where $\left\{ \widetilde{\eta }_{j,0,t,t}^{x},\widetilde{\eta }%
_{j,-1,t,t}^{x},...\right\} $ are independent copies of $\left\{ \eta
_{j,0}^{x},\eta _{j,-1}^{x},...\right\} $, and similarly $\left\{ \widetilde{%
\eta }_{j,0,t,t}^{\epsilon },\widetilde{\eta }_{j,-1,t,t}^{\epsilon
},...\right\} $. It holds that%
\begin{align*}
x_{j,t}\epsilon _{t} =&\left( x_{j,t}\pm \widetilde{x}_{j,t,t}\right) \left(
\epsilon _{t}\pm \widetilde{\epsilon }_{t,t}\right) \\
=&\widetilde{z}_{j,t,t}+\widetilde{x}_{j,t,t}\left( \epsilon _{t}-\widetilde{%
\epsilon }_{t,t}\right) +\widetilde{\epsilon }_{t,t}\left( x_{j,t}-%
\widetilde{x}_{j,t,t}\right) +\left( x_{j,t}-\widetilde{x}_{j,t,t}\right)
\left( \epsilon _{t}-\widetilde{\epsilon }_{t,t}\right) ,
\end{align*}%
where 
\begin{equation*}
\widetilde{z}_{j,t,t}=\widetilde{x}_{j,t,t}\widetilde{\epsilon }_{t,t},
\end{equation*}%
is the coupled version of $z_{j,t}$. Hence it follows that%
\begin{align*}
\left\vert z_{j,t}-\widetilde{z}_{j,t,t}\right\vert _{2} \leq &\left\vert 
\widetilde{x}_{j,t,t}\left( \epsilon _{t}-\widetilde{\epsilon }_{t,t}\right)
\right\vert _{2}+\left\vert \widetilde{\epsilon }_{t,t}\left( x_{j,t}-%
\widetilde{x}_{j,t,t}\right) \right\vert _{2}+\left\vert \left( x_{j,t}-%
\widetilde{x}_{j,t,t}\right) \left( \epsilon _{t}-\widetilde{\epsilon }%
_{t,t}\right) \right\vert _{2} \\
\leq &\left\vert \widetilde{x}_{j,t,t}\right\vert _{4}\left\vert \epsilon
_{t}-\widetilde{\epsilon }_{t,t}\right\vert _{4}+\left\vert x_{j,t}-%
\widetilde{x}_{j,t,t}\right\vert _{4}\left\vert \widetilde{\epsilon }%
_{t,t}\right\vert _{4}+\left\vert \epsilon _{t}-\widetilde{\epsilon }%
_{t,t}\right\vert _{4}\left\vert x_{j,t}-\widetilde{x}_{j,t,t}\right\vert
_{4}
\end{align*}%
having used Minkowski's inequality in the first line, and the
Cauchy-Schwartz inequality in the second one. Assumption \ref{b-shifts} now
immediately yields that there exists a constant $0<c_{0}<\infty $ such that 
\begin{equation*}
\left\vert z_{j,t}-\widetilde{z}_{j,t,t}\right\vert _{2}\leq c_{0}t^{-a},
\end{equation*}%
for $a>2$. Hence, we can directly apply Proposition 4 in %
\citet{berkes2011split}, obtaining that, for all $1\leq j\leq d$, there
exist constants $c_{j}<\infty $ such that 
\begin{equation}
\left\vert \sum_{t=1}^{m}x_{j,t}\epsilon _{t}\right\vert _{2}^{2}\leq c_{j}m.
\label{berkes}
\end{equation}%
Putting (\ref{berkes}) in (\ref{desired-result-1}), we finally obtain that $%
E\left\Vert \sum_{t=1}^{m}\mathbf{x}_{t}\epsilon _{t}\right\Vert ^{2}\leq
c_{0}dm$, whence the desired result follows from Markov inequality.
\end{proof}

\begin{lemma}
\label{max-ineq-2}We assume that Assumptions \ref{b-shifts} and \ref%
{exogeneity} hold. Then it holds that

\begin{equation*}
\max_{1\leq k\leq T_{m}}\frac{1}{k^{\zeta _{3}}}\left\Vert
\sum_{t=m+1}^{m+k}(\mathbf{x}_{t}-\mathbf{c}_{1})\right\Vert =O_{P}(d^{1/2}),
\end{equation*}%
for all $\zeta _{3}>1/2,$ where $\mathbf{c}%
_{1}=(1,E(x_{2,0}),...,E(x_{d,0}))^{\prime }$.
\end{lemma}

\begin{proof}
It holds that%
\begin{align*}
&P\left( \max_{1\leq k\leq T_{m}}\frac{1}{k^{\zeta _{3}}}\left\Vert
\sum_{t=m+1}^{m+k}(\mathbf{x}_{t}-\mathbf{c}_{1})\right\Vert \geq
xd^{1/2}\right) \\
\leq &P\left( \max_{0\leq \ell \leq \left\lceil \ln T_{m}\right\rceil
}\max_{\exp \left( \ell \right) \leq k\leq \exp \left( \ell +1\right) }\frac{%
1}{k^{\zeta _{3}}}\left\Vert \sum_{t=m+1}^{m+k}(\mathbf{x}_{t}-\mathbf{c}%
_{1})\right\Vert \geq xd^{1/2}\right) \\
\leq &\sum_{\ell =0}^{\left\lceil \ln T_{m}\right\rceil }P\left( \max_{\exp
\left( \ell \right) \leq k\leq \exp \left( \ell +1\right) }\frac{1}{k^{\zeta
_{3}}}\left\Vert \sum_{t=m+1}^{m+k}(\mathbf{x}_{t}-\mathbf{c}%
_{1})\right\Vert \geq xd^{1/2}\right) \\
\leq &\sum_{\ell =0}^{\left\lceil \ln T_{m}\right\rceil }P\left( \max_{\exp
\left( \ell \right) \leq k\leq \exp \left( \ell +1\right) }\left\Vert
\sum_{t=m+1}^{m+k}(\mathbf{x}_{t}-\mathbf{c}_{1})\right\Vert \geq
xd^{1/2}\exp \left( \zeta _{3}\ell \right) \right) \\
\leq &c_{0}x^{-p}d^{-p/2}\sum_{\ell =0}^{\left\lceil \ln T_{m}\right\rceil
}\exp \left( -p\zeta _{3}\ell \right) E\left( \max_{\exp \left( \ell \right)
\leq k\leq \exp \left( \ell +1\right) }\left\Vert \sum_{t=m+1}^{m+k}(\mathbf{%
x}_{t}-\mathbf{c}_{1})\right\Vert ^{p}\right)
\end{align*}%
for any $2\leq p\leq 4$. Further, we have%
\begin{align*}
&E\left( \max_{\exp \left( \ell \right) \leq k\leq \exp \left( \ell
+1\right) }\left\Vert \sum_{t=m+1}^{m+k}(\mathbf{x}_{t}-\mathbf{c}%
_{1})\right\Vert ^{p}\right) \\
\leq &d^{p/2-1}\sum_{j=2}^{d}E\left( \max_{\exp \left( \ell \right) \leq
k\leq \exp \left( \ell +1\right) }\left\vert
\sum_{t=m+1}^{m+k}(x_{j,t}-Ex_{j,0})\right\vert ^{p}\right) \\
\leq &d^{p/2-1}\sum_{j=2}^{d}E\left( \max_{\exp \left( \ell \right) \leq
k\leq \exp \left( \ell +1\right) }\left\vert
\sum_{t=m+1}^{m+k}(x_{j,t}-Ex_{j,0})\right\vert ^{p}\right) .
\end{align*}%
Recall that $x_{j,t}-Ex_{j,0}$ is a zero mean, $L_{4}$-decomposable
Bernoulli shift which satisfies the assumptions of Proposition 4 in %
\citet{berkes2011split}; using Theorem 1 in \citet{moricz1976moment}, it
follows that%
\begin{equation*}
E\left( \max_{1\leq k\leq \exp \left( \ell \right) }\left\vert
\sum_{t=m+1}^{m+k}(x_{j,t}-Ex_{j,0})\right\vert ^{p}\right) \leq c_{j}\exp
\left( \frac{p}{2}\ell \right) .
\end{equation*}%
Hence, putting all together it follows that%
\begin{equation*}
P\left( \max_{1\leq k\leq T_{m}}\frac{1}{k^{\zeta _{3}}}\left\Vert
\sum_{t=m+1}^{m+k}(\mathbf{x}_{t}-\mathbf{c}_{1})\right\Vert \geq
xd^{1/2}\right) \leq c_{0}x^{-p}\sum_{\ell =0}^{\ln T_{m}}\exp \left( p\ell
\left( -\zeta _{3}+\frac{1}{2}\right) \right) \leq c_{1},
\end{equation*}%
whenever $\zeta _{3}>1/2$, which proves the desired result.
\end{proof}

\begin{lemma}
\label{max-inequ-3}We assume that Assumptions \ref{b-shifts} and \ref%
{exogeneity} hold. Then it holds that 
\begin{equation*}
\left\Vert \frac{1}{m}\sum_{t=1}^{m}\left( \mathbf{x}_{t}\mathbf{x}%
_{t}^{\prime }-\mathbf{C}\right) \right\Vert _{F}=O_{P}\left( \frac{d}{\sqrt{%
m}}\right)
\end{equation*}
\end{lemma}

\begin{proof}
We begin by showing that, for all $2\leq h,j\leq d$, $x_{h,t}x_{j,t}$ is an $%
L_{2}$-decomposable Bernoulli shift with $a>2$. Indeed, consider the coupling%
\begin{equation*}
\widetilde{x}_{j,t,t}=g_{j}^{x}\left( \eta _{j,t}^{x},...,\eta _{j,1}^{x},%
\widetilde{\eta }_{j,0,t,t}^{x},\widetilde{\eta }_{j,-1,t,t}^{x},...\right) ,
\end{equation*}%
where $\left\{ \widetilde{\eta }_{j,0,t,t}^{x},\widetilde{\eta }%
_{j,-1,t,t}^{x},...\right\} $ are independent copies of $\left\{ \eta
_{j,0}^{x},\eta _{j,-1}^{x},...\right\} $, and similarly $\left\{ \widetilde{%
\eta }_{j,0,t,t}^{\epsilon },\widetilde{\eta }_{j,-1,t,t}^{\epsilon
},...\right\} $, $2\leq j\leq d$. It holds that%
\begin{align*}
x_{h,t}x_{j,t} =&\left( x_{h,t}\pm \widetilde{x}_{h,t,t}\right) \left(
x_{j,t}\pm \widetilde{x}_{j,t,t}\right) \\
=&\widetilde{z}_{j,h,t,t}+\widetilde{x}_{j,t,t}\left( x_{h,t}-\widetilde{x}%
_{h,t,t}\right) +\widetilde{x}_{h,t,t}\left( x_{j,t}-\widetilde{x}%
_{j,t,t}\right) +\left( x_{j,t}-\widetilde{x}_{j,t,t}\right) \left( x_{h,t}-%
\widetilde{x}_{h,t,t}\right) ,
\end{align*}%
where $\widetilde{z}_{j,h,t,t}=\widetilde{x}_{j,t,t}\widetilde{x}_{h,t,t}$
is the coupled version of $z_{j,h,t}=x_{h,t}x_{j,t}$. Hence it follows that%
\begin{align*}
\left\vert z_{j,t}-\widetilde{z}_{j,h,t,t}\right\vert _{2} \leq &\left\vert 
\widetilde{x}_{j,t,t}\left( x_{h,t}-\widetilde{x}_{h,t,t}\right) \right\vert
_{2}+\left\vert \widetilde{x}_{h,t,t}\left( x_{j,t}-\widetilde{x}%
_{j,t,t}\right) \right\vert _{2}+\left\vert \left( x_{j,t}-\widetilde{x}%
_{j,t,t}\right) \left( x_{h,t}-\widetilde{x}_{h,t,t}\right) \right\vert _{2}
\\
\leq &\left\vert \widetilde{x}_{j,t,t}\right\vert _{4}\left\vert x_{h,t}-%
\widetilde{x}_{h,t,t}\right\vert _{4}+\left\vert x_{j,t}-\widetilde{x}%
_{j,t,t}\right\vert _{4}\left\vert \widetilde{x}_{h,t,t}\right\vert
_{4}+\left\vert x_{j,t}-\widetilde{x}_{j,t,t}\right\vert _{4}\left\vert
x_{h,t}-\widetilde{x}_{h,t,t}\right\vert _{4}
\end{align*}%
having used Minkowski's inequality in the first line, and the
Cauchy-Schwartz inequality in the second one. Assumption \ref{b-shifts}%
\textit{(ii)} now immediately yields that there exists a constant $%
0<c_{0}<\infty $ such that 
\begin{equation*}
\left\vert z_{j,h,t}-\widetilde{z}_{j,h,t,t}\right\vert _{2}\leq c_{0}t^{-a}.
\end{equation*}%
We now note that%
\begin{equation*}
\left\Vert \sum_{t=1}^{m}\left( \mathbf{x}_{t}\mathbf{x}_{t}^{\prime }-%
\mathbf{C}\right) \right\Vert _{F}^{2}=\sum_{j,h=1}^{d}\left\vert
\sum_{t=1}^{m}\left( x_{h,t}x_{j,t}-\mathbf{C}_{h,j}\right) \right\vert ^{2}.
\end{equation*}%
Using Proposition 4 in \citet{berkes2011split}, it follows that 
\begin{equation*}
\left\vert \sum_{t=1}^{m}\left( x_{h,t}x_{j,t}-\mathbf{C}_{h,j}\right)
\right\vert _{2}^{2}\leq c_{hj}m,
\end{equation*}%
where the constants $c_{hj}$ are finite. Hence the desired result
immediately follows from Markov inequality.
\end{proof}

\begin{lemma}
\label{approx-1}We assume that Assumptions \ref{b-shifts}-\ref{contamination}%
, and $d=O\left( m^{1/4}\right) $, hold. Then it holds that 
\begin{align*}
&r_{m}^{\eta -1/2}\max_{a_{m}\leq k\leq T_{m}}\frac{\left\vert \displaystyle%
\sum_{t=m+1}^{m+k}\widehat{\epsilon }_{t}\right\vert }{m^{1/2}\left( 1+%
\displaystyle\frac{k}{m}\right) \displaystyle\left( \frac{k}{k+m}\right)
^{\eta }} \\
=&r_{m}^{\eta -1/2}\max_{a_{m}\leq k\leq T_{m}}\frac{\left\vert \displaystyle%
\sum_{t=m+1}^{m+k}\epsilon _{t}-\displaystyle\frac{k}{m}\displaystyle%
\sum_{t=1}^{m}\epsilon _{t}\right\vert }{m^{1/2}\left( 1+\displaystyle\frac{k%
}{m}\right) \displaystyle\left( \frac{k}{k+m}\right) ^{\eta }}+o_{P}(1).
\end{align*}
\end{lemma}

Note that

\begin{align*}
\sum_{t=m+1}^{m+k}\widehat{\epsilon }_{t}& =\sum_{t=m+1}^{m+k}\left( y_{t}-%
\mathbf{x}_{t}^{\prime }\widehat{\beta }_{m}\right)
=\sum_{t=m+1}^{m+k}\left( \epsilon _{t}-\mathbf{x}_{t}^{\prime }\left( 
\widehat{\beta }_{m}-\beta _{0}\right) \right) \\
& =\sum_{t=m+1}^{m+k}\epsilon _{t}-\sum_{t=m+1}^{m+k}\mathbf{x}_{t}^{\prime
}\left( \sum_{t=1}^{m}\mathbf{x}_{t}\mathbf{x}_{t}^{\prime }\right)
^{-1}\left( \sum_{t=1}^{m}\mathbf{x}_{t}\epsilon _{t}\right) .
\end{align*}%
We begin by showing that%
\begin{equation*}
r_{m}^{\eta -1/2}\max_{a_{m}\leq k\leq T_{m}}\frac{\left\vert \displaystyle%
\sum_{t=m+1}^{m+k}\mathbf{x}_{t}^{\prime }\left( \displaystyle\left( \frac{1%
}{m}\sum_{t=1}^{m}\mathbf{x}_{t}\mathbf{x}_{t}^{\prime }\right) ^{-1}-%
\mathbf{C}^{-1}\right) \displaystyle\left( \frac{1}{m}\sum_{t=1}^{m}\mathbf{x%
}_{t}\epsilon _{t}\right) \right\vert }{m^{1/2}\left( 1+\displaystyle\frac{k%
}{m}\right) \displaystyle\left( \frac{k}{k+m}\right) ^{\eta }}=o_{P}(1).
\end{equation*}%
Indeed%
\begin{align*}
& r_{m}^{\eta -1/2}\max_{a_{m}\leq k\leq T_{m}}\frac{\left\vert \displaystyle%
\sum_{t=m+1}^{m+k}\mathbf{x}_{t}^{\prime }\left( \displaystyle\left( \frac{1%
}{m}\sum_{t=1}^{m}\mathbf{x}_{t}\mathbf{x}_{t}^{\prime }\right) ^{-1}-%
\mathbf{C}^{-1}\right) \displaystyle\left( \frac{1}{m}\sum_{t=1}^{m}\mathbf{x%
}_{t}\epsilon _{t}\right) \right\vert }{m^{1/2}\left( 1+\displaystyle\frac{k%
}{m}\right) \displaystyle\left( \frac{k}{k+m}\right) ^{\eta }} \\
\leq & r_{m}^{\eta -1/2}\max_{a_{m}\leq k\leq T_{m}}\frac{\left\Vert %
\displaystyle\sum_{t=m+1}^{m+k}\mathbf{x}_{t}\right\Vert \left\Vert %
\displaystyle\left( \frac{1}{m}\sum_{t=1}^{m}\mathbf{x}_{t}\mathbf{x}%
_{t}^{\prime }\right) ^{-1}-\mathbf{C}^{-1}\right\Vert _{F}\displaystyle%
\left\Vert \frac{1}{m}\sum_{t=1}^{m}\mathbf{x}_{t}\epsilon _{t}\right\Vert }{%
m^{1/2}\left( 1+\displaystyle\frac{k}{m}\right) \displaystyle\left( \frac{k}{%
k+m}\right) ^{\eta }}.
\end{align*}%
By stationarity implied by Assumption \ref{b-shifts}, and standard
arguments, it follows immediately that $\left\Vert \sum_{t=m+1}^{m+k}\mathbf{%
x}_{t}\right\Vert =O_{P}\left( d^{1/2}k\right) $. Further, note that, using
Taylor's expansion%
\begin{align*}
& \left( \frac{1}{m}\sum_{t=1}^{m}\mathbf{x}_{t}\mathbf{x}_{t}^{\prime
}\right) ^{-1}-\mathbf{C}^{-1} \\
=& \mathbf{C}^{-1}-\left( \frac{1}{m}\sum_{t=1}^{m}\mathbf{x}_{t}\mathbf{x}%
_{t}^{\prime }\pm \mathbf{C}\right) ^{-1}=\mathbf{C}^{-1}-\left( \left( 
\frac{1}{m}\sum_{t=1}^{m}\mathbf{x}_{t}\mathbf{x}_{t}^{\prime }\mathbf{C}%
^{-1}\pm \mathbf{I}_{d}\right) \mathbf{C}\right) ^{-1} \\
=& \mathbf{C}^{-1}-\mathbf{C}^{-1}\left( \mathbf{I}_{d}+\left( \frac{1}{m}%
\sum_{t=1}^{m}\mathbf{x}_{t}\mathbf{x}_{t}^{\prime }\mathbf{C}^{-1}-\mathbf{I%
}_{d}\right) \right) ^{-1} \\
=& \mathbf{C}^{-1}-\mathbf{C}^{-1}\left( \mathbf{I}_{d}-\left( \frac{1}{m}%
\sum_{t=1}^{m}\mathbf{x}_{t}\mathbf{x}_{t}^{\prime }-\mathbf{C}\right) 
\mathbf{C}^{-1}\right) +o\left( \left\Vert \mathbf{C}^{-1}\left( \frac{1}{m}%
\sum_{t=1}^{m}\mathbf{x}_{t}\mathbf{x}_{t}^{\prime }-\mathbf{C}\right) 
\mathbf{C}^{-1}\right\Vert _{F}\right) \\
=& \mathbf{C}^{-1}\left( \frac{1}{m}\sum_{t=1}^{m}\mathbf{x}_{t}\mathbf{x}%
_{t}^{\prime }-\mathbf{C}\right) \mathbf{C}^{-1}+o\left( \left\Vert \mathbf{C%
}^{-1}\left( \frac{1}{m}\sum_{t=1}^{m}\mathbf{x}_{t}\mathbf{x}_{t}^{\prime }-%
\mathbf{C}\right) \mathbf{C}^{-1}\right\Vert _{F}\right) .
\end{align*}%
Hence it holds that%
\begin{align*}
& \left\Vert \left( \frac{1}{m}\sum_{t=1}^{m}\mathbf{x}_{t}\mathbf{x}%
_{t}^{\prime }\right) ^{-1}-\mathbf{C}^{-1}\right\Vert _{F}-\left\Vert 
\mathbf{C}^{-1}\left( \frac{1}{m}\sum_{t=1}^{m}\mathbf{x}_{t}\mathbf{x}%
_{t}^{\prime }-\mathbf{C}\right) \mathbf{C}^{-1}\right\Vert _{F} \\
=& o\left( \left\Vert \mathbf{C}^{-1}\left( \frac{1}{m}\sum_{t=1}^{m}\mathbf{%
x}_{t}\mathbf{x}_{t}^{\prime }-\mathbf{C}\right) \mathbf{C}^{-1}\right\Vert
_{F}\right) ,
\end{align*}%
and%
\begin{equation*}
\left\Vert \mathbf{C}^{-1}\left( \frac{1}{m}\sum_{t=1}^{m}\mathbf{x}_{t}%
\mathbf{x}_{t}^{\prime }-\mathbf{C}\right) \mathbf{C}^{-1}\right\Vert
_{F}\leq s_{\max }^{2}\left( \mathbf{C}^{-1}\right) \left\Vert \frac{1}{m}%
\sum_{t=1}^{m}\mathbf{x}_{t}\mathbf{x}_{t}^{\prime }-\mathbf{C}\right\Vert
_{F}=O_{P}\left( \frac{d}{m^{1/2}}\right) ,
\end{equation*}%
using Lemma \ref{max-inequ-3} and seeing as $s_{\max }\left( \mathbf{C}%
^{-1}\right) =s_{\min }^{-1}\left( \mathbf{C}\right) $ and, by Assumption %
\ref{exogeneity}\textit{(iii)}, $s_{\min }\left( \mathbf{C}\right) >0$. By a
similar logic%
\begin{equation*}
\left\Vert \mathbf{C}^{-1}\left( \frac{1}{m}\sum_{t=1}^{m}\mathbf{x}_{t}%
\mathbf{x}_{t}^{\prime }-\mathbf{C}\right) \mathbf{C}^{-1}\right\Vert
_{F}\leq s_{\max }^{2}\left( \mathbf{C}^{-1}\right) \left\Vert \frac{1}{m}%
\sum_{t=1}^{m}\mathbf{x}_{t}\mathbf{x}_{t}^{\prime }-\mathbf{C}\right\Vert
_{F}=O_{P}\left( dm^{-1/2}\right) .
\end{equation*}%
Thus it finally follows that 
\begin{equation}
\left\Vert \left( \frac{1}{m}\sum_{t=1}^{m}\mathbf{x}_{t}\mathbf{x}%
_{t}^{\prime }\right) ^{-1}-\mathbf{C}^{-1}\right\Vert _{F}=O_{P}\left( 
\frac{d}{m^{1/2}}\right) .  \label{frob-inv}
\end{equation}%
Finally, using Lemma \ref{max-ineq-1} and (\ref{frob-inv})%
\begin{align*}
& r_{m}^{\eta -1/2}\max_{a_{m}\leq k\leq T_{m}}\frac{\displaystyle\left\Vert
\sum_{t=m+1}^{m+k}\mathbf{x}_{t}\right\Vert \left\Vert \left( \displaystyle%
\left( \frac{1}{m}\sum_{t=1}^{m}\mathbf{x}_{t}\mathbf{x}_{t}^{\prime
}\right) ^{-1}-\mathbf{C}^{-1}\right) \right\Vert _{F}\displaystyle%
\left\Vert \frac{1}{m}\sum_{t=1}^{m}\mathbf{x}_{t}\epsilon _{t}\right\Vert }{%
m^{1/2}\left( 1+\displaystyle\frac{k}{m}\right) \displaystyle\left( \frac{k}{%
k+m}\right) ^{\eta }} \\
=& O_{P}\left( 1\right) r_{m}^{\eta -1/2}\max_{a_{m}\leq k\leq T_{m}}\frac{%
\left( d^{1/2}k\right) \left( dm^{-1/2}\right) \left( d^{1/2}m^{-1/2}\right) 
}{m^{1/2}\left( 1+\displaystyle\frac{k}{m}\right) \displaystyle\left( \frac{k%
}{k+m}\right) ^{\eta }} \\
=& O_{P}\left( 1\right) r_{m}^{\eta -1/2}d^{2}m^{-1/2}\max_{a_{m}\leq k\leq
T_{m}}\left( \frac{k}{k+m}\right) ^{1-\eta }.
\end{align*}%
We note that, depending on the values of $\eta $, there are two bounds. If $%
\eta \leq 1$, then $\max_{a_{m}\leq k\leq T_{m}}\left( k/\left( m+k\right)
\right) ^{1-\eta }$ is bounded; conversely, if $\eta >1$, the maximum is
attained at $\left( a_{m}/\left( m+a_{m}\right) \right) ^{1-\eta }$. Hence,
putting all together, it finally holds that%
\begin{align*}
& r_{m}^{\eta -1/2}\max_{a_{m}\leq k\leq T_{m}}\frac{\left\vert
\sum_{t=m+1}^{m+k}\mathbf{x}_{t}^{\prime }\displaystyle\left( \frac{1}{m}%
\sum_{t=1}^{m}\mathbf{x}_{t}\mathbf{x}_{t}^{\prime }-\mathbf{C}\right) ^{-1}%
\displaystyle\left( \frac{1}{m}\sum_{t=1}^{m}\mathbf{x}_{t}\epsilon
_{t}\right) \right\vert }{m^{1/2}\left( 1+\displaystyle\frac{k}{m}\right) %
\displaystyle\left( \frac{k}{k+m}\right) ^{\eta }} \\
=& O_{P}\left( \frac{d^{2}}{m^{\min \left\{ 1,\eta \right\} }}a_{m}^{\max
\left\{ 1/2,\eta -1/2\right\} }\right) =o_{P}\left( 1\right) ,
\end{align*}%
on account of the fact that $d=O\left( m^{1/4}\right) $. Therefore it
follows that%
\begin{align*}
& r_{m}^{\eta -1/2}\max_{a_{m}\leq k\leq T_{m}}\frac{\left\vert \displaystyle%
\sum_{t=m+1}^{m+k}\mathbf{x}_{t}^{\prime }\left( \sum_{t=1}^{m}\mathbf{x}_{t}%
\mathbf{x}_{t}^{\prime }\right) ^{-1}\left( \sum_{t=1}^{m}\mathbf{x}%
_{t}\epsilon _{t}\right) \right\vert }{m^{1/2}\left( 1+\displaystyle\frac{k}{%
m}\right) \displaystyle\left( \frac{k}{k+m}\right) ^{\eta }} \\
=& r_{m}^{\eta -1/2}\max_{a_{m}\leq k\leq T_{m}}\frac{\left\vert %
\displaystyle\sum_{t=m+1}^{m+k}\mathbf{x}_{t}^{\prime }\mathbf{C}^{-1}%
\displaystyle\left( \frac{1}{m}\sum_{t=1}^{m}\mathbf{x}_{t}\epsilon
_{t}\right) \right\vert }{m^{1/2}\left( 1+\displaystyle\frac{k}{m}\right) %
\displaystyle\left( \frac{k}{k+m}\right) ^{\eta }}+o_{P}\left( 1\right) .
\end{align*}%
We now show that 
\begin{equation*}
r_{m}^{\eta -1/2}\max_{a_{m}\leq k\leq T_{m}}\frac{\left\vert \displaystyle%
\sum_{t=m+1}^{m+k}\left( \mathbf{x}_{t}-\mathbf{c}_{1}\right) ^{\prime }%
\mathbf{C}^{-1}\displaystyle\left( \frac{1}{m}\sum_{t=1}^{m}\mathbf{x}%
_{t}\epsilon _{t}\right) \right\vert }{m^{1/2}\left( 1+\displaystyle\frac{k}{%
m}\right) \displaystyle\left( \frac{k}{k+m}\right) ^{\eta }}=o_{P}\left(
1\right) .
\end{equation*}%
Indeed

\begin{align*}
r_{m}^{\eta -1/2}& \max_{a_{m}\leq k\leq T_{m}}\frac{\left\vert \displaystyle%
\sum_{t=m+1}^{m+k}\left( \mathbf{x}_{t}-\mathbf{c}_{1}\right) ^{\prime }%
\left[ \mathbf{C}^{-1}\displaystyle\left( \frac{1}{m}\sum_{t=1}^{m}\mathbf{x}%
_{t}\epsilon _{t}\right) \right] \right\vert }{m^{1/2}\left( 1+\displaystyle%
\frac{k}{m}\right) \displaystyle\left( \frac{k}{k+m}\right) ^{\eta }} \\
& \leq r_{m}^{\eta -1/2}\max_{a_{m}\leq k\leq T_{m}}\frac{\displaystyle%
\left\Vert \sum_{t=m+1}^{m+k}\left( \mathbf{x}_{t}-\mathbf{c}_{1}\right)
\right\Vert \left\Vert \mathbf{C}^{-1}\right\Vert _{F}\displaystyle%
\left\Vert \frac{1}{m}\sum_{t=1}^{m}\mathbf{x}_{t}\epsilon _{t}\right\Vert }{%
m^{1/2}\left( 1+\displaystyle\frac{k}{m}\right) \displaystyle\left( \frac{k}{%
k+m}\right) ^{\eta }}.
\end{align*}%
Using Lemmas \ref{max-ineq-1}-\ref{max-inequ-3}, this entails that 
\begin{align*}
r_{m}^{\eta -1/2}& \max_{a_{m}\leq k\leq T_{m}}\frac{\left\vert
\sum_{t=m+1}^{m+k}\left( \mathbf{x}_{t}-\mathbf{c}_{1}\right) ^{\prime }%
\left[ \mathbf{C}^{-1}\displaystyle\left( \frac{1}{m}\sum_{t=1}^{m}\mathbf{x}%
_{t}\epsilon _{t}\right) \right] \right\vert }{m^{1/2}\left( 1+\displaystyle%
\frac{k}{m}\right) \displaystyle\left( \frac{k}{k+m}\right) ^{\eta }} \\
& =O_{P}(1)r_{m}^{\eta -1/2}d\left\Vert \mathbf{C}^{-1}\right\Vert
_{F}\max_{a_{m}\leq k\leq T_{m}}\frac{k^{\zeta _{3}}}{\left( 1+\displaystyle%
\frac{k}{m}\right) \displaystyle\left( \frac{k}{k+m}\right) ^{\eta }} \\
& =O_{P}(1)r_{m}^{\eta -1/2}d\left\Vert \mathbf{C}^{-1}\right\Vert
_{F}\max_{a_{m}\leq k\leq T_{m}}\left( \frac{k}{k+m}\right) ^{\zeta
_{3}-\eta }\max_{a_{m}\leq k\leq T_{m}}\frac{1}{\left( m+k\right) ^{1-\zeta
_{3}}}.
\end{align*}%
By virtue of Lemma \ref{max-ineq-2}, we can choose $1/2<\zeta _{3}<\min
\{1,\eta \}$ in the above, so that 
\begin{align*}
& r_{m}^{\eta -1/2}\max_{a_{m}\leq k\leq T_{m}}\frac{\left\vert \displaystyle%
\sum_{t=m+1}^{m+k}\left( \mathbf{x}_{t}-\mathbf{c}_{1}\right) ^{\prime }%
\left[ \mathbf{C}^{-1}\displaystyle\left( \frac{1}{m}\sum_{t=1}^{m}\mathbf{x}%
_{t}\epsilon _{t}\right) \right] \right\vert }{m^{1/2}\left( 1+\displaystyle%
\frac{k}{m}\right) \displaystyle\left( \frac{k}{k+m}\right) ^{\eta }} \\
& =O_{P}(1)r_{m}^{\eta -1/2}dm^{\zeta _{3}-1}\left\Vert \mathbf{C}%
^{-1}\right\Vert _{F}\left( \frac{a_{m}}{m}\right) ^{\zeta _{3}-\eta } \\
& =O_{P}(1)\left\Vert \mathbf{C}^{-1}\right\Vert _{F}dm^{-1/2}a_{m}^{\zeta
_{3}-1/2}.
\end{align*}%
In order to bound $\left\Vert \mathbf{C}^{-1}\right\Vert _{F}$, we note that%
\begin{equation*}
\left\Vert \mathbf{C}^{-1}\right\Vert _{F}\leq d^{1/2}\left\Vert \mathbf{C}%
^{-1}\right\Vert _{2}=d^{1/2}\frac{1}{s_{\min }\left( \mathbf{C}\right) },
\end{equation*}%
where $\left\Vert \mathbf{C}^{-1}\right\Vert _{2}$ is the $L_{2}$-norm of a
matrix. Again by Assumption \ref{exogeneity}\textit{(iii)}, $s_{\min }\left( 
\mathbf{C}\right) >0$, which entails that there exists a finite constant $%
c_{0}$ such that 
\begin{equation}
\left\Vert \mathbf{C}^{-1}\right\Vert _{F}\leq c_{0}d^{1/2}.  \label{c-frob}
\end{equation}%
Hence it holds that%
\begin{equation*}
r_{m}^{\eta -1/2}\max_{a_{m}\leq k\leq T_{m}}\frac{\left\vert \displaystyle%
\sum_{t=m+1}^{m+k}\left( \mathbf{x}_{t}-\mathbf{c}_{1}\right) ^{\prime }%
\left[ \mathbf{C}^{-1}\displaystyle\left( \frac{1}{m}\sum_{t=1}^{m}\mathbf{x}%
_{t}\epsilon _{t}\right) \right] \right\vert }{m^{1/2}\left( 1+\displaystyle%
\frac{k}{m}\right) \displaystyle\left( \frac{k}{k+m}\right) ^{\eta }}%
=O_{P}\left( \frac{d^{3/2}}{m^{1/2}}a_{m}^{\zeta _{3}-1/2}\right) ;
\end{equation*}%
under the constraint $d=O\left( m^{1/4}\right) $, this finally entails that 
\begin{align*}
&r_{m}^{\eta -1/2}\max_{a_{m}\leq k\leq T_{m}}\frac{\left\vert \displaystyle%
\sum_{t=m+1}^{m+k}\mathbf{x}_{t}^{\prime }\left[ \mathbf{C}^{-1}\displaystyle%
\left( \frac{1}{m}\sum_{t=1}^{m}\mathbf{x}_{t}\epsilon _{t}\right) \right]
\right\vert }{m^{1/2}\left( 1+\displaystyle\frac{k}{m}\right) \displaystyle%
\left( \frac{k}{k+m}\right) ^{\eta }} \\
=&r_{m}^{\eta -1/2}\max_{a_{m}\leq k\leq T_{m}}\frac{\left\vert \displaystyle%
\frac{k}{m}\mathbf{c}_{1}^{\prime }\mathbf{C}^{-1}\displaystyle\left(
\sum_{t=1}^{m}\mathbf{x}_{t}\epsilon _{t}\right) \right\vert }{m^{1/2}\left(
1+\displaystyle\frac{k}{m}\right) \displaystyle\left( \frac{k}{k+m}\right)
^{\eta }}+o_{P}\left( 1\right) .
\end{align*}%
By construction, $\mathbf{c}_{1}^{\prime }\mathbf{C}^{-1}=\left(
1,0,...,0\right) ^{\prime }$, and therefore it finally follows that%
\begin{equation*}
r_{m}^{\eta -1/2}\max_{a_{m}\leq k\leq T_{m}}\frac{\left\vert \displaystyle%
\sum_{t=m+1}^{m+k}\mathbf{x}_{t}^{\prime }\left( \displaystyle\sum_{t=1}^{m}%
\mathbf{x}_{t}\mathbf{x}_{t}^{\prime }\right) ^{-1}\left( \displaystyle%
\sum_{t=1}^{m}\mathbf{x}_{t}\epsilon _{t}\right) -\displaystyle\frac{k}{m}%
\displaystyle\sum_{t=1}^{m}\epsilon _{t}\right\vert }{m^{1/2}\left( 1+%
\displaystyle\frac{k}{m}\right) \displaystyle\left( \frac{k}{k+m}\right)
^{\eta }}=o_{P}\left( 1\right) .
\end{equation*}%
The desired result now follows from putting all together.


\begin{lemma}
\label{approx-2}We assume that Assumptions \ref{b-shifts}\textit{(i)} and %
\ref{horizon} hold. Then, on a suitably enlarged probability space, it is
possible to define two independent standard Wiener processes $\left\{
W_{1,m}\left( k\right) ,k\geq 1\right\} $ and $\left\{ W_{2,m}\left(
k\right) ,k\geq 1\right\} $ such that%
\begin{align*}
&r_{m}^{\eta -1/2}\max_{a_{m}\leq k\leq T_{m}}\frac{1}{\sigma }\frac{%
\left\vert \displaystyle\sum_{t=m+1}^{m+k}\epsilon _{t}-\displaystyle\frac{k%
}{m}\displaystyle\sum_{t=1}^{m}\epsilon _{t}\right\vert }{m^{1/2}\left( 1+%
\displaystyle\frac{k}{m}\right) \displaystyle\left( \frac{k}{k+m}\right)
^{\eta }} \\
=&r_{m}^{\eta -1/2}\max_{a_{m}\leq k\leq T_{m}}\frac{\left\vert
W_{1,m}\left( k\right) -\displaystyle\frac{k}{m}W_{2,m}\left( m\right)
\right\vert }{m^{1/2}\left( 1+\displaystyle\frac{k}{m}\right) \displaystyle%
\left( \frac{k}{k+m}\right) ^{\eta }}+o_{P}\left( 1\right) .
\end{align*}
\end{lemma}

Henceforth, we set $\sigma =1$ for simplicity and without loss of
generality. Standard arguments entail that 
\begin{align*}
r_{m}^{\eta -1/2}\max_{a_{m}\leq k\leq T_{m}}& \frac{\left\vert \left( %
\displaystyle\sum_{t=m+1}^{m+k}\epsilon _{t}-\displaystyle\frac{k}{m}%
\displaystyle\sum_{t=1}^{m}\epsilon _{t}\right) -\left( W_{1,m}(k)-%
\displaystyle\frac{k}{m}W_{2,m}(m)\right) \right\vert }{m^{1/2}\left( 1+%
\displaystyle\frac{k}{m}\right) \displaystyle\left( \frac{k}{m+k}\right)
^{\eta }} \\
& =r_{m}^{\eta -1/2}\max_{a_{m}\leq k\leq T_{m}}\frac{\left\vert \left( %
\displaystyle\sum_{t=m+1}^{m+k}\epsilon _{t}-W_{1,m}(k)\right) -\displaystyle%
\frac{k}{m}\left( \displaystyle\sum_{t=1}^{m}\epsilon _{t}-W_{2,m}(m)\right)
\right\vert }{m^{1/2}\left( 1+\displaystyle\frac{k}{m}\right) \displaystyle%
\left( \frac{k}{m+k}\right) ^{\eta }} \\
& \leq r_{m}^{\eta -1/2}\max_{a_{m}\leq k\leq T_{m}}\frac{\left\vert %
\displaystyle\sum_{t=m+1}^{m+k}\epsilon _{t}-W_{1,m}(k)\right\vert }{%
m^{1/2}\left( 1+\displaystyle\frac{k}{m}\right) \displaystyle\left( \frac{k}{%
m+k}\right) ^{\eta }} \\
& +r_{m}^{\eta -1/2}\max_{a_{m}\leq k\leq T_{m}}\frac{\left\vert %
\displaystyle\frac{k}{m}\left( \displaystyle\sum_{t=1}^{m}\epsilon
_{t}-W_{2,m}(m)\right) \right\vert }{m^{1/2}\left( 1+\displaystyle\frac{k}{m}%
\right) \displaystyle\left( \frac{k}{m+k}\right) ^{\eta }}.
\end{align*}%
Using (\ref{SIP1}) in Lemma \ref{sip}, it follows that 
\begin{align*}
r_{m}^{\eta -1/2}\max_{a_{m}\leq k\leq T_{m}}& \frac{\left\vert \displaystyle%
\sum_{t=m+1}^{m+k}\epsilon _{t}-W_{1,m}(k)\right\vert }{m^{1/2}\left( 1+%
\displaystyle\frac{k}{m}\right) \displaystyle\left( \frac{k}{m+k}\right)
^{\eta }} \\
& =r_{m}^{\eta -1/2}\max_{a_{m}\leq k\leq T_{m}}k^{\zeta _{1}}\frac{%
k^{-\zeta _{1}}\left\vert \displaystyle\sum_{t=m+1}^{m+k}\epsilon
_{t}-W_{1,m}(k)\right\vert }{m^{1/2}\left( 1+\displaystyle\frac{k}{m}\right) %
\displaystyle\left( \frac{k}{m+k}\right) ^{\eta }} \\
& =O_{P}\left( 1\right) r_{m}^{\eta -1/2}\max_{a_{m}\leq k\leq T_{m}}\frac{%
k^{\zeta _{1}}}{m^{1/2}\left( 1+\displaystyle\frac{k}{m}\right) \displaystyle%
\left( \frac{k}{m+k}\right) ^{\eta }} \\
& =O_{P}\left( 1\right) r_{m}^{\eta -1/2}\max_{a_{m}\leq k\leq
T_{m}}k^{\zeta _{1}-1/2}\left( \frac{m+k}{k}\right) ^{\eta -1/2} \\
& =O_{P}\left( a_{m}^{\zeta _{1}-1/2}\right) =o_{P}\left( 1\right) .
\end{align*}%
Similarly, using (\ref{SIP2}) in Lemma \ref{sip}

\begin{align*}
& r_{m}^{\eta -1/2}\max_{a_{m}\leq k\leq T_{m}}\frac{k}{m}\frac{\left\vert
\left( \displaystyle\sum_{t=1}^{m}\epsilon _{t}-W_{2,m}(m)\right)
\right\vert }{m^{1/2}\left( 1+\displaystyle\frac{k}{m}\right) \displaystyle%
\left( \frac{k}{m+k}\right) ^{\eta }} \\
& =r_{m}^{\eta -1/2}\max_{a_{m}\leq k\leq T_{m}}\frac{k}{m}\frac{m^{\zeta
_{2}}}{m^{\zeta _{2}}}\frac{\left\vert \left( \displaystyle%
\sum_{t=1}^{m}\epsilon _{t}-W_{2,m}(m)\right) \right\vert }{m^{1/2}\left( 1+%
\displaystyle\frac{k}{m}\right) \displaystyle\left( \frac{k}{m+k}\right)
^{\eta }} \\
& =O_{P}\left( 1\right) r_{m}^{\eta -1/2}\max_{a_{m}\leq k\leq T_{m}}\frac{%
km^{\zeta _{2}}}{m^{3/2}\left( 1+\displaystyle\frac{k}{m}\right) %
\displaystyle\left( \frac{k}{m+k}\right) ^{\eta }}.
\end{align*}%
Hence, after some algebra, it can be shown that%
\begin{equation*}
r_{m}^{\eta -1/2}\max_{a_{m}\leq k\leq T_{m}}\frac{k}{m}\frac{\left\vert
\left( \displaystyle\sum_{t=1}^{m}\epsilon _{t}-W_{2,m}(m)\right)
\right\vert }{m^{1/2}\left( 1+\displaystyle\frac{k}{m}\right) \displaystyle%
\left( \frac{k}{m+k}\right) ^{\eta }}=O_{P}\left( r_{m}^{\min \left\{
1/2,\eta -1/2\right\} }m^{\zeta _{2}-1/2}\right) =o_{P}\left( 1\right) .
\end{equation*}%
The desired result now follows.

\clearpage
\newpage

\renewcommand*{\thesection}{\Alph{section}}

\setcounter{equation}{0} \setcounter{lemma}{0} \setcounter{theorem}{0} %
\renewcommand{\theassumption}{B.\arabic{assumption}} 
\renewcommand{\thetheorem}{B.\arabic{theorem}} \renewcommand{\thelemma}{B.%
\arabic{lemma}} \renewcommand{\theproposition}{B.\arabic{proposition}} %
\renewcommand{\thecorollary}{B.\arabic{corollary}} \renewcommand{%
\theequation}{B.\arabic{equation}}

\section{Proofs\label{proofs}}

\begin{proof}[Proof of Theorem \protect\ref{nul-distr}]
Lemmas \ref{sip}-\ref{approx-2}, put together, entail that%
\begin{align*}
& r_{m}^{\eta -1/2}\max_{a_{m}\leq k\leq T_{m}}\frac{1}{\sigma }\frac{%
\left\vert \displaystyle\sum_{t=m+1}^{m+k}\widehat{\epsilon }_{t}\right\vert 
}{m^{1/2}\left( 1+\displaystyle\frac{k}{m}\right) \displaystyle\left( \frac{k%
}{k+m}\right) ^{\eta }} \\
=& r_{m}^{\eta -1/2}\max_{a_{m}\leq k\leq T_{m}}\frac{\left\vert
W_{1,m}\left( k\right) -\displaystyle\frac{k}{m}W_{2,m}\left( m\right)
\right\vert }{m^{1/2}\left( 1+\displaystyle\frac{k}{m}\right) \displaystyle%
\left( \frac{k}{k+m}\right) ^{\eta }}+o_{P}\left( 1\right) .
\end{align*}%
As in the other proofs, we will set $\sigma =1$ for simplicity and without
loss of generality. We begin by noting that the distribution of $%
W_{1,m}\left( \cdot \right) $ and $W_{2,m}\left( \cdot \right) $ does not
depend on $m$, so that%
\begin{align*}
& r_{m}^{\eta -1/2}\max_{a_{m}\leq k\leq T_{m}}\frac{\left\vert
W_{1,m}\left( k\right) -\displaystyle\frac{k}{m}W_{2,m}\left( m\right)
\right\vert }{m^{1/2}\left( 1+\displaystyle\frac{k}{m}\right) \displaystyle%
\left( \frac{k}{k+m}\right) ^{\eta }} \\
& \overset{\mathcal{D}}{=}r_{m}^{\eta -1/2}\max_{a_{m}\leq k\leq T_{m}}\frac{%
\left\vert W_{1}\left( k\right) -\displaystyle\frac{k}{m}W_{2}\left(
m\right) \right\vert }{m^{1/2}\left( 1+\displaystyle\frac{k}{m}\right) %
\displaystyle\left( \frac{k}{k+m}\right) ^{\eta }} \\
& \overset{\mathcal{D}}{=}r_{m}^{\eta -1/2}\max_{a_{m}\leq k\leq T_{m}}\frac{%
\left\vert W_{1}\displaystyle\left( \frac{k}{m}\right) -\displaystyle\frac{k%
}{m}W_{2}\left( 1\right) \right\vert }{\left( 1+\displaystyle\frac{k}{m}%
\right) \displaystyle\left( \frac{k}{k+m}\right) ^{\eta }}.
\end{align*}%
By standard arguments, it is easy to see that%
\begin{equation}
W_{1}\left( t\right) -tW_{2}\left( 1\right) \overset{\mathcal{D}}{=}\left(
1+t\right) W\left( \frac{t}{1+t}\right) ,  \label{process}
\end{equation}%
for all $t$, where $W\left( \cdot \right) $ is a standard Wiener.\footnote{%
Indeed, this can be verified readily by noting that both processes $%
W_{1}\left( s\right) $ and $W_{2}\left( s\right) $ have mean zero, are
Gaussian, and have covariance kernel given by%
\begin{align*}
& E\left( W_{1}\left( t\right) -tW_{2}\left( 1\right) \right) \left(
W_{1}\left( s\right) -sW_{2}\left( 1\right) \right) \\
=& E\left( W_{1}\left( t\right) W_{1}\left( s\right) \right) -stE\left(
W_{2}\left( 1\right) \right) ^{2} \\
=& s\wedge t-st,
\end{align*}%
and 
\begin{align*}
& E\left( \left( 1+t\right) W\left( \frac{t}{1+t}\right) \left( 1+s\right)
W\left( \frac{s}{1+s}\right) \right) \\
=& \left( 1+t\right) \left( 1+s\right) \left( \frac{t}{1+t}\wedge \frac{s}{%
1+s}\right) ,
\end{align*}%
respectively - it is now not hard to see that the two covariance kernels are
the same.} Hence we will study%
\begin{align*}
& r_{m}^{\eta -1/2}\max_{a_{m}\leq k\leq T_{m}}\frac{\left\vert \left( 1+%
\displaystyle\frac{k}{m}\right) W\displaystyle\left( \frac{\frac{k}{m}}{1+%
\frac{k}{m}}\right) \right\vert }{\left( 1+\displaystyle\frac{k}{m}\right) %
\displaystyle\left( \frac{k}{k+m}\right) ^{\eta }} \\
& \overset{\mathcal{D}}{=}r_{m}^{\eta -1/2}\max_{a_{m}/m\leq t\leq T_{m}/m}%
\frac{\left\vert W\displaystyle\left( \frac{t}{1+t}\right) \right\vert }{%
\displaystyle\left( \frac{t}{1+t}\right) ^{\eta }} \\
& \overset{\mathcal{D}}{=}r_{m}^{\eta -1/2}\max_{a_{m}/\left( m+a_{m}\right)
\leq s\leq T_{m}/\left( m+T_{m}\right) }\frac{\left\vert W\left( s\right)
\right\vert }{s^{\eta }}.
\end{align*}

Define%
\begin{equation*}
u=s\frac{a_{m}+m}{a_{m}};
\end{equation*}%
it follows that%
\begin{align*}
& r_{m}^{\eta -1/2}\max_{a_{m}/\left( m+a_{m}\right) \leq s\leq T_{m}/\left(
m+T_{m}\right) }\frac{\left\vert W\left( s\right) \right\vert }{s^{\eta }} \\
& \overset{\mathcal{D}}{=}r_{m}^{\eta -1/2}\max_{1\leq u\leq T_{m}\left(
m+a_{m}\right) /\left( a_{m}\left( m+T_{m}\right) \right) }\frac{\left\vert W%
\displaystyle\left( u\frac{a_{m}}{a_{m}+m}\right) \right\vert }{\displaystyle%
\left( u\frac{a_{m}}{a_{m}+m}\right) ^{\eta }} \\
& \overset{\mathcal{D}}{=}r_{m}^{\eta -1/2}\left( \frac{a_{m}}{a_{m}+m}%
\right) ^{1/2-\eta }\max_{1\leq u\leq T_{m}\left( m+a_{m}\right) /\left(
a_{m}\left( m+T_{m}\right) \right) }\frac{\left\vert W\left( u\right)
\right\vert }{u^{\eta }} \\
& \overset{\mathcal{D}}{=}\max_{1\leq u\leq T_{m}\left( m+a_{m}\right)
/\left( a_{m}\left( m+T_{m}\right) \right) }\frac{\left\vert W\left(
u\right) \right\vert }{u^{\eta }},
\end{align*}%
having used the scale transformation of the Wiener process in the third
passage, and the definition of $r_{m}$ in the last one. Recalling (\ref{a-m}%
), by elementary arguments it holds that%
\begin{equation*}
\lim_{m\rightarrow \infty }\frac{T_{m}\left( m+a_{m}\right) }{a_{m}\left(
m+T_{m}\right) }\rightarrow \infty ,
\end{equation*}%
whence, as $m\rightarrow \infty $, it follows that, by continuity 
\begin{equation*}
\max_{1\leq u\leq T_{m}\left( m+a_{m}\right) /\left( a_{m}\left(
m+T_{m}\right) \right) }\frac{\left\vert W\left( u\right) \right\vert }{%
u^{\eta }}\overset{a.s.}{\rightarrow }\sup_{1\leq u<\infty }\frac{\left\vert
W(u)\right\vert }{u^{\eta }}.
\end{equation*}

As a final remark, we reconsider the proof, showing that the limiting
distribution is determined only by the observations very close to $a_{m}$.
To show this, we divide into two intervals the domain where the maximum is
taken defining%
\begin{align}
& \Theta _{m,1}=r_{m}^{\eta -1/2}\max_{a_{m}/\left( m+a_{m}\right) \leq
s<\left( a_{m}/\left( m+a_{m}\right) \right) ^{1-\epsilon }}\frac{\left\vert
W\left( s\right) \right\vert }{s^{\eta }},  \label{theta-1} \\
& \Theta _{m,2}=r_{m}^{\eta -1/2}\max_{\left( a_{m}/\left( m+a_{m}\right)
\right) ^{1-\epsilon }\leq s<T_{m}/\left( m+T_{m}\right) }\frac{\left\vert
W\left( s\right) \right\vert }{s^{\eta }}.  \label{theta-2}
\end{align}%
for some $\epsilon >0$ such that $\left( a_{m}/\left( m+a_{m}\right) \right)
^{1-\epsilon }=o\left( T_{m}/\left( m+T_{m}\right) \right) $. We begin by
studying $\Theta _{m,2}$. Repeatedly using the scale transformation for
Wiener process, we have%
\begin{align*}
& \Theta _{m,2}\overset{\mathcal{D}}{=}r_{m}^{\eta -1/2}\max_{\left(
a_{m}/\left( m+a_{m}\right) \right) ^{1-\epsilon }\leq s\leq T_{m}/\left(
m+T_{m}\right) }\frac{s^{1/2}}{s^{1/2}}\frac{\left\vert W\left( s\right)
\right\vert }{s^{\eta }} \\
& \leq r_{m}^{\eta -1/2}\max_{\left( a_{m}/\left( m+a_{m}\right) \right)
^{1-\epsilon }\leq s\leq T_{m}/\left( m+T_{m}\right) }\frac{\left\vert
W\left( s\right) \right\vert }{s^{1/2}}\max_{\left( a_{m}/\left(
m+a_{m}\right) \right) ^{1-\epsilon }\leq s\leq T_{m}/\left( m+T_{m}\right)
}s^{1/2-\eta } \\
& \leq \left( \left( \frac{a_{m}}{a_{m}+m}\right) ^{1/2}\right) ^{\eta
-1/2}\max_{\left( a_{m}/\left( m+a_{m}\right) \right) ^{1-\epsilon }\leq
s\leq T_{m}/\left( m+T_{m}\right) }\frac{\left\vert W\left( s\right)
\right\vert }{s^{1/2}}.
\end{align*}

We note that%
\begin{equation*}
\max_{\left( a_{m}/\left( m+a_{m}\right) \right) ^{1-\epsilon }\leq s\leq
T_{m}/\left( m+T_{m}\right) }\frac{\left\vert W\left( s\right) \right\vert }{%
s^{1/2}}\overset{\mathcal{D}}{=}\max_{1\leq u\leq \left( T_{m}\left(
m+a_{m}\right) ^{1-\epsilon }\right) /\left( \left( m+T_{m}\right)
a_{m}^{1-\epsilon }\right) }\frac{\left\vert W\left( u\right) \right\vert }{%
u^{1/2}},
\end{equation*}%
so that the Law of the Iterated Logarithm entails that%
\begin{equation*}
\max_{\left( a_{m}/\left( m+a_{m}\right) \right) ^{1-\epsilon }\leq s\leq
T_{m}/\left( m+T_{m}\right) }\frac{\left\vert W\left( s\right) \right\vert }{%
s^{1/2}}=O_{P}\left( \sqrt{\ln \ln \frac{T_{m}\left( m+a_{m}\right)
^{1-\epsilon }}{\left( m+T_{m}\right) a_{m}^{1-\epsilon }}}\right)
=O_{P}\left( \sqrt{\ln \ln \left( \frac{m}{a_{m}}\right) ^{1-\epsilon }}%
\right) .
\end{equation*}%
Hence, recalling (\ref{a-m})%
\begin{equation*}
\Theta _{m,2}=O_{P}\left( \left( \left( \frac{a_{m}}{a_{m}+m}\right)
^{1/2}\right) ^{\eta -1/2}\sqrt{\ln \ln \left( \frac{m}{a_{m}}\right)
^{1-\epsilon }}\right) =o_{P}\left( 1\right) .
\end{equation*}%
Using this result in conjunction with Lemmas \ref{sip}-\ref{approx-2} it
follows that, as $m\rightarrow \infty $%
\begin{equation*}
P\left( r_{m}^{\eta -1/2}\max_{a_{m}\leq k\leq T_{m}}\frac{1}{\sigma }\frac{%
\left\vert \displaystyle\sum_{t=m+1}^{m+k}\widehat{\epsilon }_{t}\right\vert 
}{m^{1/2}\left( 1+\displaystyle\frac{k}{m}\right) \displaystyle\left( \frac{k%
}{k+m}\right) ^{\eta }}=\Theta _{m,1}\right) =1.
\end{equation*}%
We now conclude our remarks by studying $\Theta _{m,1}$ in (\ref{theta-1}).
Let 
\begin{equation*}
u=s\frac{a_{m}+m}{a_{m}}.
\end{equation*}%
Then it holds that%
\begin{align*}
\Theta _{m,1}=& r_{m}^{\eta -1/2}\max_{a_{m}/\left( m+a_{m}\right) \leq
s\leq \left( a_{m}/\left( m+a_{m}\right) \right) ^{1-\epsilon }}\frac{%
\left\vert W\left( s\right) \right\vert }{s^{\eta }} \\
\overset{\mathcal{D}}{=}& r_{m}^{\eta -1/2}\max_{1\leq u\leq \left( \left(
m+a_{m}\right) /a_{m}\right) ^{\epsilon }}\frac{\left\vert W\displaystyle%
\left( u\frac{a_{m}}{a_{m}+m}\right) \right\vert }{\displaystyle\left( u%
\frac{a_{m}}{a_{m}+m}\right) ^{\eta }} \\
\overset{\mathcal{D}}{=}& r_{m}^{\eta -1/2}\left( \frac{a_{m}}{a_{m}+m}%
\right) ^{1/2-\eta }\max_{1\leq u\leq \left( \left( m+a_{m}\right)
/a_{m}\right) ^{\epsilon }}\frac{\left\vert W\left( u\right) \right\vert }{%
u^{\eta }} \\
\overset{\mathcal{D}}{=} & \max_{1\leq u\leq \left( \left( m+a_{m}\right)
/a_{m}\right) ^{\epsilon }}\frac{\left\vert W\left( u\right) \right\vert }{%
u^{\eta }},
\end{align*}%
and%
\begin{equation*}
\Theta _{m,1}\overset{\mathcal{D}}{=}\max_{1\leq u\leq \left( \left(
m+a_{m}\right) /a_{m}\right) ^{\epsilon }}\frac{\left\vert W\left( u\right)
\right\vert }{u^{\eta }}\overset{a.s.}{\rightarrow }\sup_{1\leq u<\infty }%
\frac{\left\vert W(u)\right\vert }{u^{\eta }}.
\end{equation*}
\end{proof}

\begin{proof}[Proof of Lemma \protect\ref{lrv}]
The proof of the lemma is standard, and we only report its main passages,
mainly to highlight the impact of $d$. Let $\gamma _{j}=E\left( \epsilon
_{0}\epsilon _{j}\right) $. By Assumption \ref{b-shifts}, it is easy to see
that%
\begin{equation*}
\sigma ^{2}=\lim_{m\rightarrow \infty }\sum_{t,s=1}^{m}E\left( \epsilon
_{t}\epsilon _{s}\right) =\gamma _{0}+2\sum_{j=1}^{\infty }\gamma _{j}.
\end{equation*}%
We begin by showing that the $\gamma _{j}$s satisfy%
\begin{equation}
\sum_{j=1}^{\infty }j\left\vert \gamma _{j}\right\vert <\infty .
\label{summable}
\end{equation}%
Indeed, let $\epsilon _{j}=h\left( \eta _{j},\eta _{j-1},...\right) $, and
define $\epsilon _{j,j}=h\left( \eta _{j},\eta _{j-1},...,\eta _{1},%
\widetilde{\eta }_{0,j,j},\widetilde{\eta }_{-1,j,j},...\right) $ where, as
per Definition \ref{bernoulli}, $\left\{ \widetilde{\eta }_{t,j,j},-\infty
<t,j<\infty \right\} $ \textit{i.i.d.} copies of $\eta _{0}$ independent of $%
\left\{ \eta _{t},-\infty <t<\infty \right\} $. By construction, $\epsilon
_{0}$ and $\epsilon _{j,j}$ are independent and have mean zero, so that $%
E\left( \epsilon _{0}\epsilon _{j,j}\right) =0$ and 
\begin{gather*}
\sum_{j=1}^{\infty }j\left\vert E\left( \epsilon _{0}\epsilon _{j}\right)
\right\vert =\sum_{j=1}^{\infty }j\left\vert E\left( \epsilon _{0}\left(
\epsilon _{j}-\epsilon _{j,j}\right) \right) \right\vert \leq
\sum_{j=1}^{\infty }j\left\vert \epsilon _{0}\right\vert _{2}\left\vert
\epsilon _{j}-\epsilon _{j,j}\right\vert _{2} \\
\leq c_{0}\sum_{j=1}^{\infty }j\left\vert \epsilon _{j}-\epsilon
_{j,j}\right\vert _{2}\leq c_{0}\sum_{j=1}^{\infty }j\left\vert \epsilon
_{j}-\epsilon _{j,j}\right\vert _{4}\leq c_{0}\sum_{j=1}^{\infty
}j^{1-a}\leq c_{1},
\end{gather*}%
having used: the Cauchy-Schwartz inequality in the third passage; the fact
that, by Assumption \ref{b-shifts}\textit{(i)} $\left\vert \epsilon
_{0}\right\vert _{2}<\infty $ in the fourth passage; the $L_{p}$-norm
inequality in the fifth passage; and Assumption \ref{b-shifts}\textit{(i)}
in the last one, recalling that $a>2$. We now write%
\begin{align*}
\widehat{\sigma }_{m}^{2}-\sigma ^{2} =&\left( \widehat{\gamma }_{0}-\gamma
_{0}\right) +2\sum_{j=1}^{H}\left( 1-\frac{j}{H+1}\right) \left( \widehat{%
\gamma }_{j}-\gamma _{j}\right) +2\sum_{j=1}^{H}\frac{j}{H+1}\gamma
_{j}-2\sum_{j=H+1}^{\infty }\gamma _{j} \\
=&I+II+III+IV.
\end{align*}%
By (\ref{summable}), it follows that%
\begin{equation*}
\frac{1}{2}III\leq \frac{1}{H+1}\sum_{j=1}^{\infty }j\left\vert \gamma
_{j}\right\vert \leq c_{0}\frac{1}{H+1};
\end{equation*}%
similarly%
\begin{equation*}
\frac{1}{2}IV\leq \sum_{j=H+1}^{\infty }\frac{j}{H+1}\left\vert \gamma
_{j}\right\vert \leq \frac{1}{H+1}\sum_{j=1}^{\infty }j\left\vert \gamma
_{j}\right\vert \leq c_{0}\frac{1}{H+1},
\end{equation*}%
so that $III+IV=O\left( H^{-1}\right) $. We now turn to $I$ and $II$\ noting
that%
\begin{align*}
&\widehat{\gamma }_{j}-\gamma _{j}=\frac{1}{m}\sum_{t=j+1}^{m}\widehat{%
\epsilon }_{t}\widehat{\epsilon }_{t-j}-E\left( \epsilon _{0}\epsilon
_{j}\right) \\
=&\frac{1}{m}\sum_{t=j+1}^{m}\left( \epsilon _{t}\epsilon _{t-j}-E\left(
\epsilon _{0}\epsilon _{j}\right) \right) +\left( \beta -\widehat{\beta }%
_{m}\right) ^{\prime }\frac{1}{m}\sum_{t=j+1}^{m}\mathbf{x}_{t}\epsilon
_{t-j} \\
&+\left( \beta -\widehat{\beta }_{m}\right) ^{\prime }\frac{1}{m}%
\sum_{t=j+1}^{m}\mathbf{x}_{t-j}\epsilon _{t}+\left( \widehat{\beta }%
_{m}-\beta \right) ^{\prime }\left( \frac{1}{m}\sum_{t=j+1}^{m}\mathbf{x}_{t}%
\mathbf{x}_{t-j}^{\prime }\right) \left( \widehat{\beta }_{m}-\beta \right) ,
\end{align*}%
whence%
\begin{align*}
&\sum_{j=1}^{H}\left( 1-\frac{j}{H+1}\right) \left( \widehat{\gamma }%
_{j}-\gamma _{j}\right) \\
=&\sum_{j=1}^{H}\left( 1-\frac{j}{H+1}\right) \left( \frac{1}{m}%
\sum_{t=j+1}^{m}\left( \epsilon _{t}\epsilon _{t-j}-E\left( \epsilon
_{0}\epsilon _{j}\right) \right) \right) \\
&+\left( \beta -\widehat{\beta }_{m}\right) ^{\prime }\left[
\sum_{j=1}^{H}\left( 1-\frac{j}{H+1}\right) \left( \frac{1}{m}%
\sum_{t=j+1}^{m}\mathbf{x}_{t}\epsilon _{t-j}\right) \right] \\
&+\left( \beta -\widehat{\beta }_{m}\right) ^{\prime }\left[
\sum_{j=1}^{H}\left( 1-\frac{j}{H+1}\right) \left( \frac{1}{m}%
\sum_{t=j+1}^{m}\mathbf{x}_{t-j}\epsilon _{t}\right) \right] \\
&+\left( \widehat{\beta }_{m}-\beta \right) ^{\prime }\left[
\sum_{j=1}^{H}\left( 1-\frac{j}{H+1}\right) \left( \frac{1}{m}%
\sum_{t=j+1}^{m}\mathbf{x}_{t}\mathbf{x}_{t-j}^{\prime }\right) \right]
\left( \widehat{\beta }_{m}-\beta \right) \\
=&II_{a}+II_{b}+II_{c}+II_{d}.
\end{align*}%
It holds that%
\begin{align*}
&E\left( \sum_{j=1}^{H}\left( 1-\frac{j}{H+1}\right) \left( \frac{1}{m}%
\sum_{t=j+1}^{m}\left( \epsilon _{t}\epsilon _{t-j}-E\left( \epsilon
_{0}\epsilon _{j}\right) \right) \right) \right) ^{2} \\
\leq &c_{0}H\sum_{j=1}^{H}\left( 1-\frac{j}{H+1}\right) ^{2}E\left( \frac{1}{%
m}\sum_{t=j+1}^{m}\left( \epsilon _{t}\epsilon _{t-j}-E\left( \epsilon
_{0}\epsilon _{j}\right) \right) \right) ^{2} \\
\leq &c_{0}H\sum_{j=1}^{H}E\left( \frac{1}{m}\sum_{t=j+1}^{m}\left( \epsilon
_{t}\epsilon _{t-j}-E\left( \epsilon _{0}\epsilon _{j}\right) \right)
\right) ^{2}.
\end{align*}%
By Assumption \ref{b-shifts}\textit{(i)} and the same logic as in the proof
of Lemma \ref{max-ineq-1}, it follows that $\epsilon _{t}\epsilon _{t-j}$ is
an $L_{2}$-decomposable Bernoulli shift for all $j$; hence, using
Proposition 4 in \citet{berkes2011split}, it follows immediately that%
\begin{equation*}
H\sum_{j=1}^{H}E\left( \frac{1}{m}\sum_{t=j+1}^{m}\left( \epsilon
_{t}\epsilon _{t-j}-E\left( \epsilon _{0}\epsilon _{j}\right) \right)
\right) ^{2}\leq c_{0}\frac{H^{2}}{m},
\end{equation*}%
which entails that $II_{a}=O_{P}\left( Hm^{-1/2}\right) $. Standard
passages, using Assumption \ref{exogeneity}\textit{(iii)}, entail that%
\begin{equation}
\left\Vert \widehat{\beta }_{m}-\beta \right\Vert =O_{P}\left( \sqrt{\frac{d%
}{m}}\right) ;  \label{b-ols}
\end{equation}%
further%
\begin{align*}
&E\sum_{j=1}^{H}\left( 1-\frac{j}{H+1}\right) \left( \frac{1}{m}%
\sum_{t=j+1}^{m}\mathbf{x}_{t}\epsilon _{t-j}\right) \\
\leq &E\sum_{j=1}^{H}\left\Vert \frac{1}{m}\sum_{t=j+1}^{m}\mathbf{x}%
_{t}\epsilon _{t-j}\right\Vert \leq c_{0}\frac{1}{m}\sum_{j=1}^{H}%
\sum_{t=j+1}^{m}\left( E\left\Vert \mathbf{x}_{t}\right\Vert ^{2}\right)
^{1/2}\left( E\left\vert \epsilon _{t-j}\right\vert ^{2}\right) ^{1/2}\leq
c_{0}H;
\end{align*}%
this bound is not the sharpest possible (one could assume $E\left( \mathbf{x}%
_{t}\epsilon _{t-j}\right) =0$ and then use the Markov inequality), but it
suffices for our purposes. Using (\ref{b-ols}), this entails that both $%
II_{b}$ and $II_{c}$ are $O_{P}\left( d^{1/2}Hm^{-1/2}\right) $. Finally%
\begin{align*}
&\left\vert \left( \widehat{\beta }_{m}-\beta \right) ^{\prime }\left[
\sum_{j=1}^{H}\left( 1-\frac{j}{H+1}\right) \left( \frac{1}{m}%
\sum_{t=j+1}^{m}\mathbf{x}_{t}\mathbf{x}_{t-j}^{\prime }\right) \right]
\left( \widehat{\beta }_{m}-\beta \right) \right\vert \\
\leq &\left\Vert \widehat{\beta }_{m}-\beta \right\Vert ^{2}\frac{1}{m}%
\sum_{j=1}^{H}\sum_{t=j+1}^{m}\left\vert \mathbf{x}_{t}\mathbf{x}%
_{t-j}^{\prime }\right\vert ,
\end{align*}%
and since $\left\vert \mathbf{x}_{t}\mathbf{x}_{t-j}^{\prime }\right\vert
\leq \left\Vert \mathbf{x}_{t}\right\Vert \left\Vert \mathbf{x}%
_{t-j}\right\Vert $, using the Cauchy-Schwartz inequality and the
stationarity of $\mathbf{x}_{t}$%
\begin{equation*}
E\frac{1}{m}\sum_{j=1}^{H}\sum_{t=j+1}^{m}\left\vert \mathbf{x}_{t}\mathbf{x}%
_{t-j}^{\prime }\right\vert \leq \frac{1}{m}\sum_{j=1}^{H}\sum_{t=j+1}^{m}%
\left( E\left\Vert \mathbf{x}_{t}\right\Vert ^{2}\right) ^{1/2}\left(
E\left\Vert \mathbf{x}_{t-j}\right\Vert ^{2}\right) ^{1/2}\leq c_{0}H.
\end{equation*}%
Thus, again using (\ref{b-ols}), it follows that $II_{d}=O_{P}\left(
Hdm^{-1}\right) $. The final result now follows from putting all together.
\end{proof}

\begin{proof}[Proof of Theorem \protect\ref{power}]
The proof is related to that of Theorem 2.2 in \citet{horvath2004monitoring}%
, and we only report the main passages for the sake of a concise discussion.
Standard algebra entails that%
\begin{equation*}
\sum_{t=m+1}^{m+\widetilde{k}}\widehat{\epsilon }_{t}=\sum_{t=m+1}^{m+%
\widetilde{k}}\epsilon _{t}-\left( \widehat{\beta }_{m}-\beta _{0}\right)
^{\prime }\sum_{t=m+1}^{m+\widetilde{k}}\mathbf{x}_{t}+\Delta _{m}^{\prime
}\sum_{t=m+1}^{m+\widetilde{k}}\mathbf{x}_{t}I\left( t>m+k^{\ast }\right) .
\end{equation*}%
The proof of Theorem \ref{nul-distr} immediately implies that%
\begin{equation*}
r_{m}^{\eta -1/2}\frac{\left\vert \displaystyle\sum_{t=m+1}^{m+\widetilde{k}%
}\epsilon _{t}-\left( \widehat{\beta }_{m}-\beta _{0}\right) ^{\prime }%
\displaystyle\sum_{t=m+1}^{m+\widetilde{k}}\mathbf{x}_{t}\right\vert }{%
g_{\eta }\left( m;\widetilde{k}\right) }=O_{P}\left( 1\right) .
\end{equation*}%
Also, using Lemma \ref{max-ineq-2}, it is not hard to see that%
\begin{align*}
& \Delta _{m}^{\prime }\sum_{t=m+1}^{m+\widetilde{k}}\mathbf{x}_{t}I\left(
t>k^{\ast }\right) \\
=& \left( \widetilde{k}-k^{\ast }\right) \mathbf{c}_{1}^{\prime }\Delta _{m}+%
\mathbf{c}_{1}^{\prime }\Delta _{m}O_{P}\left( \left( \widetilde{k}-k^{\ast
}\right) ^{1/2}\right) =\mathbf{c}_{1}^{\prime }\Delta _{m}\left(
a_{m}+O_{P}\left( a_{m}^{1/2}\right) \right) .
\end{align*}%
Recall that, by assumption, it holds that $\left\vert \mathbf{c}_{1}^{\prime
}\Delta _{m}\right\vert >0$. Under condition \textit{(i)}, let $\widetilde{k}%
=k^{\ast }+a_{m}$. It follows that%
\begin{align*}
& r_{m}^{\eta -1/2}\frac{\left\vert \Delta _{m}^{\prime }\displaystyle%
\sum_{t=m+k^{\ast }+1}^{m+\widetilde{k}}\mathbf{x}_{t}\right\vert }{%
m^{1/2}\left( 1+\displaystyle\frac{\widetilde{k}}{m}\right) \displaystyle%
\left( \frac{\widetilde{k}}{m+\widetilde{k}}\right) ^{\eta }} \\
=& r_{m}^{\eta -1/2}\frac{a_{m}}{m^{1/2}\left( 1+\displaystyle\frac{%
\widetilde{k}}{m}\right) \displaystyle\left( \frac{\widetilde{k}}{m+%
\widetilde{k}}\right) ^{\eta }}\mathbf{c}_{1}^{\prime }\Delta
_{m}+O_{P}\left( \mathbf{c}_{1}^{\prime }\Delta _{m}\right) \\
=& a_{m}^{1/2}\mathbf{c}_{1}^{\prime }\Delta _{m}+O_{P}\left( \mathbf{c}%
_{1}^{\prime }\Delta _{m}\right) +o\left( a_{m}^{1/2}\mathbf{c}_{1}^{\prime
}\Delta _{m}\right) ,
\end{align*}%
whence the theorem follows immediately if (\ref{regime-1}) holds. Under
condition \textit{(ii)}, let $\widetilde{k}=\left\lfloor \left( 1+c\right)
k^{\ast }\right\rfloor \leq T_{m}$ for some $c>0$; we have, using the mean
value theorem%
\begin{align*}
& r_{m}^{\eta -1/2}\frac{\left\vert \Delta _{m}^{\prime }\displaystyle%
\sum_{t=m+k^{\ast }+1}^{m+\widetilde{k}}\mathbf{x}_{t}\right\vert }{%
m^{1/2}\left( 1+\displaystyle\frac{\widetilde{k}}{m}\right) \displaystyle%
\left( \frac{\widetilde{k}}{m+\widetilde{k}}\right) ^{\eta }} \\
=& c_{0}r_{m}^{\eta -1/2}\frac{k^{\ast }+O_{P}\left( \left( k^{\ast }\right)
^{1/2}\right) }{m^{1/2}\left( 1+\displaystyle\frac{k^{\ast }}{m}\right) %
\displaystyle\left( \frac{k^{\ast }}{m+k^{\ast }}\right) ^{\eta }}\mathbf{c}%
_{1}^{\prime }\Delta _{m} \\
=& c_{0}r_{m}^{\eta -1/2}m^{1/2}\left( \frac{k^{\ast }}{m+k^{\ast }}\right)
^{1-\eta }\mathbf{c}_{1}^{\prime }\Delta _{m},
\end{align*}%
which again immediately yields the desired result as long as (\ref{regime-2}%
) holds.
\end{proof}

\begin{proof}[Proof of Theorem \protect\ref{delay}]
We will study the limiting behaviour of%
\begin{equation*}
r_{m}^{\eta -1/2}\max_{a_{m}\leq k\leq a_{m}}\frac{\left\vert Q\left(
m;k\right) \right\vert }{g_{\eta }\left( m;k\right) }=r_{m}^{\eta -1/2}\frac{%
\left\vert \displaystyle\sum_{t=m+1}^{m+a_{m}}\widehat{\epsilon }%
_{t}\right\vert }{m^{1/2}\left( 1+\displaystyle\frac{a_{m}}{m}\right) %
\displaystyle\left( \frac{a_{m}}{m+a_{m}}\right) ^{\eta }},
\end{equation*}%
under the alternative hypothesis of (\ref{alt-changepoint}). By standard
algebra, it holds that%
\begin{align*}
&r_{m}^{\eta -1/2}\frac{\left\vert \displaystyle\sum_{t=m+1}^{m+a_{m}}%
\widehat{\epsilon }_{t}\right\vert }{m^{1/2}\left( 1+\displaystyle\frac{a_{m}%
}{m}\right) \displaystyle\left( \frac{a_{m}}{m+a_{m}}\right) ^{\eta }} \\
=&r_{m}^{\eta -1/2}\frac{\left\vert \displaystyle\sum_{t=m+1}^{m+a_{m}}%
\epsilon _{t}-\displaystyle\sum_{t=m+1}^{m+a_{m}}\mathbf{x}_{t}^{\prime
}\left( \displaystyle\sum_{t=1}^{m}\mathbf{x}_{t}\mathbf{x}_{t}^{\prime
}\right) ^{-1}\left( \displaystyle\sum_{t=1}^{m}\mathbf{x}_{t}\epsilon
_{t}\right) +\displaystyle\sum_{t=m+k^{\ast }+1}^{m+a_{m}}\mathbf{x}%
_{t}^{\prime }\Delta _{m}\right\vert }{m^{1/2}\left( 1+\displaystyle\frac{%
a_{m}}{m}\right) \displaystyle\left( \frac{a_{m}}{m+a_{m}}\right) ^{\eta }}.
\end{align*}%
We begin by noting that, following the passages above, it is not hard to see
that%
\begin{equation}
\left\Vert \sum_{t=m+1}^{m+a_{m}}\mathbf{x}_{t}\right\Vert =O_{P}\left(
d^{1/2}a_{m}\right) .  \label{lln}
\end{equation}%
Using Lemmas \ref{max-ineq-1}, (\ref{frob-inv}) and (\ref{lln}), it follows
that%
\begin{align*}
&r_{m}^{\eta -1/2}\frac{\left\vert \displaystyle\sum_{t=m+1}^{m+a_{m}}%
\mathbf{x}_{t}^{\prime }\left( \displaystyle\left( \frac{1}{m}\sum_{t=1}^{m}%
\mathbf{x}_{t}\mathbf{x}_{t}^{\prime }\right) ^{-1}-\mathbf{C}^{-1}\right) %
\displaystyle\left( \frac{1}{m}\sum_{t=1}^{m}\mathbf{x}_{t}\epsilon
_{t}\right) \right\vert }{m^{1/2}\left( 1+\displaystyle\frac{a_{m}}{m}%
\right) \displaystyle\left( \frac{a_{m}}{m+a_{m}}\right) ^{\eta }} \\
\leq &r_{m}^{\eta -1/2}\frac{\left\Vert \displaystyle\sum_{t=m+1}^{m+a_{m}}%
\mathbf{x}_{t}^{\prime }\right\Vert \left\Vert \displaystyle\left( \frac{1}{m%
}\sum_{t=1}^{m}\mathbf{x}_{t}\mathbf{x}_{t}^{\prime }\right) ^{-1}-\mathbf{C}%
^{-1}\right\Vert _{F}\displaystyle\left\Vert \frac{1}{m}\sum_{t=1}^{m}%
\mathbf{x}_{t}\epsilon _{t}\right\Vert }{m^{1/2}\left( 1+\displaystyle\frac{%
a_{m}}{m}\right) \displaystyle\left( \frac{a_{m}}{m+a_{m}}\right) ^{\eta }}
\\
=&O_{P}\left( 1\right) r_{m}^{\eta -1/2}m^{-1/2}\left( \frac{a_{m}}{m}%
\right) ^{-\eta }d^{1/2}a_{m}\frac{d}{m^{1/2}}\left( \frac{d}{m}\right)
^{1/2}=O_{P}\left( d^{2}\frac{a_{m}^{1/2}}{m}\right) =o_{P}\left( 1\right) ,
\end{align*}%
by having $d=O\left( m^{1/4}\right) $. Therefore%
\begin{align*}
&r_{m}^{\eta -1/2}\frac{\left\vert \displaystyle\sum_{t=m+1}^{m+a_{m}}%
\epsilon _{t}-\displaystyle\sum_{t=m+1}^{m+a_{m}}\mathbf{x}_{t}^{\prime
}\left( \displaystyle\sum_{t=1}^{m}\mathbf{x}_{t}\mathbf{x}_{t}^{\prime
}\right) ^{-1}\left( \displaystyle\sum_{t=1}^{m}\mathbf{x}_{t}\epsilon
_{t}\right) +\displaystyle\sum_{t=m+k^{\ast }+1}^{m+a_{m}}\mathbf{x}%
_{t}^{\prime }\Delta _{m}\right\vert }{m^{1/2}\left( 1+\displaystyle\frac{%
a_{m}}{m}\right) \displaystyle\left( \frac{a_{m}}{m+a_{m}}\right) ^{\eta }}
\\
=&r_{m}^{\eta -1/2}\frac{\left\vert \displaystyle\sum_{t=m+1}^{m+a_{m}}%
\epsilon _{t}-\displaystyle\sum_{t=m+1}^{m+a_{m}}\mathbf{x}_{t}^{\prime }%
\mathbf{C}^{-1}\displaystyle\left( \frac{1}{m}\sum_{t=1}^{m}\mathbf{x}%
_{t}\epsilon _{t}\right) +\displaystyle\sum_{t=m+k^{\ast }+1}^{m+a_{m}}%
\mathbf{x}_{t}^{\prime }\Delta _{m}\right\vert }{m^{1/2}\left( 1+%
\displaystyle\frac{a_{m}}{m}\right) \displaystyle\left( \frac{a_{m}}{m+a_{m}}%
\right) ^{\eta }}+o_{P}\left( 1\right) .
\end{align*}%
Similarly, we note that%
\begin{equation*}
E\left\Vert \sum_{t=m+1}^{m+a_{m}}\left( \mathbf{x}_{t}-\mathbf{c}%
_{1}\right) \right\Vert ^{2}=\sum_{j=2}^{d}E\left(
\sum_{t=m+1}^{m+a_{m}}\left( x_{j,t}-E\left( x_{j,0}\right) \right) \right)
^{2}\leq c_{0}da_{m},
\end{equation*}%
by using the fact that $x_{j,t}-E\left( x_{j,0}\right) $ is a weakly
dependent process in the sense of Definition \ref{bernoulli}, and
Proposition 4 in \citet{berkes2011split}, whence $\left\Vert
\sum_{t=m+1}^{m+a_{m}}\left( \mathbf{x}_{t}-\mathbf{c}_{1}\right)
\right\Vert =O_{P}\left( d^{1/2}a_{m}^{1/2}\right) $. Therefore, it is easy
to see that%
\begin{align*}
&r_{m}^{\eta -1/2}\frac{\left\vert \displaystyle\sum_{t=m+1}^{m+a_{m}}\left( 
\mathbf{x}_{t}-\mathbf{c}_{1}\right) ^{\prime }\mathbf{C}^{-1}\displaystyle%
\left( \frac{1}{m}\sum_{t=1}^{m}\mathbf{x}_{t}\epsilon _{t}\right)
\right\vert }{m^{1/2}\left( 1+\displaystyle\frac{a_{m}}{m}\right) %
\displaystyle\left( \frac{a_{m}}{m+a_{m}}\right) ^{\eta }} \\
\leq &r_{m}^{\eta -1/2}\frac{\left\Vert \displaystyle\sum_{t=m+1}^{m+a_{m}}%
\left( \mathbf{x}_{t}-\mathbf{c}_{1}\right) \right\Vert \left\Vert \mathbf{C}%
^{-1}\right\Vert _{F}\displaystyle\left\Vert \frac{1}{m}\sum_{t=1}^{m}%
\mathbf{x}_{t}\epsilon _{t}\right\Vert }{m^{1/2}\left( 1+\displaystyle\frac{%
a_{m}}{m}\right) \displaystyle\left( \frac{a_{m}}{m+a_{m}}\right) ^{\eta }}
\\
=&O_{P}\left( r_{m}^{\eta -1/2}m^{-1/2}\left( \frac{a_{m}}{m}\right) ^{-\eta
}d^{1/2}a_{m}^{1/2}d^{1/2}\left( \frac{d}{m}\right) ^{1/2}\right)
=O_{P}\left( \frac{d^{3/2}}{m^{1/2}}\right) =o_{P}\left( 1\right) ,
\end{align*}%
having used\ Lemmas \ref{max-ineq-1} and \ref{max-ineq-2}, (\ref{c-frob})
and $d=O\left( m^{1/4}\right) $. This entails that%
\begin{align*}
&r_{m}^{\eta -1/2}\frac{\left\vert \displaystyle\sum_{t=m+1}^{m+a_{m}}%
\epsilon _{t}-\displaystyle\sum_{t=m+1}^{m+a_{m}}\mathbf{x}_{t}^{\prime }%
\mathbf{C}^{-1}\displaystyle\left( \frac{1}{m}\sum_{t=1}^{m}\mathbf{x}%
_{t}\epsilon _{t}\right) +\displaystyle\sum_{t=m+k^{\ast }+1}^{m+a_{m}}%
\mathbf{x}_{t}^{\prime }\Delta _{m}\right\vert }{m^{1/2}\left( 1+%
\displaystyle\frac{a_{m}}{m}\right) \displaystyle\left( \frac{a_{m}}{m+a_{m}}%
\right) ^{\eta }} \\
=&r_{m}^{\eta -1/2}\frac{\left\vert \displaystyle\sum_{t=m+1}^{m+a_{m}}%
\epsilon _{t}-a_{m}\displaystyle\left( \frac{1}{m}\sum_{t=1}^{m}\mathbf{c}%
_{1}^{\prime }\mathbf{C}^{-1}\mathbf{x}_{t}\epsilon _{t}\right) +%
\displaystyle\sum_{t=m+k^{\ast }+1}^{m+a_{m}}\mathbf{x}_{t}^{\prime }\Delta
_{m}\right\vert }{m^{1/2}\left( 1+\displaystyle\frac{a_{m}}{m}\right) %
\displaystyle\left( \frac{a_{m}}{m+a_{m}}\right) ^{\eta }}+o_{P}\left(
1\right) \\
=&r_{m}^{\eta -1/2}\frac{\left\vert \displaystyle\sum_{t=m+1}^{m+a_{m}}%
\epsilon _{t}-\displaystyle\frac{a_{m}}{m}\left( \displaystyle%
\sum_{t=1}^{m}\epsilon _{t}\right) +\displaystyle\sum_{t=m+k^{\ast
}+1}^{m+a_{m}}\mathbf{x}_{t}^{\prime }\Delta _{m}\right\vert }{m^{1/2}\left(
1+\displaystyle\frac{a_{m}}{m}\right) \displaystyle\left( \frac{a_{m}}{%
m+a_{m}}\right) ^{\eta }}+o_{P}\left( 1\right) ,
\end{align*}%
on account of the fact that $\mathbf{c}_{1}^{\prime }\mathbf{C}^{-1}\mathbf{x%
}_{t}=1$ for all $t$. Also note that%
\begin{equation*}
r_{m}^{\eta -1/2}\frac{a_{m}}{m}\frac{\left\vert \displaystyle\left( \frac{1%
}{m^{1/2}}\sum_{t=1}^{m}\epsilon _{t}\right) \right\vert }{\left( 1+%
\displaystyle\frac{a_{m}}{m}\right) \displaystyle\left( \frac{a_{m}}{m+a_{m}}%
\right) ^{\eta }}=O_{P}\left( 1\right) r_{m}^{\eta -1/2}\left( \frac{a_{m}}{m%
}\right) ^{1-\eta }=O_{P}\left( \left( \frac{a_{m}}{m}\right) ^{1/2}\right)
=o_{P}\left( 1\right) ,
\end{equation*}%
whence%
\begin{align*}
&r_{m}^{\eta -1/2}\frac{\left\vert \displaystyle\sum_{t=m+1}^{m+a_{m}}%
\epsilon _{t}-\displaystyle\frac{a_{m}}{m}\left( \displaystyle%
\sum_{t=1}^{m}\epsilon _{t}\right) +\displaystyle\sum_{t=m+k^{\ast
}+1}^{m+a_{m}}\mathbf{x}_{t}^{\prime }\Delta _{m}\right\vert }{m^{1/2}\left(
1+\displaystyle\frac{a_{m}}{m}\right) \displaystyle\left( \frac{a_{m}}{%
m+a_{m}}\right) ^{\eta }} \\
=&r_{m}^{\eta -1/2}\frac{\left\vert \displaystyle\sum_{t=m+1}^{m+a_{m}}%
\epsilon _{t}+\displaystyle\sum_{t=m+k^{\ast }+1}^{m+a_{m}}\mathbf{x}%
_{t}^{\prime }\Delta _{m}\right\vert }{m^{1/2}\left( 1+\displaystyle\frac{%
a_{m}}{m}\right) \displaystyle\left( \frac{a_{m}}{m+a_{m}}\right) ^{\eta }}%
+o_{P}\left( 1\right) .
\end{align*}%
Finally, using Lemma \ref{sip}, it holds that%
\begin{align*}
&r_{m}^{\eta -1/2}\frac{\left\vert \displaystyle\sum_{t=m+1}^{m+a_{m}}%
\epsilon _{t}-\sigma W_{1,m}\left( a_{m}\right) \right\vert }{m^{1/2}\left(
1+\displaystyle\frac{a_{m}}{m}\right) \displaystyle\left( \frac{a_{m}}{%
m+a_{m}}\right) ^{\eta }} \\
=&O_{P}\left( 1\right) r_{m}^{\eta -1/2}m^{\zeta _{1}-1/2}\left( \frac{a_{m}%
}{m}\right) ^{-\eta }=o_{P}\left( 1\right) ,
\end{align*}%
so that%
\begin{align*}
&r_{m}^{\eta -1/2}\frac{\left\vert \displaystyle\sum_{t=m+1}^{m+a_{m}}%
\epsilon _{t}+\displaystyle\sum_{t=m+k^{\ast }+1}^{m+a_{m}}\mathbf{x}%
_{t}^{\prime }\Delta _{m}\right\vert }{m^{1/2}\left( 1+\displaystyle\frac{%
a_{m}}{m}\right) \displaystyle\left( \frac{a_{m}}{m+a_{m}}\right) ^{\eta }}
\\
=&r_{m}^{\eta -1/2}\frac{\left\vert \sigma W_{1,m}\left( a_{m}\right) +%
\displaystyle\sum_{t=m+k^{\ast }+1}^{m+a_{m}}\mathbf{x}_{t}^{\prime }\Delta
_{m}\right\vert }{m^{1/2}\left( 1+\displaystyle\frac{a_{m}}{m}\right) %
\displaystyle\left( \frac{a_{m}}{m+a_{m}}\right) ^{\eta }}+o_{P}\left(
1\right) \\
=&a_{m}^{-1/2}\left\vert \sigma W_{1,m}\left( a_{m}\right) +\sum_{t=k^{\ast
}+1}^{a_{m}}\mathbf{x}_{t}^{\prime }\Delta _{m}\right\vert +o_{P}\left(
1\right) ,
\end{align*}%
where the last result follows from having $a_{m}=o\left( m\right) $. Given
that the distribution of $W_{1,m}\left( \cdot \right) $ does not depend on $%
m $, it follows that%
\begin{equation*}
a_{m}^{-1/2}\left\vert \sigma W_{1,m}\left( a_{m}\right) +\sum_{t=m+k^{\ast
}+1}^{m+a_{m}}\mathbf{x}_{t}^{\prime }\Delta _{m}\right\vert \overset{%
\mathcal{D}}{=}\left\vert \sigma W\left( 1\right)
+a_{m}^{-1/2}\sum_{t=m+k^{\ast }+1}^{m+a_{m}}\mathbf{x}_{t}^{\prime }\Delta
_{m}\right\vert ,
\end{equation*}%
where $W\left( \cdot \right) $ is a standard Wiener. The results above
entail that, as $m\rightarrow \infty $%
\begin{equation*}
P\left( r_{m}^{\eta -1/2}\max_{a_{m}\leq k\leq a_{m}}\frac{\left\vert
Q\left( m;a_{m}\right) \right\vert }{g\left( m;a_{m}\right) }\leq x\right)
=P\left( \left\vert \sigma Z+a_{m}^{-1/2}\sum_{t=m+k^{\ast }+1}^{m+a_{m}}%
\mathbf{x}_{t}^{\prime }\Delta _{m}\right\vert \leq x\right) +o_{P}\left(
1\right) ,
\end{equation*}%
for all $-\infty <x<\infty $, with $Z\sim N\left( 0,1\right) $.

We are now ready to prove the main statement of the theorem. By assumption, $%
\left\vert \mathbf{c}_{1}^{\prime }\Delta _{m}\right\vert =c_{0}\left\Vert 
\mathbf{c}_{1}\right\Vert \left\Vert \Delta _{m}\right\Vert $, where $%
c_{0}>0 $ is the cosine of the angle between $\mathbf{c}_{1}$ and $\Delta
_{m}$. Hence, Assumption \ref{delay-assumption}\textit{(ii)} entails%
\begin{equation*}
\left\vert \mathbf{c}_{1}^{\prime }\Delta _{m}\right\vert =c_{0}\left\Vert 
\mathbf{c}_{1}\right\Vert \left\Vert \Delta _{m}\right\Vert =\Omega \left(
d^{1/2}\Delta _{m}\right) ;
\end{equation*}%
henceforth, we can and will assume, without loss of generality, that $%
\mathbf{c}_{1}^{\prime }\Delta _{m}>0$. We note that, using the fact that,
by Assumption \ref{delay-assumption}\textit{(i)}, $a_{m}-k^{\ast }=\Omega
\left( a_{m}\right) $, we have%
\begin{align}
&a_{m}^{-1/2}\sum_{t=m+k^{\ast }+1}^{m+a_{m}}\mathbf{x}_{t}^{\prime }\Delta
_{m}  \label{drift} \\
=&a_{m}^{-1/2}\sum_{t=m+k^{\ast }+1}^{m+a_{m}}\mathbf{c}_{1}^{\prime }\Delta
_{m}+a_{m}^{-1/2}\sum_{t=m+k^{\ast }+1}^{m+a_{m}}\left( \mathbf{x}_{t}-%
\mathbf{c}_{1}\right) ^{\prime }\Delta _{m}  \notag \\
=&a_{m}^{1/2}\Omega \left( d^{1/2}\Delta _{m}\right) +O_{P}\left(
d^{1/2}\Delta _{m}\right) ,  \notag
\end{align}%
where the second term follows from the CLT, viz.%
\begin{equation*}
\left\vert a_{m}^{-1/2}\sum_{t=m+k^{\ast }+1}^{m+a_{m}}\left( \mathbf{x}_{t}-%
\mathbf{c}_{1}\right) ^{\prime }\Delta _{m}\right\vert \leq \left\Vert
a_{m}^{-1/2}\sum_{t=m+k^{\ast }+1}^{m+a_{m}}\left( \mathbf{x}_{t}-\mathbf{c}%
_{1}\right) \right\Vert \left\Vert \Delta _{m}\right\Vert =O_{P}\left(
d^{1/2}\right) \Omega \left( \Delta _{m}\right) .
\end{equation*}%
Equation (\ref{drift}), and the fact that we are assuming $\mathbf{c}%
_{1}^{\prime }\Delta _{m}>0$, imply that%
\begin{equation*}
\lim_{m\rightarrow \infty }P\left( \sigma Z+a_{m}^{-1/2}\sum_{t=m+k^{\ast
}+1}^{m+a_{m}}\mathbf{x}_{t}^{\prime }\Delta _{m}>0\right) =1,
\end{equation*}%
and%
\begin{equation}
a_{m}^{-1/2}\sum_{t=m+k^{\ast }+1}^{m+a_{m}}\mathbf{x}_{t}^{\prime }\Delta
_{m}\overset{\mathcal{P}}{\rightarrow }\infty ,  \label{drift2}
\end{equation}%
as $m\rightarrow \infty $. Therefore, as $m\rightarrow \infty $%
\begin{equation*}
P\left( r_{m}^{\eta -1/2}\frac{\left\vert Q\left( m;a_{m}\right) \right\vert 
}{g_{\eta }\left( m;a_{m}\right) }\leq x\right) =P\left( \sigma
Z+a_{m}^{-1/2}\sum_{t=m+k^{\ast }+1}^{m+a_{m}}\mathbf{x}_{t}^{\prime }\Delta
_{m}\leq x\right) +o_{P}\left( 1\right) .
\end{equation*}%
Given that, by (\ref{drift2}) 
\begin{equation*}
\lim_{m\rightarrow \infty }P\left( \sigma Z+a_{m}^{-1/2}\sum_{t=m+k^{\ast
}+1}^{m+a_{m}}\mathbf{x}_{t}^{\prime }\Delta _{m}\leq c_{\alpha ,\eta
}\right) =\lim_{m\rightarrow \infty }P\left( Z\leq \frac{c_{\alpha ,\eta }}{%
\sigma }-\frac{1}{\sigma }a_{m}^{-1/2}\sum_{t=m+k^{\ast }+1}^{m+a_{m}}%
\mathbf{x}_{t}^{\prime }\Delta _{m}\right) =0,
\end{equation*}%
putting all the above together, we finally have%
\begin{equation*}
\lim_{m\rightarrow \infty }P\left( r_{m}^{\eta -1/2}\frac{\left\vert Q\left(
m;a_{m}\right) \right\vert }{g_{\eta }\left( m;a_{m}\right) }\leq c_{\alpha
,\eta }\right) =0.
\end{equation*}%
We now conclude, by noting that the event $\left\{ \tau _{m}>a_{m}\right\} $
is equivalent to having 
\begin{equation*}
\left\{ r_{m}^{\eta -1/2}\left\vert Q\left( m;a_{m}\right) \right\vert \leq
c_{\alpha ,\eta }g_{\eta }\left( m;a_{m}\right) \right\} .
\end{equation*}%
We have shown that $\lim_{m\rightarrow \infty }P\left( \tau
_{m}>a_{m}\right) =0$, whence $\lim_{m\rightarrow \infty }P\left( \tau
_{m}=a_{m}\right) =1$.
\end{proof}

\begin{proof}[Proof of Theorem \protect\ref{dynamic-reg}]
We prove that Lemmas \ref{max-ineq-1}-\ref{max-inequ-3} hold under (\ref%
{dyn-reg}). Thereafter, the proofs of Theorems \ref{nul-distr}, \ref{power}
and \ref{delay} and Lemma \ref{lrv} can be repeated \textit{verbatim}.
Further, for simplicity, we report the proof for the case $p=1$, writing (%
\ref{dyn-reg}) as%
\begin{equation}
y_{t}=\rho y_{t-1}+\mathbf{z}_{t}^{\prime }\beta ^{Z}+\epsilon _{t}.
\label{dyn-2}
\end{equation}%
Equation (\ref{dyn-2}) can be expressed in recursive form%
\begin{equation}
y_{t}=\rho ^{t}y_{0}+\sum_{j=0}^{t-1}\rho ^{j}\left( \mathbf{z}%
_{t-j}^{\prime }\beta ^{Z}+\epsilon _{t-j}\right) ,  \label{recursive}
\end{equation}%
where $y_{0}$ is an initial condition. Under Assumption \ref{ass-dyn}\textit{%
(ii)}, (\ref{dyn-2}) admits a unique stationary, non-anticipative solution
given by%
\begin{equation}
\overline{y}_{t}=\sum_{j=0}^{\infty }\rho ^{j}\left( \mathbf{z}%
_{t-j}^{\prime }\beta ^{Z}+\epsilon _{t-j}\right) =\sum_{j=0}^{\infty }\rho
^{j}u_{t-j},  \label{stationary}
\end{equation}%
where $u_{t}=\mathbf{z}_{t}^{\prime }\beta ^{Z}+\epsilon _{t}$, for any
(stochastically bounded) initial value $y_{0}$. By Assumption \ref{b-shifts}%
, $x_{j,t}=h_{j}^{x}\left( \eta _{j,t}^{x},...\right) $ for $2\leq j\leq d-p$
and $\epsilon _{t}=h^{\epsilon }\left( \eta _{t}^{\epsilon },...\right) $;
hence, it immediately follows that we can use the representation $%
u_{t}=h^{u}\left( \eta _{t}^{u},...\right) $ and that $u_{t}$ satisfies
Assumption \ref{b-shifts}. Define the $\sigma $-fields $\mathcal{F}_{t-\ell
}^{t-j}=\left\{ \eta _{i}^{u}\right\} _{i=t-\ell }^{t-j}$, with the
convention that $\mathcal{F}_{t-\ell }^{t-j}$ is empty if $j>\ell $, and $%
\mathcal{F}_{t-\min \left\{ \ell ,j\right\} }^{-\infty }=\left\{ \widetilde{%
\eta }_{t-\min \left\{ \ell ,j\right\} }^{u},...,\widetilde{\eta }_{-\infty
}^{u}\right\} $, where $\widetilde{\eta }_{t}^{u}$ is an independent copy of 
$\eta _{t}^{u}$, forming an \textit{i.i.d.} sequence. We now define the
coupling constructions 
\begin{equation*}
\widetilde{u}_{t-j,\ell }=h^{u}\left( \mathcal{F}_{t-\ell }^{t-j},\mathcal{F}%
_{t-\min \left\{ \ell ,j\right\} }^{-\infty }\right) ,
\end{equation*}%
and 
\begin{equation*}
\widetilde{y}_{t,\ell }=\sum_{j=0}^{\infty }\rho ^{j}\widetilde{u}_{t-j,\ell
}.
\end{equation*}%
Using Minkowski's inequality, we have that%
\begin{equation*}
\left\vert y_{t}-\widetilde{y}_{t,\ell }\right\vert _{4}\leq \left\vert
y_{t}-\overline{y}_{t}\right\vert _{4}+\left\vert \overline{y}_{t}-%
\widetilde{y}_{t,\ell }\right\vert _{4}=I+II,
\end{equation*}%
with 
\begin{equation}
\left\vert y_{t}-\overline{y}_{t}\right\vert _{4}\leq \left\vert \rho
^{t}y_{0}\right\vert _{4}+\left\vert \sum_{j=t}^{\infty }\rho
^{j}u_{t-j}\right\vert _{4}\leq c_{0}\left\vert \rho \right\vert
^{t}+\sum_{j=t}^{\infty }\left\vert \rho \right\vert ^{j}\left\vert
u_{0}\right\vert _{4}\leq c_{1}\left\vert \rho \right\vert ^{t},
\label{bs-1}
\end{equation}%
and%
\begin{align}
& \left\vert \overline{y}_{t}-\widetilde{y}_{t,\ell }\right\vert _{4}
\label{bs-2} \\
\leq & \sum_{j=0}^{\infty }\left\vert \rho \right\vert ^{j}\left\vert
u_{t-j}-\widetilde{u}_{t-j,\ell }\right\vert _{4}=\sum_{j=0}^{\ell
}\left\vert \rho \right\vert ^{j}\left\vert u_{t-j}-\widetilde{u}_{t-j,\ell
}\right\vert _{4}+\sum_{j=\ell +1}^{\infty }\left\vert \rho \right\vert
^{j}\left\vert u_{t-j}-\widetilde{u}_{t-j,\ell }\right\vert _{4}  \notag \\
\leq & \sum_{j=0}^{\ell }\left\vert \rho \right\vert ^{j}\left\vert u_{t-j}-%
\widetilde{u}_{t-j,\ell }\right\vert _{4}+\left\vert \rho \right\vert ^{\ell
}\sum_{j=1}^{\infty }\left\vert \rho \right\vert ^{j}\left( \left\vert
u_{t-j}\right\vert _{4}+\left\vert \widetilde{u}_{t-j,\ell }\right\vert
_{4}\right)  \notag \\
\leq & c_{0}\left( \sum_{j=0}^{\ell }\left\vert \rho \right\vert ^{j}\right)
\ell ^{-a}+c_{1}\left\vert \rho \right\vert ^{\ell }\leq c_{2}\ell ^{-a}, 
\notag
\end{align}%
for some $c_{2}<\infty $. Note that equation (\ref{bs-2}) entails that $%
\overline{y}_{t}$ also satisfies Assumption \ref{b-shifts}.

We are now ready to prove that Lemmas \ref{max-ineq-1}-\ref{max-inequ-3}
hold under (\ref{dyn-reg}). Let $\overline{\mathbf{x}}_{t}=\left( \mathbf{z}%
_{t}^{\prime },\overline{\mathbf{y}}_{p,t}^{\prime }\right) ^{\prime }$. We
begin with Lemma \ref{max-inequ-3}, and note that%
\begin{align}
& \frac{1}{m}\sum_{t=1}^{m}\mathbf{x}_{t}\mathbf{x}_{t}^{\prime }-\frac{1}{m}%
\sum_{t=1}^{m}\overline{\mathbf{x}}_{t}\overline{\mathbf{x}}_{t}^{\prime }
\label{la4-1} \\
=& \frac{1}{m}\sum_{t=1}^{m}\overline{\mathbf{x}}_{t}\left( \mathbf{x}_{t}-%
\overline{\mathbf{x}}_{t}\right) ^{\prime }+\frac{1}{m}\sum_{t=1}^{m}\left( 
\mathbf{x}_{t}-\overline{\mathbf{x}}_{t}\right) \overline{\mathbf{x}}%
_{t}^{\prime }+\frac{1}{m}\sum_{t=1}^{m}\left( \mathbf{x}_{t}-\overline{%
\mathbf{x}}_{t}\right) \left( \mathbf{x}_{t}-\overline{\mathbf{x}}%
_{t}\right) ^{\prime }.  \notag
\end{align}%
It holds that%
\begin{align*}
\left\Vert \frac{1}{m}\sum_{t=1}^{m}\overline{\mathbf{x}}_{t}\left( \mathbf{x%
}_{t}-\overline{\mathbf{x}}_{t}\right) ^{\prime }\right\Vert _{F}^{2}=&
\sum_{j,h=1}^{d}\left( \frac{1}{m}\sum_{t=1}^{m}\overline{x}_{j,t}\left(
x_{h,t}-\overline{x}_{h,t}\right) \right) ^{2} \\
\leq & \sum_{j=1}^{d}\left( \frac{1}{m}\sum_{t=1}^{m}\overline{x}%
_{j,t}^{2}\right) \sum_{h=1}^{d}\left( \frac{1}{m}\sum_{t=1}^{m}\left(
x_{h,t}-\overline{x}_{h,t}\right) ^{2}\right) .
\end{align*}%
It is easy to see that%
\begin{equation*}
E\sum_{j=1}^{d}\left( \frac{1}{m}\sum_{t=1}^{m}\overline{x}_{j,t}^{2}\right)
\leq c_{0}d,
\end{equation*}%
and%
\begin{equation*}
E\sum_{h=1}^{d}\left( \frac{1}{m}\sum_{t=1}^{m}\left( x_{h,t}-\overline{x}%
_{h,t}\right) ^{2}\right) \leq c_{0}dm^{-1},
\end{equation*}%
because by (\ref{bs-1}), for all $h$, $\left\vert x_{h,t}-\overline{x}%
_{h,t}\right\vert _{2}\leq \left\vert x_{h,t}-\overline{x}_{h,t}\right\vert
_{4}\leq c_{1}\left\vert \rho \right\vert ^{t}$. Hence 
\begin{equation*}
\left\Vert \frac{1}{m}\sum_{t=1}^{m}\overline{\mathbf{x}}_{t}\left( \mathbf{x%
}_{t}-\overline{\mathbf{x}}_{t}\right) ^{\prime }\right\Vert
_{F}=O_{P}\left( \frac{d}{m^{1/2}}\right) ,
\end{equation*}%
and the same can be shown for the other terms in (\ref{la4-1}); we note that
these bounds are not the sharpest possible (in essence, this is due to the
fact that some coordinates of $\mathbf{x}_{t}$ and $\overline{\mathbf{x}}%
_{t} $ may be the same, thus reducing the dimensionality), but they suffice
for our purposes. Lemma \ref{max-inequ-3} now follows readily from repeating
the original proof, using $\sum_{t=1}^{m}\overline{\mathbf{x}}_{t}\overline{%
\mathbf{x}}_{t}^{\prime }$ instead of $\sum_{t=1}^{m}\mathbf{x}_{t}\mathbf{x}%
_{t}^{\prime }$. We now turn to Lemma \ref{max-ineq-1}. We estimate 
\begin{equation*}
\left\Vert \frac{1}{m}\sum_{t=1}^{m}\left( \mathbf{x}_{t}-\overline{\mathbf{x%
}}_{t}\right) \epsilon _{t}\right\Vert =\sum_{h=1}^{d}\left\vert
\sum_{t=1}^{m}\left( x_{h,t}-\overline{x}_{h,t}\right) \epsilon
_{t}\right\vert ^{2}\leq \sum_{h=1}^{d}\left\vert \sum_{t=1}^{m}\rho
^{t}y_{0}\epsilon _{t}\right\vert ^{2}+\sum_{h=1}^{d}\left\vert
\sum_{t=1}^{m}\left( \epsilon _{t}\sum_{j=t}^{\infty }\rho
^{j}u_{t-j}\right) \right\vert ^{2},
\end{equation*}%
having used (\ref{recursive}) and (\ref{stationary}). Considering the first
term, it holds that%
\begin{equation*}
E\sum_{h=1}^{d}\left\vert \sum_{t=1}^{m}\rho ^{t}y_{0}\epsilon
_{t}\right\vert ^{2}=\sum_{h=1}^{d}\sum_{t,s=1}^{m}\left\vert \rho
\right\vert ^{t+s}E\left( y_{0}^{2}\epsilon _{t}\epsilon _{s}\right) \leq
\sum_{h=1}^{d}\sum_{t,s=1}^{m}\left\vert \rho \right\vert ^{t+s}\left\vert
y_{0}\right\vert _{4}^{2}\left\vert \epsilon _{t}\right\vert _{4}\left\vert
\epsilon _{s}\right\vert _{4}\leq c_{0}d,
\end{equation*}%
having used the Cauchy-Schwartz inequality (twice), Assumption \ref{b-shifts}%
, and Assumption \ref{ass-dyn}. As far as the second term is concerned, note
that%
\begin{align*}
& E\sum_{h=1}^{d}\left\vert \sum_{t=1}^{m}\left( \epsilon
_{t}\sum_{j=t}^{\infty }\rho ^{j}u_{t-j}\right) \right\vert ^{2} \\
=& E\sum_{h=1}^{d}\left\vert \sum_{t=1}^{m}\left( \rho ^{t}\epsilon
_{t}\sum_{h=0}^{\infty }\rho ^{h}u_{h}\right) \right\vert
^{2}=E\sum_{h=1}^{d}\sum_{t,s=1}^{m}\left( \rho ^{t+s}\epsilon _{t}\epsilon
_{s}\sum_{h,i=0}^{\infty }\rho ^{h+i}u_{h}u_{i}\right) \\
=& \sum_{h=1}^{d}\sum_{t,s=1}^{m}\rho ^{t+s}\sum_{h,i=0}^{\infty }\rho
^{h+i}E\left( \epsilon _{t}\epsilon _{s}u_{h}u_{i}\right) \leq c_{0}d,
\end{align*}%
which follows again from Assumptions \ref{b-shifts} and \ref{ass-dyn}\textit{%
(ii)}, and H\"{o}lder's inequality. Hence it follows that%
\begin{equation*}
\left\Vert \frac{1}{m}\sum_{t=1}^{m}\left( \mathbf{x}_{t}-\overline{\mathbf{x%
}}_{t}\right) \epsilon _{t}\right\Vert =O_{P}\left( \frac{d^{1/2}}{m}\right)
,
\end{equation*}%
and now the desired result follows by repeating the proof of Lemma \ref%
{max-ineq-1} using $\overline{\mathbf{x}}_{t}$ instead of $\mathbf{x}_{t}$.
Finally, we consider Lemma \ref{max-ineq-2}, and we estimate%
\begin{equation*}
\max_{1\leq k\leq T_{m}}\frac{1}{k^{\zeta _{3}}}\left\Vert
\sum_{t=m+1}^{m+k}(\mathbf{x}_{t}-\overline{\mathbf{x}}_{t})\right\Vert \leq
c_{0}d^{1/2}\left( \max_{1\leq k\leq T_{m}}\frac{1}{k^{\zeta _{3}}}%
\left\vert \sum_{t=m+1}^{m+k}\rho ^{t}y_{0}\right\vert +\max_{1\leq k\leq
T_{m}}\frac{1}{k^{\zeta _{3}}}\left\vert
\sum_{t=m+1}^{m+k}\sum_{j=t}^{\infty }\rho ^{j}u_{t-j}\right\vert \right) ,
\end{equation*}%
where again we do not derive the sharpest bounds because we do not take into
account the fact that some elements of $\mathbf{x}_{t}$ and $\overline{%
\mathbf{x}}_{t}$ coincide. It is easy to see that%
\begin{equation*}
\max_{1\leq k\leq T_{m}}\frac{1}{k^{\zeta _{3}}}\left\vert
\sum_{t=m+1}^{m+k}\rho ^{t}y_{0}\right\vert =\left\vert y_{0}\right\vert
\max_{1\leq k\leq T_{m}}\frac{1}{k^{\zeta _{3}}}\left\vert
\sum_{t=m+1}^{m+k}\rho ^{t}\right\vert =O_{P}\left( 1\right) ,
\end{equation*}%
for all $\zeta _{3}>0$. Also%
\begin{align*}
& P\left( \max_{1\leq k\leq T_{m}}\frac{1}{k^{\zeta _{3}}}\left\vert
\sum_{t=m+1}^{m+k}\sum_{j=t}^{\infty }\rho ^{j}u_{t-j}\right\vert \geq
x\right) \\
\leq & P\left( \max_{0\leq \ell \leq \left\lceil \ln T_{m}\right\rceil
}\max_{\exp \left( \ell \right) \leq k\leq \exp \left( \ell +1\right) }\frac{%
1}{k^{\zeta _{3}}}\left\vert \sum_{t=m+1}^{m+k}\sum_{j=t}^{\infty }\rho
^{j}u_{t-j}\right\vert \geq x\right) \\
\leq & \sum_{\ell =0}^{\left\lceil \ln T_{m}\right\rceil }P\left( \max_{\exp
\left( \ell \right) \leq k\leq \exp \left( \ell +1\right) }\left\vert
\sum_{t=m+1}^{m+k}\sum_{j=t}^{\infty }\rho ^{j}u_{t-j}\right\vert \geq x\exp
\left( \zeta _{3}\ell \right) \right) \\
\leq & c_{0}x^{-p}\exp \left( -p\zeta _{3}\ell \right) E\max_{\exp \left(
\ell \right) \leq k\leq \exp \left( \ell +1\right) }\left\vert
\sum_{t=m+1}^{m+k}\sum_{j=t}^{\infty }\rho ^{j}u_{t-j}\right\vert ^{p} \\
\leq & x^{-p}\exp \left( -p\zeta _{3}\ell \right) E\max_{1\leq k\leq \exp
\left( \ell +1\right) }\left\vert \sum_{t=m+1}^{m+k}\sum_{j=t}^{\infty }\rho
^{j}u_{t-j}\right\vert ^{p}
\end{align*}%
for all $p\leq 4$. Note now that, by repeated application of Minkowski
inequality and by Assumptions \ref{b-shifts} and \ref{ass-dyn}\textit{(ii)}%
\begin{align*}
& \left\vert \sum_{t=m+1}^{m+k}\sum_{j=t}^{\infty }\rho
^{j}u_{t-j}\right\vert _{p} \\
\leq & \sum_{t=m+1}^{m+k}\left\vert \sum_{j=t}^{\infty }\rho
^{j}u_{t-j}\right\vert _{p}\leq \sum_{t=m+1}^{m+k}\sum_{j=t}^{\infty
}\left\vert \rho \right\vert ^{j}\left\vert u_{t-j}\right\vert _{p} \\
\leq & \sum_{t=m+1}^{m+k}\sum_{j=t}^{\infty }\left\vert \rho \right\vert
^{j}\left\vert u_{t-j}\right\vert _{p}\leq c_{0}\sum_{t=m+1}^{m+k}\left\vert
\rho \right\vert ^{t},
\end{align*}%
so that%
\begin{equation*}
\left\vert \sum_{t=m+1}^{m+k}\sum_{j=t}^{\infty }\rho ^{j}u_{t-j}\right\vert
_{p}^{p}\leq c_{0}\left( \sum_{t=m+1}^{m+k}\left\vert \rho \right\vert
^{t}\right) ^{p};
\end{equation*}%
this is true for all $m$ and $k$, and by construction $\left\vert \rho
\right\vert ^{t}\geq 0$. Hence, we can apply Theorem F in %
\citet{moricz1976moment}, whence%
\begin{equation*}
E\max_{1\leq k\leq \exp \left( \ell +1\right) }\left\vert
\sum_{t=m+1}^{m+k}\sum_{j=t}^{\infty }\rho ^{j}u_{t-j}\right\vert ^{p}\leq
\log _{2}\left( 4\exp \left( \ell +1\right) \right) \left(
\sum_{t=m+1}^{m+k}\left\vert \rho \right\vert ^{t}\right) ^{p}\leq c_{0}\ell
,
\end{equation*}%
which entails that%
\begin{equation*}
\max_{1\leq k\leq T_{m}}\frac{1}{k^{\zeta _{3}}}\left\vert
\sum_{t=m+1}^{m+k}\sum_{j=t}^{\infty }\rho ^{j}u_{t-j}\right\vert
=O_{P}\left( 1\right) ,
\end{equation*}%
for all $\zeta _{3}>0$. Putting all together, for all $\zeta _{3}>0$ it
holds that%
\begin{equation*}
\max_{1\leq k\leq T_{m}}\frac{1}{k^{\zeta _{3}}}\left\Vert
\sum_{t=m+1}^{m+k}(\mathbf{x}_{t}-\overline{\mathbf{x}}_{t})\right\Vert
=O_{P}\left( d^{1/2}\right) ;
\end{equation*}%
again, Lemma \ref{max-ineq-2} now follows from repeating the original proof
with $\overline{\mathbf{x}}_{t}$ instead of $\mathbf{x}_{t}$.
\end{proof}

\begin{proof}[Proof of Proposition \protect\ref{veto-prop}]
Let $\sigma =1$ for simplicity. We begin by showing that, on a suitably
enlarged probability space, there exist two independent standard Wiener
processes $\left\{ W_{1,m}\left( k\right) ,k\geq 1\right\} $ and $\left\{
W_{2,m}\left( k\right) ,k\geq 1\right\} $ such that%
\begin{equation}
\max_{1\leq k\leq T_{m}}\frac{\left\vert \displaystyle\sum_{t=m+1}^{m+k}%
\widehat{\epsilon }_{t}-\left( W_{1,m}\left( k\right) -\displaystyle\frac{k}{%
m}W_{2,m}\left( m\right) \right) \right\vert }{m^{1/2}\left( 1+\displaystyle%
\frac{k}{m}\right) \min_{1\leq j\leq J}c_{\alpha ,\eta _{j}}\widetilde{r}%
_{m}^{1/2-\eta _{j}}d_{a_{m},\eta }\displaystyle\left( \frac{k}{m+k}\right)
^{\eta _{j}}}=o_{P}\left( 1\right) .  \label{veto-1}
\end{equation}%
Some passages are repetitive, but we report them anyway to make the proof
easier to follow. We begin by noting that%
\begin{align}
&\max_{1\leq k\leq T_{m}}\frac{\left\vert \displaystyle\sum_{t=m+1}^{m+k}%
\mathbf{x}_{t}^{\prime }\left( \displaystyle\left( \frac{1}{m}\sum_{t=1}^{m}%
\mathbf{x}_{t}\mathbf{x}_{t}^{\prime }\right) ^{-1}-\mathbf{C}^{-1}\right)
\left( \displaystyle\sum_{t=1}^{m}\mathbf{x}_{t}\epsilon _{t}\right)
\right\vert }{m^{1/2}\left( 1+\displaystyle\frac{k}{m}\right) \min_{1\leq
j\leq J}c_{\alpha ,\eta _{j}}\widetilde{r}_{m}^{1/2-\eta _{j}}d_{a_{m},\eta }%
\displaystyle\left( \frac{k}{m+k}\right) ^{\eta _{j}}}  \label{veto-a} \\
\leq &\max_{1\leq k\leq T_{m}}\frac{\left\Vert \displaystyle%
\sum_{t=m+1}^{m+k}\mathbf{x}_{t}\right\Vert \left\Vert \left( \displaystyle%
\left( \frac{1}{m}\sum_{t=1}^{m}\mathbf{x}_{t}\mathbf{x}_{t}^{\prime
}\right) ^{-1}-\mathbf{C}^{-1}\right) \right\Vert _{F}\displaystyle%
\left\Vert \frac{1}{m}\sum_{t=1}^{m}\mathbf{x}_{t}\epsilon _{t}\right\Vert }{%
m^{1/2}\left( 1+\displaystyle\frac{k}{m}\right) \min_{1\leq j\leq
J}c_{\alpha ,\eta _{j}}\widetilde{r}_{m}^{1/2-\eta _{j}}d_{a_{m},\eta }%
\displaystyle\left( \frac{k}{m+k}\right) ^{\eta _{j}}}  \notag \\
=&O_{P}\left( 1\right) \max_{1\leq k\leq T_{m}}\frac{%
d^{1/2}kdm^{-1/2}d^{1/2}m^{-1/2}}{m^{1/2}\left( 1+\displaystyle\frac{k}{m}%
\right) \min_{1\leq j\leq J}c_{\alpha ,\eta _{j}}\widetilde{r}_{m}^{1/2-\eta
_{j}}d_{a_{m},\eta }\displaystyle\left( \frac{k}{m+k}\right) ^{\eta _{j}}} 
\notag \\
=&O_{P}\left( 1\right) \frac{d^{2}}{m^{1/2}}\max_{1\leq j\leq J}\frac{1}{%
\widetilde{r}_{m}^{1/2-\eta _{j}}c_{\alpha ,\eta _{j}}}\max_{1\leq k\leq
T_{m}}\left( \frac{k}{m+k}\right) ^{1-\eta _{j}}\frac{1}{d_{a_{m},\eta }} 
\notag \\
=&O_{P}\left( 1\right) \frac{d^{2}}{m^{1/2}}=o_{P}\left( 1\right) ,  \notag
\end{align}%
having used: Lemmas \ref{max-ineq-1} and \ref{max-inequ-3} in the third
line, and the fact that $d=o\left( m^{1/4}\right) $ in the last line.
Similarly, it holds that%
\begin{align}
&\max_{1\leq k\leq T_{m}}\frac{\left\vert \displaystyle\sum_{t=m+1}^{m+k}%
\left( \mathbf{x}_{t}-\mathbf{c}_{1}\right) ^{\prime }\mathbf{C}^{-1}%
\displaystyle\left( \frac{1}{m}\sum_{t=1}^{m}\mathbf{x}_{t}\epsilon
_{t}\right) \right\vert }{m^{1/2}\left( 1+\displaystyle\frac{k}{m}\right)
\min_{1\leq j\leq J}c_{\alpha ,\eta _{j}}\widetilde{r}_{m}^{1/2-\eta
_{j}}d_{a_{m},\eta }\displaystyle\left( \frac{k}{m+k}\right) ^{\eta _{j}}}
\label{veto-b} \\
\leq &\max_{1\leq k\leq T_{m}}\frac{\left\Vert \displaystyle%
\sum_{t=m+1}^{m+k}\left( \mathbf{x}_{t}-\mathbf{c}_{1}\right) \right\Vert
\left\Vert \mathbf{C}^{-1}\right\Vert _{F}\displaystyle\left\Vert \frac{1}{m}%
\sum_{t=1}^{m}\mathbf{x}_{t}\epsilon _{t}\right\Vert }{m^{1/2}\left( 1+%
\displaystyle\frac{k}{m}\right) \min_{1\leq j\leq J}c_{\alpha ,\eta _{j}}%
\widetilde{r}_{m}^{1/2-\eta _{j}}d_{a_{m},\eta }\displaystyle\left( \frac{k}{%
m+k}\right) ^{\eta _{j}}}  \notag \\
=&O_{P}\left( 1\right) \max_{1\leq k\leq T_{m}}\frac{k^{\zeta
_{3}}d^{1/2}d^{1/2}d^{1/2}m^{-1/2}}{m^{1/2}\left( 1+\displaystyle\frac{k}{m}%
\right) \min_{1\leq j\leq J}c_{\alpha ,\eta _{j}}\widetilde{r}_{m}^{1/2-\eta
_{j}}d_{a_{m},\eta }\displaystyle\left( \frac{k}{m+k}\right) ^{\eta _{j}}} 
\notag \\
=&O_{P}\left( 1\right) d^{3/2}\max_{1\leq j\leq J}\frac{1}{\widetilde{r}%
_{m}^{1/2-\eta _{j}}c_{\alpha ,\eta _{j}}}\max_{1\leq k\leq T_{m}}\left( 
\frac{k}{m+k}\right) ^{\zeta _{3}-\eta _{j}}\left( \frac{1}{m+k}\right)
^{1-\zeta _{3}}\frac{1}{d_{a_{m},\eta }}  \notag \\
=&O_{P}\left( 1\right) d^{3/2}\left( \frac{1}{m}\right) ^{1-\zeta _{3}}%
\widetilde{r}_{m}^{\zeta _{3}-1/2}=o_{P}\left( 1\right) ,  \notag
\end{align}%
having used: Lemma \ref{max-ineq-2} with $\zeta _{3}=1/2+\varepsilon $ for
all $\eta _{j}<1/2$\ and $1/2<\zeta _{3}<\min \left\{ 1,\eta _{j}\right\} $
for all $\eta _{j}>1/2$, Lemma \ref{max-ineq-1} and (\ref{c-frob}) in the
third line, and the fact that $d=o\left( m^{1/4}\right) $ in the last line.
Combining (\ref{veto-a}) and (\ref{veto-b}), and recalling that $\mathbf{c}%
_{1}^{\prime }\mathbf{C}^{-1}\mathbf{x}_{t}=1$, we get%
\begin{equation*}
\max_{1\leq k\leq T_{m}}\frac{\left\vert \left( \displaystyle%
\sum_{t=m+1}^{m+k}\mathbf{x}_{t}^{\prime }\right) \left( \displaystyle%
\sum_{t=1}^{m}\mathbf{x}_{t}\mathbf{x}_{t}^{\prime }\right) ^{-1}\left( %
\displaystyle\sum_{t=1}^{m}\mathbf{x}_{t}\epsilon _{t}\right) -\displaystyle%
\frac{k}{m}\displaystyle\sum_{t=1}^{m}\epsilon _{t}\right\vert }{%
m^{1/2}\left( 1+\displaystyle\frac{k}{m}\right) \min_{1\leq j\leq
J}c_{\alpha ,\eta _{j}}\widetilde{r}_{m}^{1/2-\eta _{j}}d_{a_{m},\eta }%
\displaystyle\left( \frac{k}{m+k}\right) ^{\eta _{j}}}=o_{P}\left( 1\right) ,
\end{equation*}%
whence%
\begin{align}
&\max_{1\leq k\leq T_{m}}\frac{\left\vert \displaystyle\sum_{t=m+1}^{m+k}%
\widehat{\epsilon }_{t}\right\vert }{m^{1/2}\left( 1+\displaystyle\frac{k}{m}%
\right) \min_{1\leq j\leq J}c_{\alpha ,\eta _{j}}\widetilde{r}_{m}^{1/2-\eta
_{j}}d_{a_{m},\eta }\displaystyle\left( \frac{k}{m+k}\right) ^{\eta _{j}}}
\label{veto-c} \\
=&\max_{1\leq k\leq T_{m}}\frac{\left\vert \displaystyle\sum_{t=m+1}^{m+k}%
\epsilon _{t}-\displaystyle\frac{k}{m}\displaystyle\sum_{t=1}^{m}\epsilon
_{t}\right\vert }{m^{1/2}\left( 1+\displaystyle\frac{k}{m}\right)
\min_{1\leq j\leq J}c_{\alpha ,\eta _{j}}\widetilde{r}_{m}^{1/2-\eta
_{j}}d_{a_{m},\eta }\displaystyle\left( \frac{k}{m+k}\right) ^{\eta _{j}}}%
+o_{P}\left( 1\right) .  \notag
\end{align}%
Also, by standard algebra based on Lemma \ref{sip}%
\begin{align*}
&\max_{1\leq k\leq T_{m}}\frac{\left\vert \displaystyle\sum_{t=m+1}^{m+k}%
\epsilon _{t}-W_{1,m}\left( k\right) \right\vert }{m^{1/2}\left( 1+%
\displaystyle\frac{k}{m}\right) \min_{1\leq j\leq J}c_{\alpha ,\eta _{j}}%
\widetilde{r}_{m}^{1/2-\eta _{j}}d_{a_{m},\eta }\displaystyle\left( \frac{k}{%
m+k}\right) ^{\eta _{j}}} \\
=&O_{P}\left( 1\right) \max_{1\leq j\leq J}\max_{1\leq k\leq T_{m}}\frac{%
k^{\zeta _{1}}}{m^{1/2}\left( 1+\displaystyle\frac{k}{m}\right) \widetilde{r}%
_{m}^{1/2-\eta _{j}}d_{a_{m},\eta }\displaystyle\left( \frac{k}{m+k}\right)
^{\eta _{j}}} \\
=&O_{P}\left( 1\right) \max_{1\leq j\leq J}\frac{1}{\widetilde{r}%
_{m}^{1/2-\eta _{j}}d_{a_{m},\eta }}\max_{1\leq k\leq T_{m}}\frac{k^{\zeta
_{1}-\eta _{j}}}{\left( m+k\right) ^{1/2-\eta _{j}}} \\
=&O_{P}\left( 1\right) a_{m}^{\zeta _{1}-1/2}=o_{P}\left( 1\right) ,
\end{align*}%
and%
\begin{align*}
&\max_{1\leq k\leq T_{m}}\frac{\left\vert \displaystyle\frac{k}{m}%
\displaystyle\sum_{t=1}^{m}\epsilon _{t}-\displaystyle\frac{k}{m}%
W_{2,m}\left( m\right) \right\vert }{m^{1/2}\left( 1+\displaystyle\frac{k}{m}%
\right) \min_{1\leq j\leq J}c_{\alpha ,\eta _{j}}\widetilde{r}_{m}^{1/2-\eta
_{j}}d_{a_{m},\eta }\displaystyle\left( \frac{k}{m+k}\right) ^{\eta _{j}}} \\
=&O_{P}\left( 1\right) \max_{1\leq k\leq T_{m}}\frac{\frac{k}{m}m^{\zeta
_{2}}}{m^{1/2}\left( 1+\displaystyle\frac{k}{m}\right) \min_{1\leq j\leq
J}c_{\alpha ,\eta _{j}}\widetilde{r}_{m}^{1/2-\eta _{j}}d_{a_{m},\eta }%
\displaystyle\left( \frac{k}{m+k}\right) ^{\eta _{j}}} \\
=&m^{\zeta _{2}-1/2}O_{P}\left( 1\right) \max_{1\leq j\leq J}\frac{1}{%
\widetilde{r}_{m}^{1/2-\eta _{j}}d_{a_{m},\eta }}\max_{1\leq k\leq
T_{m}}\left( \frac{k}{m+k}\right) ^{1-\eta _{j}}=m^{\zeta
_{2}-1/2}O_{P}\left( 1\right) =o_{P}\left( 1\right) .
\end{align*}%
Thence (\ref{veto-1}) follows. We now study%
\begin{align*}
&\max_{1\leq k\leq T_{m}}\frac{\left\vert W_{1,m}\left( k\right) -%
\displaystyle\frac{k}{m}W_{2,m}\left( m\right) \right\vert }{m^{1/2}\left( 1+%
\displaystyle\frac{k}{m}\right) \min_{1\leq j\leq J}c_{\alpha ,\eta _{j}}%
\widetilde{r}_{m}^{1/2-\eta _{j}}d_{a_{m},\eta }\displaystyle\left( \frac{k}{%
m+k}\right) ^{\eta _{j}}} \\
\overset{\mathcal{D}}{=}&\max_{1\leq k\leq T_{m}}\frac{\left\vert
W_{1}\left( k\right) -\displaystyle\frac{k}{m}W_{2}\left( m\right)
\right\vert }{m^{1/2}\left( 1+\displaystyle\frac{k}{m}\right) \min_{1\leq
j\leq J}c_{\alpha ,\eta _{j}}\widetilde{r}_{m}^{1/2-\eta _{j}}d_{a_{m},\eta }%
\displaystyle\left( \frac{k}{m+k}\right) ^{\eta _{j}}} \\
\overset{\mathcal{D}}{=}&\max_{1\leq k\leq T_{m}}\frac{\left( 1+\displaystyle%
\frac{k}{m}\right) \left\vert W\displaystyle\left( \frac{k/m}{1+k/m}\right)
\right\vert }{\left( 1+\displaystyle\frac{k}{m}\right) \min_{1\leq j\leq
J}c_{\alpha ,\eta _{j}}\widetilde{r}_{m}^{1/2-\eta _{j}}d_{a_{m},\eta }%
\displaystyle\left( \frac{k}{m+k}\right) ^{\eta _{j}}} \\
\overset{\mathcal{D}}{=}&\max_{1\leq k\leq T_{m}}\frac{\left\vert W%
\displaystyle\left( \frac{k/m}{1+k/m}\right) \right\vert }{\min_{1\leq j\leq
J}c_{\alpha ,\eta _{j}}\widetilde{r}_{m}^{1/2-\eta _{j}}d_{a_{m},\eta }%
\displaystyle\left( \frac{k}{m+k}\right) ^{\eta _{j}}} \\
\overset{\mathcal{D}}{=}&\max_{1/m\leq t\leq T_{m}/m}\frac{\left\vert W%
\displaystyle\left( \frac{t}{1+t}\right) \right\vert }{\min_{1\leq j\leq
J}c_{\alpha ,\eta _{j}}\widetilde{r}_{m}^{1/2-\eta _{j}}d_{a_{m},\eta }%
\displaystyle\left( \frac{t}{1+t}\right) ^{\eta _{j}}},
\end{align*}%
seeing as the distributions of $W_{1,m}\left( k\right) $ and $W_{2,m}\left(
k\right) $ do not depend on $m$ (second line), and by the same token as in (%
\ref{process}), with $W\left( \cdot \right) $ standard Wiener. Therefore,
letting $J^{\ast }$ be the set of indices $j$ for which $\eta _{j}>1/2$ and
repeatedly using the scale transform of Wiener process%
\begin{align*}
&\max_{1/m\leq t\leq T_{m}/m}\frac{\left\vert W\displaystyle\left( \frac{t}{%
1+t}\right) \right\vert }{\min_{1\leq j\leq J}c_{\alpha ,\eta _{j}}%
\widetilde{r}_{m}^{1/2-\eta _{j}}d_{a_{m},\eta }\displaystyle\left( \frac{t}{%
1+t}\right) ^{\eta _{j}}} \\
\overset{\mathcal{D}}{=}&\max \left\{ \max_{j\notin J^{\ast }}\max_{1/m\leq
t\leq T_{m}/m}\frac{\left\vert W\displaystyle\left( \frac{t}{1+t}\right)
\right\vert }{c_{\alpha ,\eta _{j}}\displaystyle\left( \frac{t}{1+t}\right)
^{\eta _{j}}},\max_{j\in J^{\ast }}r_{m}^{\eta _{j}-1/2}\max_{a_{m}/m\leq
t\leq T_{m}/m}\frac{\left\vert W\displaystyle\left( \frac{t}{1+t}\right)
\right\vert }{c_{\alpha ,\eta _{j}}\displaystyle\left( \frac{t}{1+t}\right)
^{\eta _{j}}}\right\} \\
\overset{\mathcal{D}}{=}&\max \left\{ \max_{j\notin J^{\ast }}\max_{1/\left(
m+1\right) \leq u\leq T_{m}/\left( m+T_{m}\right) }\frac{\left\vert W\left(
u\right) \right\vert }{c_{\alpha ,\eta _{j}}u^{\eta _{j}}},\max_{j\in
J^{\ast }}r_{m}^{\eta _{j}-1/2}\max_{a_{m}/\left( m+a_{m}\right) \leq u\leq
T_{m}/\left( m+T_{m}\right) }\frac{\left\vert W\left( u\right) \right\vert }{%
c_{\alpha ,\eta _{j}}u^{\eta _{j}}}\right\} \\
\overset{\mathcal{D}}{=}&\max \left\{ \max_{j\notin J^{\ast }}\max_{1/\left(
m+1\right) \leq u\leq T_{m}/\left( m+T_{m}\right) }\frac{\left\vert W\left(
u\right) \right\vert }{c_{\alpha ,\eta _{j}}u^{\eta _{j}}},\max_{j\in
J^{\ast }}\max_{1\leq s\leq \left( T_{m}\left( m+a_{m}\right) \right)
/\left( a_{m}\left( m+T_{m}\right) \right) }\frac{\left\vert W\left(
s\right) \right\vert }{c_{\alpha ,\eta _{j}}s^{\eta _{j}}}\right\} \\
\overset{\mathcal{D}}{=}&\max \left\{ \max_{j\notin J^{\ast }}\max_{1/\left(
m+1\right) \leq u\leq T_{m}/\left( m+T_{m}\right) }\frac{\left\vert W\left(
u\right) \right\vert }{c_{\alpha ,\eta _{j}}u^{\eta _{j}}},\max_{j\in
J^{\ast }}\max_{\left( a_{m}\left( m+T_{m}\right) \right) /\left(
T_{m}\left( m+a_{m}\right) \right) \leq u\leq 1}\frac{\left\vert W\left(
u\right) \right\vert }{c_{\alpha ,\eta _{j}}u^{1-\eta _{j}}}\right\} \\
\overset{a.s}{\rightarrow }&\max \left\{ \max_{j\notin J^{\ast }}\sup_{0<u<1}%
\frac{\left\vert W\left( u\right) \right\vert }{c_{\alpha ,\eta _{j}}u^{\eta
_{j}}},\max_{j\in J^{\ast }}\sup_{0<u<1}\frac{\left\vert W\left( u\right)
\right\vert }{c_{\alpha ,\eta _{j}}u^{1-\eta _{j}}}\right\} \overset{%
\mathcal{D}}{=}\sup_{0<u<1}\frac{\left\vert W\left( u\right) \right\vert }{%
\min_{1\leq j\leq J}c_{\alpha ,\eta _{j}}u^{\widetilde{\eta }_{j}}},
\end{align*}%
whence the desired result.
\end{proof}

\begin{proof}[Proof of Proposition \protect\ref{veto-power}]
We know from the proof of Theorem \ref{power} that the presence of power
depends only on the non-centrality induced by the break, so - following the
logic in that proof - we need to study%
\begin{equation}
\frac{\left( \widetilde{k}-k^{\ast }\right) \left\vert \mathbf{c}%
_{1}^{\prime }\Delta _{m}\right\vert }{m^{1/2}\left( 1+\displaystyle\frac{%
\widetilde{k}}{m}\right) \min_{1\leq j\leq J}\widetilde{r}_{m}^{1/2-\eta
_{j}}\displaystyle\left( \frac{\widetilde{k}}{m+\widetilde{k}}\right) ^{\eta
_{j}}},  \label{noncentral}
\end{equation}%
where $\widetilde{k}$ is specified later. Under condition \textit{(i)}, let $%
\widetilde{k}=a_{m}+k^{\ast }$; note that whenever $j\in J^{\ast }$, (\ref%
{noncentral}) is of the same order as%
\begin{equation*}
\frac{a_{m}\left\vert \mathbf{c}_{1}^{\prime }\Delta _{m}\right\vert }{%
m^{1/2}\min_{j\in J^{\ast }}\displaystyle\left( \frac{a_{m}}{m}\right)
^{\eta _{j}}}=\frac{a_{m}^{1-\max_{j\in J^{\ast }}\eta _{j}}\left\vert 
\mathbf{c}_{1}^{\prime }\Delta _{m}\right\vert }{m^{1/2-\max_{j\in J^{\ast
}}\eta _{j}}},
\end{equation*}%
and seeing as%
\begin{equation*}
\lim_{m\rightarrow \infty }\frac{a_{m}^{1-\max_{j\in J^{\ast }}\eta
_{j}}\left\vert \mathbf{c}_{1}^{\prime }\Delta _{m}\right\vert }{%
m^{1/2-\max_{j\in J^{\ast }}\eta _{j}}}\times \frac{1}{a_{m}^{1/2}\left\vert 
\mathbf{c}_{1}^{\prime }\Delta _{m}\right\vert }=\infty ,
\end{equation*}%
on account of (\ref{a-m}), it follows that under condition \textit{(i)}%
\begin{equation*}
\frac{\left( \widetilde{k}-k^{\ast }\right) \left\vert \mathbf{c}%
_{1}^{\prime }\Delta _{m}\right\vert }{m^{1/2}\left( 1+\displaystyle\frac{%
\widetilde{k}}{m}\right) \min_{1\leq j\leq J}\widetilde{r}_{m}^{1/2-\eta
_{j}}\displaystyle\left( \frac{\widetilde{k}}{m+\widetilde{k}}\right) ^{\eta
_{j}}}\geq c_{0}a_{m}^{1/2}\left\vert \mathbf{c}_{1}^{\prime }\Delta
_{m}\right\vert ,
\end{equation*}%
for some $c_{0}>0$, whence (\ref{regime-1}) yields the final result. Under
conditions \textit{(ii)} and \textit{(iii)}, let $\widetilde{k}=\left\lfloor
\left( 1+c_{0}\right) k^{\ast }\right\rfloor $; by standard algebra, it is
easy to see that (\ref{noncentral}) is proportional to%
\begin{equation*}
\frac{k^{\ast }\left\vert \mathbf{c}_{1}^{\prime }\Delta _{m}\right\vert }{%
m^{1/2}\left( 1+\displaystyle\frac{k^{\ast }}{m}\right) \min_{1\leq j\leq J}%
\widetilde{r}_{m}^{1/2-\eta _{j}}\displaystyle\left( \frac{k^{\ast }}{%
m+k^{\ast }}\right) ^{\eta _{j}}},
\end{equation*}%
whence (\ref{veto-diverge}) yields the desired result.
\end{proof}

\end{document}